\documentclass[11pt]{article}

\usepackage{geometry, array, xcolor}             
\usepackage{amsmath,amsfonts,graphicx, amssymb, placeins}
\usepackage{algorithm}
\usepackage[noend]{algpseudocode}
\usepackage{enumerate, url}
\usepackage{xargs}
\usepackage{cite}
\usepackage{subfig}
\usepackage{enumitem}
\usepackage[colorinlistoftodos]{todonotes}
\usepackage{dsfont}
\usepackage{soul}
\usepackage{xcolor}
\usepackage{footnote}
\makesavenoteenv{tabular}
\makesavenoteenv{table}
\usepackage{times}
\usepackage{titlesec}
\titlespacing*{\section}{6pt}{6pt}{6pt}
\titlespacing*{\subsection}{4pt}{4pt}{4pt}

\makeatletter
\def\mathcolor#1#{\@mathcolor{#1}}
\def\@mathcolor#1#2#3{%
  \protect\leavevmode
  \begingroup
    \color#1{#2}#3%
  \endgroup
}
\makeatother

\makeatletter
\g@addto@macro\normalsize{%
  \setlength\abovedisplayskip{3pt}
  \setlength\belowdisplayskip{3pt}
  \setlength\abovedisplayshortskip{3pt}
  \setlength\belowdisplayshortskip{3pt}
}
\makeatother

\newcommand{\wratio}{R_w}

\DeclareGraphicsExtensions{.pdf,.png,.jpg}

\usepackage{amsthm}

\makeatletter
\def\thm@space@setup{\thm@preskip=0pt
\thm@postskip=3pt}
\makeatother

\newtheorem{theorem}{Theorem}[section]
\newtheorem{lemma}[theorem]{Lemma}
\newtheorem{proposition}[theorem]{Proposition}

\newtheorem{fact}[theorem]{Fact}
\newtheorem{corollary}[theorem]{Corollary}
\newtheorem{definition}[theorem]{Definition}

\newenvironment{proof-sketch}{{\bf Proof Sketch:}}{\hfill\rule{2mm}{2mm}}
\makeatletter
\newenvironment{proofof}[1]{\par
  \pushQED{\qed}%
  \normalfont \topsep3\p@\relax
  \trivlist
  \item[\hskip\labelsep
        \bfseries
    Proof of #1\@addpunct{.}]\ignorespaces
}{%
  \popQED\endtrivlist\@endpefalse
}
\makeatother

\makeatletter
\renewcommand{\paragraph}{%
  \@startsection{paragraph}{4}%
  {\z@}{1ex \@plus 1ex \@minus .2ex}{-.5em}%
  {\normalfont\normalsize\bfseries}%
}
\makeatother

\newif\iffullpaper
\fullpapertrue 
\title{{A Fast Distributed Stateless Algorithm for $\alpha$-Fair Packing Problems}}
\author{
Jelena Mara\v{s}evi\'{c}\thanks{Supported in part by the NSF grant CNS-10-54856 and a Qualcomm Innovation Fellowship.}\\ 
Columbia University\\
{\tt jelena@ee.columbia.edu}\\
\and
Cliff Stein \thanks{Supported in part by the NSF grants CCF-1349602 and CCF-1421161.}\\
{Columbia University}\\
{\tt cliff@ieor.columbia.edu}\\
\and
Gil Zussman\thanks{Supported in part by the NSF grant CNS-10-54856 and the People Programme (Marie Curie Actions) of the European Union's Seventh Framework Programme (FP7/2007-2013) under REA grant agreement n${^{\text{o}}} $[PIIF-GA-2013-629740].11.}\\
       {Columbia University}\\
       {\tt gil@ee.columbia.edu}
} 
\date{}

\newcommand{\littlesum}{\mathop{\textstyle\sum}}

\newcommand{\red}[1]{\textcolor{red}{#1}}

\iffullpaper
\allowdisplaybreaks[1]
\fi

\begin{document}
\maketitle
\begin{abstract}
Over the past two decades, fair resource allocation problems have received considerable attention in a variety of application areas. %A commonly used notion of fairness that generalizes the well known proportional fairness and max-min fairness is $\alpha-$fairness. 
However, \emph{little progress has been made in the design of distributed algorithms with convergence guarantees for general and commonly used $\alpha$-fair allocations}.  In this paper, we study weighted $\alpha$-fair packing problems, that is, the problems of maximizing the objective functions (i) $\sum_j w_j x_j^{1-\alpha}/(1-\alpha)$ when $\alpha > 0$, $\alpha \neq 1$ and (ii) $\sum_j w_j \ln x_j$ when $\alpha = 1$, over linear constraints $Ax \leq b$, $x\geq 0$, where $w_j$ are positive weights and $A$ and $b$ are non-negative. We consider the distributed computation model that was used for packing linear programs and network utility maximization problems. Under this model, \emph{we provide a distributed algorithm for general $\alpha$} that converges to an $\varepsilon-$approximate solution in time (number of distributed iterations) that has an inverse polynomial dependence on the approximation parameter $\varepsilon$ and poly-logarithmic dependence on the problem size. \emph{This is the first distributed algorithm for weighted $\alpha-$fair packing with poly-logarithmic convergence in the input size.} The algorithm uses simple local update rules and is stateless (namely,  it allows asynchronous updates, is self-stabilizing, and allows incremental and local adjustments). We also obtain a number of structural results that characterize $\alpha-$fair allocations as the value of $\alpha$ is varied. These results deepen our understanding of fairness guarantees in $\alpha-$fair packing allocations, and also provide insight into the behavior of $\alpha-$fair allocations in the asymptotic cases $\alpha\rightarrow 0$, $\alpha \rightarrow 1$, and $\alpha \rightarrow \infty$.  
\end{abstract}

\thispagestyle{empty}
\newpage

%%%%%%%%%%%%%%%%%%%%%%%%%%%%%%%%%%%%%%%%%%%%%%%%%%%%%%%%%%%%%%%%%%%%%%%%%%%%%%%%%%%%%%%%%%%%
%%%%%%%%%%%%%%%%%%%%%%%%%%%%%%%%%%%%%%%%%%%%%%%%%%%%%%%%%%%%%%%%%%%%%%%%%%%%%%%%%%%%%%%%%%%%

\section{Introduction}\label{section:intro}
\setcounter{page}{1}
\pagestyle{plain}

Over the past two decades, \emph{fair resource allocation} problems have received considerable attention in many  application areas, including  Internet congestion control \cite{low2002internet}, rate control in software defined networks \cite{mccormick2014real}, scheduling in wireless networks \cite{yi2008stochastic}, {multi-resource allocation and scheduling in datacenters} \cite{bonald2015multi, ghodsi2011dominant, joe2013multiresource, Im2014competitive}, and a variety of applications in operations research,  economics, and game theory
\cite{bertsimas2012efficiency, jain2007eisenberg}. In most of these applications, positive linear (packing) constraints arise as a natural model of the allowable allocations. 

In this paper, we focus on the problem of finding an {\em $\alpha$-fair} vector on the set determined by packing constraints $Ax\leq \mathds{1}, x\geq 0$ where all $A_{ij} \geq 0$.\footnote{Although in the network congestion control literature the constraint matrix $A$ is commonly assumed to be a 0-1 matrix \cite{kelly1998rate, kelly2014stochastic, yi2008stochastic, paganini2005congestion, MoWalrand2000, low2002internet}, important applications (such as, e.g., multi-resource allocation in datacenters) are modeled by a more general constraint matrix $A$ with arbitrary non-negative elements \cite{bonald2015multi, ghodsi2011dominant, joe2013multiresource, Im2014competitive}.} We refer to this problem as $\alpha-$fair packing. 
For a vector of positive weights $w$ and $\alpha\geq 0$, an allocation vector $x^*$ of size $n$ is weighted $\alpha$-fair%(also referred to as $(w, \alpha)-$proportionally fair or $(w, \alpha)-$fair)
, if for any alternative feasible vector $x$: $\sum_j w_j \frac{x_j-x_j^*}{(x_j^*)^\alpha}\leq 0$ \cite{MoWalrand2000}. For a compact and convex feasible region, $x^*$ can be equivalently defined as a vector that solves the problem of maximizing $p_{\alpha}(x) = \sum_j w_j f_{\alpha}(x_j)$ \cite{MoWalrand2000}, where: 
\begin{equation}
f_{\alpha}(x_j) = 
\begin{cases} \ln(x_j), & \mbox{if } \alpha=1 \\ \frac{x_j^{1-\alpha}}{1-\alpha}, & \mbox{if } \alpha\neq 1 \end{cases} \label{eq:f-alpha}.
\end{equation}
$\alpha$-fairness provides a trade-off between efficiency (sum of allocated resources) and fairness (minimum allocated resource) as a function of $\alpha$: the higher the $\alpha$, the better the fairness guarantees and the lower the efficiency \cite{atkinson1970measurement,bertsimas2012efficiency,lan2010axiomatic}.  Important special cases are 
proportional fairness ($\alpha = 1$) and max-min fairness ($\alpha \rightarrow \infty$).  When $\alpha = 0$, we have the ``unfair" case of  linear optimization.
% Three special cases of $\alpha$-fair resource allocations have been widely studied in the literature: (i) $\alpha = 1$ -- known as proportional fairness \cite{kelly1998rate}, (ii) $\alpha = 0$ -- a utilitarian resource allocation, i.e., $f_\alpha(x_j)$ is linear in $x_j$, and (iii) $\alpha \rightarrow \infty$ that converges to the max-min fair solution \cite{MoWalrand2000}, the most egalitarian resource allocation. %In the bargaining theory, $\alpha=1$ case can be interpreted as a Nash solution \cite{nash1950bargaining} and $\alpha\rightarrow\infty$ case corresponds to the Kalai-Smorodinsky solution \cite{kalai1975other} (see \cite{bertsimas2012efficiency} for a discussion).
% In general, 

Distributed algorithms for $\alpha-$fair packing are of particular interest, as many applications are inherently distributed (such as, e.g., network congestion control), while in others parallelization is highly desirable due to the large problem size (as in, e.g., resource allocation in datacenters). We adopt the model of distributed computation commonly used in the design of 
packing linear programming (LP) algorithms \cite{AwerbuchKhandekar2009, d-allen2014using, d-bartal1997global, d-kuhn2006price, d-luby1993parallel, d-papadimitriou1993linear} and which generalizes the model from network congestion control \cite{kelly2014stochastic}.  In this model, an agent $j$ controls the variable $x_j$ and has information about: (i) the $j^\text{th}$ column of the $m\times n$ constraint matrix $A$, (ii) the weight $w_j$, (iii) upper bounds on the global problem parameters $m, n, w_{\max}$, and $A_{\max}$, where $w_{\max} = \max_j w_j$, and $A_{\max} = \max_{ij} A_{ij}$, and (iv) in each round, the relative slack of each constraint $i$ in which $x_j$ takes part. %(See Section \ref{section:prelims} for more details). 

Distributed algorithms for $\alpha-$fair resource allocations have been most widely studied in the network congestion control literature, using a control-theoretic approach \cite{kelly1998rate, kelly2014stochastic, yi2008stochastic, paganini2005congestion, MoWalrand2000, low2002internet}. Such an approach yields continuous-time algorithms that converge after ``finite'' time; however, the convergence time of these algorithms as a function of the input size is poorly understood. % and it is unclear whether it is even polynomial.   
%Many of the applications mentioned above are packing problems, that is, we are optimizing $p_{\alpha}(x)$ over the constraints $Ax\leq \mathds{1}, x\geq 0$ where all $A_{ij} \geq 0$. 
Some  other distributed pseudo-polynomial-time approximation algorithms that can address $\alpha$-fair packing are described in Table~\ref{table:prev}. These algorithms all have convergence times that are at least linear in the parameters describing the problem.  

No previous work has given truly fast (poly-log iterations) distributed algorithms for the general case of $\alpha$-fair packing.  Only for the unfair $\alpha=0$ case (packing LPs), are such algorithms  known \cite{AwerbuchKhandekar2009, d-luby1993parallel, d-bartal1997global, dc-young2001sequential, d-kuhn2006price, d-allen2014using}.

\paragraph{Our Results.} 
\emph{We provide the first  efficient, distributed, {and stateless} algorithm for %\todo{Why are we saying uncoordinated but not stateless}  
weighted $\alpha$-fair packing}, namely, for the problem %of efficiently solving 
$\max\{p_\alpha(x): Ax\leq \mathds{1}, x\geq 0\}$,  % in a distributed manner, 
{where distributed agents update the values of $x_j$'s asynchronously and react only to the current state of the constraints}. We assume that all non-zero entries $A_{ij}$ of matrix $A$ satisfy $A_{ij}\geq 1$. Considering such a normalized form of the problem is without loss of generality (see Appendix \iffullpaper \ref{appendix:scaling}\else {A} in the full version of the paper\fi).  

The approximation provided by the algorithm, to which we refer as the  $\varepsilon$-approximation, is (i) $(1+\varepsilon)$-multiplicative for $\alpha \neq 1$, and (ii) $W\varepsilon$-additive\footnote{Note that $W$ cannot be avoided here, as additive approximation is not invariant to the scaling of the objective.} for $\alpha = 1$, where $W = \sum_j w_j$. The main results are summarized in the following theorem, where, to unify the statement of the results, we treat $\alpha$ as a constant that is either equal to 1 or bounded away from 0 and 1, and we also loosen the bound in terms of $\varepsilon^{-1}, n, m, \wratio = \max_{j, k} {w_j}/{w_k},$ and $A_{\max}$. %for technical reasons we assume that $\alpha$ is bounded away from 0 and bounded above by a constant, and we use $\wratio$ to denote ${w_{\max}}/{w_{\min}}$.  
For a more detailed statement, see Theorems \ref{thm:convergence-alpha<1} -- \ref{thm:convergence-alpha>1}. 
%A more detailed statement of the results appears in Theorems \ref{thm:convergence-alpha<1}, \ref{thm:convergence-alpha=1}, and \ref{thm:convergence-alpha>1}. 
\begin{theorem}
(Main Result%\footnote{For the purpose of presentation, convergence time bounds provided here are looser than the actual bounds that we obtain. For tighter bounds, see Section \ref{section:main-results}.}
) For a given weighted $\alpha$-fair packing problem $\max\{\sum_j w_j f_\alpha(x_j): Ax\leq \mathds{1}, x\geq 0\}$, where $f_{\alpha}(x_j)$ is given by (\ref{eq:f-alpha}), there exists a stateless and distributed algorithm (\textsc{$\alpha$-FairPSolver}) that computes an $\varepsilon$-approximate solution in ${O}(\varepsilon^{-5}\ln^4(\wratio nmA_{\max}\varepsilon^{-1}))$ rounds.   
\end{theorem}

\emph{To the best of our knowledge, for any constant approximation parameter $\varepsilon$, our algorithm is the first distributed algorithm for weighted $\alpha$-fair packing problems with a poly-logarithmic convergence time.} 

The algorithm is \emph{stateless} according to the definition given by  Awerbuch and Khandekar \cite{AwerbuchKhandekar2009, awerbuch2007greedy}: it starts from any initial state, the agents update the variables $x_j$ in a cooperative but uncoordinated manner, reacting only to the current state of the constraints that they observe, and without access to a global clock. Statelessness implies various desirable properties of a distributed algorithm, such as: asynchronous updates, self-stabilization, and incremental and local adjustments \cite{AwerbuchKhandekar2009, awerbuch2007greedy}.  

We also obtain the following structural results that characterize $\alpha-$fair packing allocations as a function of the value of 
$\alpha$: 
\begin{itemize}[topsep=5pt, leftmargin=10pt]
\itemsep0em 
\item We derive a lower bound on the minimum coordinate of the $\alpha-$fair packing allocation as a function of $\alpha$ and the problem parameters (Lemma \ref{lemma:lower-bound}). This bound deepens our understanding of how the fairness (a minimum allocated value) changes with $\alpha$. %The bound is also used in one of the algorithm parameters.
\item We prove that for $\alpha \leq \frac{\varepsilon/4}{\ln(nA_{\max}/\varepsilon)}$, $\alpha-$fair packing can be $O(\varepsilon)-$approximated by any $\varepsilon-$approximation packing LP solver (Lemma \ref{lemma:LP-close-to-small-alpha-fair}). This result allows us to focus on the $\alpha > \frac{\varepsilon/4}{\ln(nA_{\max}/\varepsilon)}$ cases.
\item We show that for $|\alpha-1| = O({\varepsilon^2}/{\ln^2(\varepsilon^{-1}\wratio mnA_{\max})})$, $\alpha-$fair allocation is $\varepsilon-$approximated by a $1-$fair allocation returned by our algorithm (Lemmas \ref{lemma:alpha-close-to-1-below} and \ref{lemma:alpha-close-to-1-above}).
\item We show that for $\alpha \geq \ln(\wratio n A_{\max})/\varepsilon$, the $\alpha-$fair packing allocation $x^*$ and the max-min fair allocation $z^*$ are $\varepsilon$-close to each other: $(1-\varepsilon)z^*\leq x^* \leq (1+\varepsilon)z^*$ element-wise. This result is especially interesting as {(i)} max-min fair packing is not a convex problem, but rather a multi-objective problem 
(see, e.g., \cite{kleinberg1999fairness, radunovic2007unified}) {and (ii) the result yields the first convex relaxation of max-min fair allocation problems with a $1\pm \varepsilon$ gap}.
\end{itemize}

We now %give an overview of 
overview some of the main technical details of \textsc{$\alpha$-FairPSolver}.  In doing so, we %also 
point out connections to the two main bodies of previous work, from packing LPs\cite{AwerbuchKhandekar2009} and network congestion control \cite{kelly1998rate}. We also outline the new algorithmic ideas and proofs that were needed to obtain the results.  
%We combine some ideas from these two lines of work with significantly new algorithmic ideas and proofs to obtain our results.

\paragraph{The algorithm and KKT conditions.}
The algorithm  maintains primal and dual feasible solutions and updates each primal variable $x_j$ whenever a Karush-Kuhn-Tucker (KKT) condition ${x_j}^\alpha \sum_{i} y_iA_{ij} = w_j$ is not \emph{approximately} satisfied. {In previous work, relevant update rules include}: \cite{kelly1998rate} (for $\alpha=1$), where the update of each variable $x_j$ is proportional to the difference $w_j - {x_j}^{\alpha}\sum_i y_i A_{ij}$, and \cite{AwerbuchKhandekar2009} (for $\alpha=0$), where each $x_j$ is updated by a multiplicative factor $1\pm \beta$, whenever ${x_j}^\alpha \sum_{i} y_iA_{ij} = w_j$ is not {approximately} satisfied. {For our techniques (addressing a general $\alpha$) such rules do not suffice and we introduce the following modifications}: (i) in the $\alpha<1$ case we use multiplicative updates by factors $(1+\beta_1)$ and $(1-\beta_2)$,  where $\beta_1 \neq \beta_2$ and (ii) we use additional threshold values $\delta_j$ to make sure that $x_j$'s do not become too small.  These thresholds guarantee that we maintain a feasible solution, but they significantly complicate  (compared to the linear case) the argument that each step makes a significant progress.

\paragraph{Dual Variables.} 
In \textsc{$\alpha$-FairPSolver}, a dual variable $y_i$ %associated with constraint $i$ 
is an exponential function of the $i^\text{th}$ constraint's relative slack: $y_i(x) = C\cdot e^{\kappa(\sum_jA_{ij}x_j - 1)}$, where $C$ and $\kappa$ are functions of global input parameters $\alpha, w_{\max}, n, m,$ and $A_{\max}$. Packing LP algorithms \cite{AwerbuchKhandekar2009, d-allen2014using, c-plotkin1995fast, d-bartal1997global, c-garg2007faster, c-fleischer2000approximating, c-koufogiannakis2007beating} use similar dual variables {with $C=1$}. Our work requires choosing $C$ to be a function of $\alpha, w_{\max}, n, m,A_{\max}$ rather than a constant.

\paragraph{Convergence Argument.}
The convergence analysis of \textsc{$\alpha$-FairPSolver} relies on the appropriately chosen concave potential function that is bounded below and above for $x_j\in[\delta_j, 1]$, $\forall j$, and that increases with every primal update. The algorithm can also be interpreted as a gradient ascent on a regularized objective function (the potential function), using a generalized entropy regularizer (see \cite{d-allen2014using, c-allen2015nearly}). A  similar potential function was used in many works on packing and covering linear programs, such as, e.g., in \cite{AwerbuchKhandekar2009} and (implicitly) in \cite{dc-young2001sequential}. The Lyapunov function from \cite{kelly1998rate} is also equivalent to this potential function when $y_i(x) = C\cdot e^{\kappa(\sum_jA_{ij}x_j - 1)}$, $\forall i$.  As in these works, the 
main idea in the analysis is to show that whenever a solution $x$ is not ``close" to the optimal one, the potential function increases substantially.
However, our work requires several new ideas in the convergence proofs, the most notable being 
{\em stationary rounds}. A stationary round is roughly a time when the variables $x_j$ do not change much and are close to the optimum. Poly-logarithmic convergence time is then obtained by showing that: (i) there is at most a poly-logarithmic number of non-stationary rounds where the potential function increases additively and the increase is ``large enough'', and (ii) in all the remaining non-stationary rounds, the potential function increases multiplicatively. 
% While the remaining proof uses some ideas from \cite{AwerbuchKhandekar2009}, our problem is significantly harder, due to the non-linear objective. 
Our use of stationary rounds is new, as is the use of Lagrangian duality and all the arguments that follow.  

\begin{table}[t]
\begin{center}
\small 
\renewcommand{\arraystretch}{1.1}
  \begin{tabular}{| c | c | c | c |}
    \hline
    \textbf{Paper} & \textbf{Number of Distributed Iterations}\footnote{The convergence times in \cite{cheung2013tatonnement, Beck2014Gradient, mosk2010fully} are not stated only in terms of the input parameters, but also in terms of intermediary parameters that depend on the problem structure. Stated here are our lowest estimates of the worst-case convergence times.%, even though the actual convergence times may be higher.
    } & \textbf{Statelessness} & \textbf{Notes} \\ \hline
     \cite{cheung2013tatonnement} & $\Omega({\varepsilon}^{-1}{nA_{\max}})$ &  Semi-stateless\footnote{A distributed algorithm is semi-stateless, if all the updates depend only on the current state of the constraints, the updates are performed in a cooperative but non-coordinated manner, and \emph{the updates need to be synchronous} \cite{d-allen2014using}.} & Only for $\alpha = 1$\\ \hline
     \cite{Beck2014Gradient} & $\Omega({\varepsilon}^{-1}{mn{A_{\max}}^2})$ & Not stateless &\\
    \hline
    \cite{mosk2010fully} & poly($\varepsilon^{-1}, m, n, A_{\max}$) & Semi-stateless & \\
    \hline
    \red{[this work]} & \mathcolor{red}{$O({\varepsilon^{-5}}{\ln^4(R_w mn A_{\max}/\varepsilon)})$} & \red{Stateless} &\\
    \hline
  \end{tabular}
\end{center}
\caption{Comparison among distributed algorithms for $\alpha-$fair packing.}\vspace{-15pt}
\label{table:prev}
\end{table}

\paragraph{Relationship to Previous Work.}

Very little progress has been made in the design of efficient distributed algorithms for the general class of $\alpha$-fair objectives. 
Classical work on distributed rate control algorithms in the networking literature uses a control-theoretic approach to optimize $\alpha$-fair objectives. While such an approach has been extensively studied and applied to various network settings \cite{kelly1998rate, kelly2014stochastic, yi2008stochastic, paganini2005congestion, MoWalrand2000, low2002internet}, it {has never been proven to have polynomial} convergence time {(and it is unclear whether such a result can be established)}.%: although ``finite'', it is not clear whether the convergence time is even polynomial as a function of the input size.

Since $\alpha$-fair objectives are concave, their optimization over a region determined by linear constraints is solvable in polynomial time in a centralized setting through convex programming (see, e.g., \cite{boyd2009convex, nesterov2004introductory}). Distributed gradient methods for network utility maximization problems, such as e.g., \cite{Beck2014Gradient, mosk2010fully} summarized in Table \ref{table:prev}, can be employed to address the problem of $\alpha$-fair packing. However, the convergence times of these algorithms depend on the dual gradient's Lipschitz constant to produce good approximations. While \cite{Beck2014Gradient, mosk2010fully} provide a better dependence on the accuracy $\varepsilon$ than our work, the dependence on the dual gradient's Lipschitz constant, in general, leads to at least linear convergence time as a function of $n$, $m$, and $A_{\max}$. 

As mentioned before, some special cases have been addressed, particularly for max-min fairness ($\alpha \rightarrow \infty$) and for  packing LPs ($\alpha = 0$). Relevant work on max-min fairness includes \cite{Bertsekas:1987:DN:12517, jaffe1981bottleneck, kumar2000fairness, kleinberg1999fairness, megiddo1974optimal, marasevic2014max,charny1995congestion}, but none of these works have poly-logarithmic convergence time.  
There is a long history of interesting work on packing LPs in both centralized and distributed settings, e.g., \cite{c-allen2015nearly, c-plotkin1995fast, c-koufogiannakis2007beating, c-garg2007faster, AwerbuchKhandekar2009, d-luby1993parallel, d-bartal1997global, dc-young2001sequential, d-kuhn2006price, d-allen2014using, garg2002line}.  Only a few of these works are stateless, including 
 the packing LP algorithm of Awerbuch and Khandekar \cite{AwerbuchKhandekar2009}, flow control algorithm of Garg and Young \cite{garg2002line}, and the algorithm of Awerbuch, Azar, and Khandekar \cite{awerbuch2008fast} for the special case of load balancing in bipartite graphs. Additionally, the packing LP algorithm of Allen-Zhu and Orecchia \cite{d-allen2014using} is ``semi-stateless''; the lacking property to make it stateless is that it requires synchronous updates. 
The $\alpha=1$ case of $\alpha$-fair packing problems is equivalent to the problem of finding an equilibrium allocation in Eisenberg-Gale markets with Leontief utilities (see \cite{cheung2013tatonnement}). Similar to the aforementioned algorithms, the algorithm from \cite{cheung2013tatonnement} converges in time linear in $\varepsilon^{-1}$ but also (at least) linear in the input size (see Table 1).

{In terms of the techniques, closest to our work is the work by Awerbuch and Khandekar \cite{AwerbuchKhandekar2009} and we now highlight the differences compared to this work. \emph{Some preliminaries} of the convergence proof follow closely those from \cite{AwerbuchKhandekar2009}: mainly, Lemmas \ref{lemma:feasibility}, \ref{lemma:approx-comp-slack}, and \ref{lemma:potential-increase} use similar arguments as corresponding lemmas in \cite{AwerbuchKhandekar2009}. Some parts of the lemmas lower-bounding the potential increase in $\alpha<1$, $\alpha = 1$, and $\alpha>1$ cases (Lemmas \ref{lemma:potential-increase-alpha<1}, \ref{lemma:potential-increase-proportional}, and \ref{lemma:potential-increase-alpha>1}) use similar arguments as \cite{AwerbuchKhandekar2009}, however, even those parts require additional results due to the existence of lower thresholds $\delta_j$.}

{The similarity ends here, as the main convergence arguments are different than those used in \cite{AwerbuchKhandekar2009}. In particular, the convergence argument from \cite{AwerbuchKhandekar2009} relying on stationary intervals cannot be applied in the setting of $\alpha-$fair objectives. More details about why this argument cannot be applied and where it fails are provided in Section \ref{section:convergence}. As already mentioned, we rely on the appropriately chosen definition of a stationary round. To show that in a stationary round a solution $x$ is $\varepsilon-$approximate, we use Lagrangian duality and bound the duality gap through an intricate case analysis. We remark that such an argument could not have been used in \cite{AwerbuchKhandekar2009}, since in the packing LP case there is no guarantee that the solution $y$ is dual-feasible.
}

\paragraph{Organization of the Paper.} The rest of the paper is organized as follows. Section \ref{section:prelims} provides the background. Section \ref{section:algorithm} describes the algorithm, and Section \ref{section:convergence} \iffullpaper provides \else illustrates \fi the convergence analysis and structural results. \iffullpaper Section \ref{section:conclusion} concludes the paper. \else We conclude in Section \ref{section:conclusion}. {The omitted technical details can be found in the full version of the paper.}\fi

%%%%%%%%%%%%%%%%%%%%%%%%%%%%%%%%%%%%%%%%%%%%%%%%%%%%%%%%%%%%%%%%%%%%%%%%%%%%%%%%%%%%%%%%%%%%
%%%%%%%%%%%%%%%%%%%%%%%%%%%%%%%%%%%%%%%%%%%%%%%%%%%%%%%%%%%%%%%%%%%%%%%%%%%%%%%%%%%%%%%%%%%%
\section{Preliminaries}\label{section:prelims}

%%%%%%%%%%%%%%%%%%%%%%%%%%%%%%%%%%%%%%%%%%%%%%%%%%%%%%%%%%%%%%%%%%
\paragraph{Weighted $\alpha$-Fair Packing.} %\label{section:alpha-fairness-preliminaries}
Consider the following optimization problem with positive linear (packing) constraints:
$
{(Q_\alpha)}=\max \{p_\alpha(x)\equiv\sum_{j=1}^n w_jf_{\alpha}(x_j): Ax \leq b, x\geq 0\},
$
where $f_{\alpha}(x_j)$ is given by (\ref{eq:f-alpha}), $x=(x_1,...,x_n)$ is the vector of variables, $A$ is an $m\times n$ matrix with non-negative elements, and $b=(b_1,...,b_m)$ is a vector with strictly positive\footnote{If, for some $i$, $b_i=0$, then trivially $x_j = 0$, for all $j$ such that $A_{ij}\neq 0$.} elements. We refer to %the problem 
$(Q_\alpha)$ as the weighted $\alpha$-fair packing. 
The following definition and lemma introduced by Mo and Walrand \cite{MoWalrand2000} characterize weighted $\alpha$-fair allocations. In the rest of the paper, we will use the terms weighted $\alpha$-fair and $\alpha$-fair interchangeably.

\begin{definition}
\emph{\cite{MoWalrand2000}} Let $w=(w_1, ..., w_n)$ be a vector with positive entries and $\alpha>0$. A vector  $x^*=(x_1^*,...,x_n^*)$ is weighted $\alpha$-fair, if it is feasible and for any other feasible vector $x$: 
$
\sum_{j=1}^n w_j\frac{x_j-x_j^*}{{x_j^*}^\alpha}\leq 0.
$
\end{definition}
\begin{lemma}
\emph{\cite{MoWalrand2000}} A vector $x^*$ solves $(Q_\alpha)$ for functions $f_{\alpha}(x_j^*)$ if and only if it is weighted $\alpha$-fair.
\end{lemma}

Notice in $(Q_\alpha)$ that since $b_i > 0$, $\forall i$, and the partial derivative of the objective with respect to any of the variables $x_j$ goes to $\infty$ as $x_j\rightarrow 0$, the optimal solution must lie in the positive orthant. Moreover, since the objective is strictly concave and maximized over a convex region, the optimal solution is unique and $(Q_\alpha)$ satisfies strong duality (see, e.g., \cite{boyd2009convex}). The same observations are true for the scaled version of the problem denoted by $(P_\alpha)$ and introduced in the following subsection.

%%%%%%%%%%%%%%%%%%%%%%%%%%%%%%%%%%%%%%%%%%%%%%%%%%%%%%%%%%%%%%%%%%%%%%%%%%%%%%%%%%%%%%%%%%%%
\paragraph{Normalized Form.}%\label{section:ourproblem} 
We consider the weighted $\alpha$-fair packing problem in the normalized form:
\begin{align*}
(P_\alpha) = {\max} \big\{p_\alpha(x): Ax \leq \mathds{1}, x\geq 0\big\},
\end{align*}
where
$p_\alpha(x) = \sum_{j=1}^n w_jf_{\alpha}(x_j)$, $f_\alpha$ is defined by (\ref{eq:f-alpha}), $w=(w_1,...,w_n)$ is a vector of positive weights, $x=(x_1,...,x_n)$ is the vector of variables, $A$ is an $m\times n$ matrix with non-negative entries, and $\mathds{1}$ is a size-$m$ vector of 1's. We let $A_{\max}$ denote the maximum element of the constraint matrix $A$, and assume that every entry $A_{ij}$ of $A$ is non-negative, and moreover, that $A_{ij}\geq 1$ whenever $A_{ij}\neq 0$. The maximum weight is denoted by $w_{\max}$ and the minimum weight is denoted by $w_{\min}$. The sum of the weights is denoted by $W$ and the ratio $\frac{w_{\max}}{w_{\min}}$ by $\wratio$.  %We use $(P_{\alpha})$ to refer to this formulation.
We remark that considering problem $(Q_\alpha)$ in the normalized form $(P_\alpha)$ is without loss of generality:  any problem $(Q_\alpha)$ can be scaled to this form by (i) dividing both sides of each inequality $i$ by $b_i$ and (ii) working with scaled variables $c\cdot x_j$, where $c = \min\{1,\, \min_{\{i, j: A_{ij\neq 0}\}}\frac{A_{ij}}{b_i}\}$. %Observe that if $c\cdot x$ solves $(P_\alpha)$, then $x$ solves $(Q_\alpha)$.  
Moreover, such scaling preserves the approximation (\iffullpaper \ref{appendix:scaling} \else {see Appendix A in the full version of the paper}\fi). 

%%%%%%%%%%%%%%%%%%%%%%%%%%%%%%%%%%%%%%%%%%%%%%%%%%%%%%%%%%%%%%%%%%%%%%%%%%%%%%%%%%%%%%%%%%%%
\paragraph{KKT Conditions and Duality Gap}\label{section:lower-bound-duality-gap}
%A useful piece of information for understanding the intuition behind \textsc{$\alpha$-FairPSolver} and its analysis are the KKT conditions for $(P_\alpha)$ and the duality gap. 
We will denote the Lagrange multipliers for $(P_\alpha)$ as $y = (y_1,...,y_m)$ and refer to them as ``dual variables''. The KKT conditions for $(P_\alpha)$ are \iffullpaper (see Appendix \ref{appendix:primal-dual-duality-gap})\else {(see Appendix B in the full version)}\fi:
\begin{align}
\smallskip
\littlesum_{j=1}^n A_{ij}x_j\leq 1, \quad\forall i\in\{1,...,m\};\quad x_j\geq 0, \quad\forall j\in \{1,...,n\} \quad &\text{(primal feasibility)}\tag{K1}\label{eq:K1}\\[-3pt]
y_i \geq 0, \quad\forall i\in\{1,...,m\} \quad &\text{(dual feasibility)}\tag{K2}\label{eq:K2}\\[-3pt]
y_i\cdot\Big(\littlesum_{j=1}^m A_{ij} x_j - 1\Big) = 0, \quad\forall i\in\{1,...,m\}\quad &\text{(complementary slackness)}\tag{K3}\label{eq:K3}\\[-4pt]
{x_j}^\alpha\littlesum_{i=1}^m y_i A_{ij} = w_j, \quad \forall j\in \{1,...,m\}\quad &\text{(gradient conditions)}\tag{K4}\label{eq:K4}
\end{align}

The duality gap for $\alpha \neq 1$ is \iffullpaper (see Appendix \ref{appendix:primal-dual-duality-gap})\else {(see Appendix B in the full version of the paper)}\fi:
\begin{align}G_{\alpha}(x, y) = 
\littlesum_{j=1}^n w_j\frac{{x_j}^{1-\alpha}}{1-\alpha}\big({\xi_j}^{\frac{\alpha-1}{\alpha}}-1\big) +\littlesum_{i=1}^m y_i - \littlesum_{j=1}^n  w_j x_j^{1-\alpha}\cdot {\xi_j}^{\frac{\alpha-1}{\alpha}},\label{eq:duality-gap-alpha}
\end{align}  
where $\xi_j = \frac{{x_j}^{\alpha}\sum_{i=1}^m y_i A_{ij}}{w_j}$, while for $\alpha  = 1$:
\begin{align}
G_1(x, y) = - \littlesum_{j=1}^n w_j \ln\Big(\frac{x_j\small{\sum_{i=1}^m} y_i A_{ij}}{w_j}\Big)+\littlesum_{i=1}^m y_i -W.\label{eq:duality-gap-proportional}
\end{align}

%%%%%%%%%%%%%%%%%%%%%%%%%%%%%%%%%%%%%%%%%%%%%%%%%%%%%%%%%%%%%%%%%%%%%%%%%%%%%%%%%%%%%%%%%%%%%%
\paragraph{Model of Distributed Computation} %\label{section:model}
We adopt the same model of distributed computation as \cite{AwerbuchKhandekar2009, d-allen2014using, d-bartal1997global, d-kuhn2006price, d-luby1993parallel, d-papadimitriou1993linear}, described as follows. We assume that for each $j\in\{1,...,n\}$, there is an agent controlling the variable $x_j$. Agent $j$ is assumed to have information about the following problem parameters: (i) the $j^\text{th}$ column of $A$, (ii) the weight $w_j$, and (iii) (an upper bound on) $m, n, w_{\max}$, and $A_{\max}$. In each round, agent $j$ collects the relative slack\footnote{The slack is ``relative'' because in a non-scaled version of the problem where one could have $b_i\neq 1$, agent $j$ would need to have information about $\frac{b_i - \sum_{j=1}^n A_{ij}x_j}{b_i}$.} $1 - \sum_{j=1}^n A_{ij}x_j$ of all constraints $i$ for which $A_{ij} \neq 0$. 

We remark that this model of distributed computation is a generalization of the model considered in network congestion control problems \cite{kelly2014stochastic} where a variable $x_j$ corresponds to the rate of node $j$, $A$ is a 0-1 routing matrix, such that $A_{ij} = 1$ if and only if a node $j$ sends flow over link $i$, and $b$ is the vector of link capacities.  Under this model, the knowledge about the relative slack of each constraint corresponds to each node collecting (a function of) congestion on each link that it utilizes. Such a model was used in network utility maximization problems with $\alpha$-fair objectives \cite{kelly1998rate} and general strongly-concave objectives \cite{Beck2014Gradient}.

%%%%%%%%%%%%%%%%%%%%%%%%%%%%%%%%%%%%%%%%%%%%%%%%%%%%%%%%%%%%%%%%%%%%%%%%%%%%%%%%%%%%%%%%%%%%
%%%%%%%%%%%%%%%%%%%%%%%%%%%%%%%%%%%%%%%%%%%%%%%%%%%%%%%%%%%%%%%%%%%%%%%%%%%%%%%%%%%%%%%%%%%%
\section{Algorithm}\label{section:algorithm}

The pseudocode for the \textsc{$\alpha$-FairPSolver} algorithm that is run at each node $j$ is provided in Fig~1. 
The basic intuition is that the algorithm keeps KKT conditions (\ref{eq:K1}) and (\ref{eq:K2}) satisfied and works towards (approximately) satisfying the remaining two KKT conditions (\ref{eq:K3}) and (\ref{eq:K4}) to minimize the duality gap. 
The algorithm can run in the distributed setting described in Section \ref{section:prelims}. In each round, an agent $j$ updates the value of $x_j$ based on the relative slack of all the constraints in which $j$ takes part, as long as the KKT condition (\ref{eq:K4}) of agent $j$ is not approximately satisfied. 
The updates need not be synchronous: we will require that all agents make updates at the same speed, but without access to a global clock.
\begin{figure}[!ht]
\small{
\hrulefill\\
\textsc{$\alpha$-FairPSolver}($\varepsilon$)\\[-8pt]
\hrule
\begin{algorithmic}[1]
\vspace{2pt}\Statex (Parameters $\delta_j, C, \kappa, \gamma, \beta_1,$ and $\beta_2$ are set as described in the text below the algorithm.)
\Statex In each round of the algorithm:
\State $x_j \leftarrow \max\{x_j, \delta_j\}$, $x_j = \min\{x_j, 1\}$
\State Update the dual variables: $y_i = C\cdot e^{\kappa \left(\sum_{j=1}^n A_{ij}x_j - 1 \right)}$ $\forall i\in\{1,...,m\}$
\If {$\frac{{x_j}^{\alpha}\cdot\sum_{i=1}^m y_i A_{ij}}{w_j}\leq (1-\gamma)$}
\State $x_j \leftarrow x_j \cdot(1+\beta_1)$
\Else\If {$\frac{{x_j}^{\alpha}\cdot\sum_{i=1}^m y_i A_{ij}}{w_j}\geq (1+\gamma)$}
\State $x_j \leftarrow \max\{x_j\cdot(1-\beta_2), \delta_j\}$
\EndIf\EndIf
\end{algorithmic}}\vspace{-10pt}
\hrulefill
\caption{Pseudocode of \textsc{$\alpha$-FairPSolver} algorithm.}\vspace{-10pt}
\end{figure}

To allow for self-stabilization and dynamic changes, the algorithm runs forever at all the agents, which is a standard requirement for self-stabilizing algorithms (see, e.g., \cite{dolev2000self}). The convergence of the algorithm is measured as the number of rounds between the round in which the algorithm starts from some initial solution and the round in which it reaches an $\varepsilon-$approximate solution, assuming that there are no hard reset events or node/constraint insertions/deletions in between. 

Without loss of generality, we assume that the input parameter $\varepsilon$ that determines the approximation quality satisfies $\varepsilon\leq \min\{\frac{1}{6}, \frac{9}{10\alpha}\}$ for any $\alpha$, and $\varepsilon\leq \frac{1-\alpha}{\alpha}$ for $\alpha<1$. 
The parameters $\delta_j, C, \kappa, \gamma$, $\beta_1$, and $\beta_2$ are set as follows. For technical reasons (mainly due to reinforcing dominant multiplicative updates of the variables $x_j$), we set the values of the lower thresholds $\delta_j$ below the actual lower bound of the optimal solution that we derive in Lemma \ref{lemma:lower-bound}:
\begin{equation*}
\delta_j = \left(\frac{1}{2}\cdot \frac{w_j}{w_{\max}}\right)^{1/\alpha}\cdot\begin{cases} \big(\frac{1}{m\cdot n^2\cdot A_{\max}}\big)^{1/\alpha}, & \mbox{if } 0<\alpha\leq 1\\
\frac{1}{m\cdot n^2 {A_{\max}}^{2-1/\alpha}}, & \mbox{if } \alpha> 1 \end{cases}.
\end{equation*}

We denote $\delta_{\max}\equiv \max_j \delta_j$, $\delta_{\min}\equiv \min_j \delta_j$. 
The constant $C$ that multiplies the exponent in the dual variables $y_i$ is chosen as $C = \frac{W}{\sum_{j=1}^n {\delta_j}^\alpha}$. Because  $\delta_j$ only depends on $w_j$ and on global parameters,  we also have $C = \frac{w_j}{{\delta_j}^\alpha}$, $\forall j$. The parameter $\kappa$ that appears in the exponent of the $y_i$'s is chosen as $\kappa = \frac{1}{\varepsilon}\ln\big(\frac{C m A_{\max}}{\varepsilon  w_{\min}}\big)$. The ``absolute error'' of (\ref{eq:K4}) $\gamma$ is set to $\varepsilon/4$. For $\alpha \geq 1$, we set $\beta_1 = \beta_2 = \beta$, where the choice of $\beta$ is described below. For $\alpha < 1$, we set $\beta_1 = \beta$, $\beta_2 = \beta^2 (\ln(\frac{1}{\delta_{\min}}))^{-1}$.

Similar to  \cite{AwerbuchKhandekar2009}, we choose the value of $\beta$ so that if we set $\beta_1=\beta_2 = \beta$, in any round the value of each $\frac{{x_j}^\alpha \sum_{i=1}^m y_i(x)A_{ij}}{w_j}$ changes by a multiplicative factor of at most $(1\pm \gamma/4)$. Since the maximum increase over any $x_j$ in each iteration is by a factor $1+\beta$, and $x$ is feasible in each round (see Lemma \ref{lemma:feasibility}), we have that $\sum_{j=1}^n A_{ij}x_j\leq 1$, and therefore, the maximum increase in each $y_i$ is by a factor of $e^{\kappa\beta}$. A similar argument holds for the maximum decrease.  Hence, we choose $\beta$ so that:
\begin{align*}
(1+\beta)^{\alpha}e^{\kappa\beta}\leq 1+\gamma/4 \quad\text{ and }\quad
(1-\beta)^{\alpha}e^{-\kappa\beta}\geq 1-\gamma/4,
\end{align*}
and it suffices to set:
% Observing that $1+\beta\leq\frac{1}{1-\beta}$, to satisfy the first inequality it is enough to satisfy $\frac{1}{(1-\beta)^{\alpha}}e^{\kappa\beta}\leq 1+ \frac{\gamma}{4}$, which is equivalent to $(1-\beta)^{\alpha}e^{-\kappa\beta}\geq(1+\frac{\gamma}{4})^{-1}$. Observing that $(1+\frac{\gamma}{4})^{-1}\leq 1-\frac{\gamma}{5}$, to satisfy both inequalities, it suffices to satisfy $(1-\beta)^{\alpha}e^{-\kappa\beta}\geq 1-\gamma/5$. There are two cases: $0<\alpha\leq 1$ and $\alpha>1$. If $\alpha\leq 1$, then we have
% $
% (1-\beta)^{\alpha}e^{-\kappa\beta}\geq(1-\beta)(1-\kappa\beta)=1-(\kappa+1)\beta+\kappa\beta^2,
% $
% and setting $\beta=\frac{\gamma}{5(\kappa+1)}$ is sufficient to satisfy both inequalities.
% If $\alpha>1$, then we have
% $
% (1-\beta)^{\alpha}e^{-\kappa\beta}\geq(1-\alpha\beta)(1-\kappa\beta)=1-(\kappa+\alpha)\beta+\kappa\alpha\beta^2,
% $
% and choosing $\beta = \frac{\gamma}{5(\kappa+\alpha)}$ satisfies both inequalities. Therefore:
\begin{equation*}
\beta = \begin{cases} \frac{\gamma}{5(\kappa+1)}, & \mbox{if } \alpha\leq 1 \\
\frac{\gamma}{5(\kappa+\alpha)}, & \mbox{if } \alpha> 1 \end{cases}.
\end{equation*}

\noindent\textbf{Remark:} In the $\alpha<1$ cases, since $\beta_2 = {\beta^2}({\ln({1}/{\delta_{\min}})})^{-1}$, the maximum decrease in $\frac{{x_j}^\alpha\sum_i y_i(x)A_{ij}}{w_j}$ is by a factor $(1-({\gamma}/{4})\cdot {\beta}({\ln(1/\delta_{\min})})^{-1})$, $\forall j$.

%%%%%%%%%%%%%%%%%%%%%%%%%%%%%%%%%%%%%%%%%%%%%%%%%%%%%%%%%%%%%%%%%%%%%%%%%%%%%%%%%%%%%%%%%%%%
%%%%%%%%%%%%%%%%%%%%%%%%%%%%%%%%%%%%%%%%%%%%%%%%%%%%%%%%%%%%%%%%%%%%%%%%%%%%%%%%%%%%%%%%%%%%
\section{Convergence Analysis}\label{section:convergence}
 
In this section, we analyze the convergence time of \textsc{$\alpha$-FairPSolver}. We first state our main theorems and provide some general results that hold for all $\alpha > 0$. We show that starting from an arbitrary solution, the algorithm reaches a feasible solution within poly-logarithmic (in the input size) number of rounds, and maintains a feasible solution forever after. Similar to \cite{AwerbuchKhandekar2009, dc-young2001sequential, kelly1998rate}, we use a concave potential function that, for feasible $x$, is bounded below and above and increases with any algorithm update. %(of any variable $x_j$)
\iffullpaper Then, we analyze the convergence time separately for three cases: $\alpha<1$, $\alpha = 1$, and $\alpha>1$. 
With an appropriate definition of a \emph{stationary round} for each of the three cases, we show that in every stationary round, $x$ approximates ``well'' the optimal solution by bounding the duality gap. On the other hand, for any non-stationary round, we show that the potential increases substantially. This large increase in the potential then leads to the conclusion that there cannot be too many non-stationary rounds, thus bounding the overall convergence time. \else Then, we sketch the proof of Theorem \ref{thm:convergence-alpha>1} ($\alpha>1$), while we defer the full proofs of  the three theorems to the full version of this paper. The main proof idea in all the cases is as follows. With an appropriate definition of a \emph{stationary round} for each of the three cases $\alpha<1$, $\alpha = 1$, and $\alpha>1$, we show that in every stationary round, $x$ approximates ``well'' the optimal solution by bounding the duality gap. On the other hand, for any non-stationary round, we show that the potential increases substantially. This large increase in the potential leads to the conclusion that there cannot be too many non-stationary rounds, thus bounding the overall convergence time. \fi

We make a few remarks here. First, we require that $\alpha$ be bounded away from zero. This requirement is without loss of generality because we show that when $\alpha \leq \frac{\varepsilon/4}{\ln(nA_{\max}/\varepsilon)}$, any $\varepsilon-$approximation LP provides a $3\varepsilon-$approximate solution to $(P_\alpha)$ (Lemma \ref{lemma:LP-close-to-small-alpha-fair}). Thus, when $\alpha \leq \frac{\varepsilon/4}{\ln(nA_{\max}/\varepsilon)}$ we can switch to the algorithm of \cite{AwerbuchKhandekar2009}, and when $\alpha > \frac{\varepsilon/4}{\ln(nA_{\max}/\varepsilon)}$, the convergence time remains poly-logarithmic in the input size and polynomial in $\varepsilon^{-1}$. %When $\alpha$ is close to zero, the  values of $\delta_j$ are also close to zero, and hence, the multiplicative increase leads to a very small progress in $x_j$. Of course, when $\alpha=0$ we can just use the algorithm of \cite{AwerbuchKhandekar2009}. 
Second, the assumption that $\varepsilon\leq \frac{1-\alpha}{\alpha}$ in the $\alpha<1$ case is also without loss of generality, because we show that when $\alpha$ is close to 1 (roughly, $1 - O(\varepsilon^2/\ln^2(\wratio mnA_{\max}/\varepsilon))$), we can approximate $(P_\alpha)$ by switching to the $\alpha=1$ case of the algorithm (Lemma \ref{lemma:alpha-close-to-1-below}). 
Finally, when $\alpha > 1$, the algorithm achieves an $\varepsilon-$approximation in time ${O}(\alpha^4\varepsilon^{-4} \ln^2(\wratio nmA_{\max}\varepsilon^{-1}))$. %An $\varepsilon-$approximation can then be achieved by setting $\varepsilon' = \varepsilon/\alpha$ as the accuracy parameter, which then introduces an $\alpha^4$ multiplicative term in the convergence time bound. 
We believe that a polynomial dependence on $\alpha$ is difficult to avoid in this setting, because by increasing $\alpha$, the gradient of the $\alpha$-fair utilities $f_{\alpha}$ blows up on the interval $(0, 1)$: as $\alpha$ increases, $f_{\alpha}(x)$ quickly starts approaching a step function that is equal to $-\infty$ on the interval $(0, 1]$ and equal to 0 on the interval $(1, \infty]$. To characterize the behavior of $\alpha-$fair allocations as $\alpha$ becomes large, we show that when $\alpha \geq {\varepsilon}^{-1}{\ln(\wratio n A_{\max})}$, all the coordinates of the $\alpha-$fair vector are within a $1\pm \varepsilon$ multiplicative factor of the corresponding coordinates of the max-min fair vector (Lemma \ref{lemma:mmf-alpha-fair}). 

{Finally, we note that the main convergence argument from \cite{AwerbuchKhandekar2009} that uses an appropriate definition of \emph{stationary intervals} does not extend to our setting. The proof from \cite{AwerbuchKhandekar2009} ``breaks'' in the part that shows that the solution is $\varepsilon-$approximate throughout any stationary interval, stated as Lemma 3.7 in \cite{AwerbuchKhandekar2009}. The proof of Lemma 3.7 in \cite{AwerbuchKhandekar2009} is by contradiction: assuming that the solution is not $\varepsilon-$approximate, the proof proceeds by showing that at least one of the variables would increase in each round of the stationary interval, thus eventually making the solution infeasible and contradicting one of the preliminary lemmas. For $\alpha\geq 1$, unlike the linear objective in \cite{AwerbuchKhandekar2009}, $\alpha$-fair objectives are negative, and the assumption that the solution is not $\varepsilon-$approximate does not lead to any conclusive information. For $\alpha<1$, adapting the proof of Lemma 3.7 from \cite{AwerbuchKhandekar2009} leads to the conclusion that for at least one $j$, in each round $t$ of the stationary interval $\frac{{(x_j^*)}^{\alpha}\sum_i y_i(x^t)A_{ij}}{w_j}\leq 1-\gamma$, where $x^*$ is the optimal solution, and $x^t$ is the solution at round $t$. In \cite{AwerbuchKhandekar2009}, where $\alpha=0$, this implies that $x_j$ increases in each round of the stationary interval, while in our setting ($\alpha>0$) it is not possible to draw such a conclusion. }

%%%%%%%%%%%%%%%%%%%%%%%%%%%%%%%%%%%%%%%%%%%%%%%%%%%%%%%%%%%%%%%%%%%%%%%%%%%%%%%%%%%%%%%%%%%%
\paragraph{Main Results.}%\label{section:main-results}
 
Our main results are summarized in the following three theorems. 
The objective %$\sum_j w_j f_\alpha(x_j)$ 
is denoted by $p_\alpha(x)$, $x^t$ denotes the solution at the beginning of round $t$, and $x^*$ denotes the optimal solution.

\begin{theorem}\label{thm:convergence-alpha<1}
(Convergence for $\alpha < 1$) \textsc{$\alpha$-FairPSolver} solves $(P_\alpha)$ approximately for $\alpha < 1$ in time that is polynomial in $\frac{\ln(nmA_{\max})}{\alpha\varepsilon}$. In particular, after at most
\begin{equation}
O\left(\alpha^{-2}\varepsilon^{-5}\ln^2\left(\wratio mnA_{\max}\right)\ln^2\left(\varepsilon^{-1} \wratio mnA_{\max}\right)\right)
\label{eq:thm-alpha<1-convergence-time}
\end{equation}
rounds, there exists at least one round $t$ such that $p_\alpha(x^*) - p_\alpha(x^t)\leq \varepsilon p_\alpha(x^t)$. Moreover, the total number of rounds $s$ in which $p_\alpha(x^*) - p_\alpha(x^s)> \varepsilon p_\alpha(x^s)$ is also bounded by (\ref{eq:thm-alpha<1-convergence-time}).
\end{theorem}

\begin{theorem}\label{thm:convergence-alpha=1} (Convergence for $\alpha = 1$) \textsc{$\alpha$-FairPSolver} solves $(P_1)$ approximately in time that is polynomial in $\varepsilon^{-1}\ln(\wratio nmA_{\max})$. 
In particular, after at most  
\begin{equation}
O\left(\varepsilon^{-5} \ln^2\left(\wratio nmA_{\max} \right)\ln^2\left( \varepsilon^{-1} \wratio nmA_{\max} \right)\right) \label{eq:alpha=1-conv-time-bound}
\end{equation}
rounds, there exists at least one round $t$ such that $p(x^*) - p(x^t)\leq \varepsilon W$. Moreover, the total number of rounds $s$ in which $p(x^*) - p(x^s)> \varepsilon W$ is also bounded by (\ref{eq:alpha=1-conv-time-bound}).
\end{theorem}

\begin{theorem}\label{thm:convergence-alpha>1}
(Convergence for $\alpha>1$) \textsc{$\alpha$-FairPSolver} solves $(P_\alpha)$ approximately for $\alpha > 1$ in time that is polynomial in ${\varepsilon}^{-1} \ln(nmA_{\max})$. 
In particular, after at most:
\begin{equation}
O\left(\alpha^4\varepsilon^{-4} \ln\left( \wratio nmA_{\max}\right) \ln\left(\varepsilon^{-1} \wratio nmA_{\max}\right)\right) \label{eq:thm-alpha>1-convergence-time}
\end{equation} 
rounds, there exists at least one round $t$ such that $p_\alpha(x^*) - p_\alpha(x^t) \leq \varepsilon(-p_{\alpha}(x^t))$. Moreover, the total number of rounds $s$ in which $p_\alpha(x^*) - p_\alpha(x^s) > \varepsilon(-p_{\alpha}(x^s))$ is also bounded by (\ref{eq:thm-alpha>1-convergence-time}).
\end{theorem}

\iffullpaper \else Proofs of Theorem \ref{thm:convergence-alpha<1} and Theorem \ref{thm:convergence-alpha=1} are provided in {the full version of the paper}. We sketch the proof of Theorem \ref{thm:convergence-alpha>1} in Section \ref{section:alpha>1}.\fi

%%%%%%%%%%%%%%%%%%%%%%%%%%%%%%%%%%%%%%%%%%%%%%%%%%%%%%%%%%%%%%%%%%%%%%%%%%%%%%%%%%%%%%%%%%%%
\paragraph{Feasibility and Approximate Complementary Slackness.}
The following three lemmas are preliminaries for the convergence time analysis. Lemma \ref{lemma:feasibility} shows that starting from a feasible solution, the algorithm always maintains a feasible solution. Lemma \ref{lemma:self-stabilization} shows that any violated constraint becomes feasible within poly-logarithmic number of rounds, and remains feasible forever after. Combined with Lemma \ref{lemma:feasibility}, Lemma \ref{lemma:self-stabilization} allows us to focus only on the rounds with feasible solutions $x$. Lemma \ref{lemma:approx-comp-slack} shows that after a poly-logarithmic number of rounds,  approximate complementary slackness (KKT condition (\ref{eq:K3})) holds in an aggregate sense: $\sum_{i=1}^m y_i (x)\big(\sum_{j=1}^n A_{ij}x_j - 1\big)\approx 0$. 

\begin{lemma}\label{lemma:feasibility}
If the algorithm starts from a feasible solution, then the algorithm maintains a feasible solution $x$: $x_j\geq 0$, $\forall j$ and $\sum_{j=1}^n A_{ij}x_j\leq 1$, $\forall i$, in each round.
\end{lemma}
\iffullpaper
\begin{proof}
By the statement of the lemma, the solution is feasible initially. From the way that the algorithm makes updates to the variables $x_j$, it is always true that $x_j \geq 0$, $\forall j$.

Now assume that $x$ becomes infeasible in some round, and let $x^0$ denote the (feasible) solution before that round, $x^1$ denote the (infeasible) solution after the round. We have:
\begin{align*}
\sum_{\ell=1}^n A_{i\ell}x_{\ell}^0 \leq 1, \quad\forall i\in\{1,...,m\}, \quad\text{ and }\quad
\sum_{\ell=1}^n A_{k\ell}x_{\ell}^1 >1, \quad\text{for some } k\in\{1,...,m\}.
\end{align*}
For this to be true, $x$ must have increased over at least one coordinate $j$ such that $A_{kj}\neq 0$. For such a change to be triggered by the algorithm, it must also be true that:
\begin{equation*}
(x_j^0)^{\alpha}\sum_{i=1}^m y_i(x^0) A_{ij}\leq w_j\left(1-\gamma\right).
\end{equation*}
Since, by the choice of $\beta_1 = \beta$, this term can increase by a factor of at most $1+\gamma/4$, it follows that:
\begin{equation*}
(x_j^1)^{\alpha}\sum_{i=1}^m y_i(x^1) A_{ij}\leq w_j(1-\gamma)\left(1+\frac{\gamma}{4}\right)<w_j.
\end{equation*}
This further implies:
\begin{equation*}
(x_j^1)^{\alpha}y_k(x^1)A_{kj}<w_j,
\end{equation*}
and since whenever $A_{kj}\neq 0$ we also have $A_{kj}\geq 1$, we get:
\begin{equation}
(x_j^1)^{\alpha}y_k(x^1)<w_j. \label{eq:feasibility-less-that-wj}
\end{equation}
On the other hand, since $x_j^1\geq \delta_j$, ${\delta_j}^{\alpha}=\frac{w_j}{C}$, and $\sum_{j=1}^n A_{kj}x_j^1 >1$:
\begin{equation*}
(x_j^1)^{\alpha} y_k(x^1) \geq \frac{w_j}{C}\cdot C\cdot e^{\kappa(\sum_{j=1}^n A_{kj}x_j^1 -1)}>w_j,
\end{equation*}
which contradicts (\ref{eq:feasibility-less-that-wj}).
\end{proof}
\fi

\begin{lemma}\label{lemma:self-stabilization}
If for any $i$: $\sum_{j=1}^n A_{ij}x_j > 1$, then after at most $\tau_1 = O(\frac{1}{\beta_2}\ln(nA_{\max}))$ rounds, it is always true that $\sum_{j=1}^n A_{ij}x_j \leq 1$.
\end{lemma}
\begin{proof}
Suppose that $\sum_{j=1}^n A_{ij}x_j > 1$ for some $i$. Then $y_i > C$, and for every $x_j$ with $A_{ij}\neq 0$:
\begin{align*}
{x_j}^{\alpha}\sum_{l=1}^m y_l(x)A_{lj} \geq {x_j}^{\alpha}y_i(x)A_{ij}\geq {\delta_j}^{\alpha}C \geq w_j>w_j(1-\gamma),
\end{align*}
and therefore, none of the variables that appear in $i$ increases.

Since $\sum_{j=1}^n A_{ij}x_j > 1$, there exists at least one $x_k$ with $A_{ik}\neq 0$ such that $x_k \geq \frac{\sum_{j=1}^n A_{ij}x_j}{A_{ik}n}>\frac{1}{nA_{\max}}$. For each such $x_k$, since $C\geq 2w_{\max}nA_{\max}$:
\begin{equation*}
{x_k}^{\alpha}\sum_{l=1}^m y_l(x)A_{lj}\geq C\frac{1}{nA_{\max}}\geq 2 w_{\max} > w_k(1+\gamma),  
\end{equation*}
and therefore, $x_k$ decreases (by a factor $(1-\beta_2)$). As $x_k\leq 1$, after at most $O(\frac{1}{\beta_2}\ln(nA_{\max}))$ rounds in which $\sum_{j=1}^n A_{ij}x_j > 1$, we must have $x_k \leq \frac{1}{nA_{\max}}$, and therefore, $\sum_{j=1}^n A_{ij}x_j \leq 1$.
 
Using the same arguments as in the proof of Lemma \ref{lemma:feasibility}, the constraint $i$ never gets violated again.
\end{proof}

%A simple corollary of Lemma \ref{lemma:self-stabilization} is that the algorithm reaches a feasible solution after a node or constraint insertion/deletion, as long as all the nodes can maintain a proper estimate of the upper bounds on $n, m, A_{\max}$, and $w_{\max}$.

\begin{lemma}\label{lemma:approx-comp-slack}
If the algorithm starts from a feasible solution, then after at most $\tau_0 = \frac{1}{\beta}\ln\left(\frac{1}{\delta_{\min}}\right)$ rounds, it is always true that:
\begin{enumerate}[noitemsep,topsep=3pt]
\item There exists at least one approximately tight constraint: $\max_i\big\{\sum_{j=1}^n A_{ij}x_j\big\}\geq 1-(1+1/\kappa)\varepsilon$, 
\item $\sum_{i=1}^m y_i\leq (1+3\varepsilon)\sum_{j=1}^nx_j\sum_{i=1}^m y_i A_{ij}$, and
\item $(1-3\varepsilon)\sum_{i=1}^m y_i\leq \sum_{j=1}^n x_j\sum_{i=1}^m y_i A_{ij}\leq \sum_{i=1}^m y_i$.
\end{enumerate}
\end{lemma}
\iffullpaper
\begin{proof}
Suppose that $\max_i \sum_{j=1}^n A_{ij}x_j < 1- \varepsilon$. Then for each $y_i$ we have:
\begin{equation*}
y_i\leq C\cdot e^{-\kappa\varepsilon}=C\cdot\frac{\varepsilon  w_{\min}}{CmA_{\max}}=\frac{\varepsilon  w_{\min}}{m{A_{\max}}}.
\end{equation*}
Due to Lemma \ref{lemma:feasibility}, we have that $x$ is feasible in every round, which implies that $x_j\leq 1$ $\forall j$. This further gives:
\begin{equation*}
{x_j}^{\alpha}\sum_{i=1}^m y_iA_{ij}\leq w_j\varepsilon \leq w_j(1-\gamma),
\end{equation*}
and, therefore, all variables $x_j$ increase by a factor $1+\beta$. From Lemma \ref{lemma:feasibility}, since the solution always remains feasible, none of the variables can increase to a value larger than 1. Therefore, after at most $\tau_0=\log_{1+\beta}\left(\frac{1}{\delta_{\max}}\right) \leq \frac{1}{\beta}\ln\left(\frac{1}{\delta_{\max}}\right)$ rounds, there must exist at least one $i$ such that $\sum_{j=1}^n A_{ij}x_j\geq 1-\varepsilon$. If in any round $\max_{i}\sum_{j=1}^n A_{ij}x_j$ decreases, it can decrease by at most $\beta_2\sum_{j=1}^n A_{ij}x_j\leq\beta\sum_{j=1}^n A_{ij}x_j\leq\beta<\frac{\varepsilon}{5\kappa}$. Therefore, in every subsequent round 
\begin{equation*}\max_{i}\sum_{j=1}^n A_{ij}x_j>1-\Big(1+\frac{1}{5\kappa}\Big)\varepsilon.
\end{equation*}

For the second part of the lemma, let $S = \{i: \sum_{j=1}^nA_{ij}x_j < \max_{k\in \{1,...,m\}}\sum_{j=1}^n A_{kj}x_j - \frac{\kappa-1}{5\kappa}\varepsilon\}$ be the set of constraints that are at least ``$\frac{\kappa-1}{5\kappa}\varepsilon$-looser" than the tightest constraint. Then for $i\in S$ we have 
\begin{equation*}
y_i \leq e^{-\frac{\kappa-1}{5}\varepsilon}\max_{k\in \{1,...,m\}}y_k < \frac{\varepsilon}{m}e^{\varepsilon/5}\max_{k\in \{1,...,m\}}y_k < 1.2\frac{\varepsilon}{m}\max_{k\in \{1,...,m\}}y_k.
\end{equation*}
This further gives:
\begin{equation*}
\sum_{i=1}^m y_i = \sum_{i\in S}y_i + \sum_{k\notin S}y_k<(1+1.2\varepsilon)\sum_{i\notin S}y_i. 
\end{equation*}
 Moreover, for each $i\notin S$ we have $y_i\sum_{j=1}^nA_{ij}x_j \geq (1-1.2\varepsilon)y_i$, since for $i\notin S$:
 \begin{align*}
\sum_{j=1}^nA_{ij}x_j \geq \max_{k\in \{1,...,m\}}A_{kj}x_j - \frac{\kappa-1}{5\kappa}\varepsilon\geq 1- \left(1+\frac{1}{5\kappa} + \frac{\kappa-1}{5\kappa}\right)\varepsilon = 1 - 1.2\varepsilon.
\end{align*}
Therefore:
\begin{align*}
\sum_{i=1}^m y_i &<\frac{1+1.2\varepsilon}{1-1.2\varepsilon}\sum_{i\notin S}y_i\sum_{j=1}^n A_{ij}x_j\\
%&\leq (1+1.2\varepsilon)(1+1.5\varepsilon)\sum_{i\notin S}y_i\sum_{j=1}^n A_{ij}x_j\\
%&= (1+2.7\varepsilon + 1.8\varepsilon^2)\sum_{i\notin S}y_i\sum_{j=1}^n A_{ij}x_j\\
&\leq (1+3\varepsilon)\sum_{i\notin S}y_i\sum_{j=1}^n A_{ij}x_j \quad(\text{from }\varepsilon\leq 1/6)\\
&\leq (1+3\varepsilon)\sum_{i=1}^m y_i\sum_{j=1}^n A_{ij}x_j.
\end{align*}
Interchanging the order of summation in the last line, we reach the desired inequality.

The proof of the last part of the lemma follows from feasibility: $\sum_{j}A_{ij}x_j \leq 1$, $\forall i$ (Lemma \ref{lemma:feasibility}), and from $\frac{1}{1+3\varepsilon} \geq 1 - 3\varepsilon$.%, as $0<\varepsilon\leq 1/6$.
\end{proof}
\fi
Lemmas analogous to \ref{lemma:feasibility} and \ref{lemma:approx-comp-slack} also appear in \cite{AwerbuchKhandekar2009}. {However, the proofs of Lemmas \ref{lemma:feasibility} and \ref{lemma:approx-comp-slack} require new ideas compared to the proofs of the corresponding lemmas in \cite{AwerbuchKhandekar2009}. We need to be much more careful in our choice of lower thresholds $\delta_j$ and constant $C$ in the dual variables, particularly by choosing $C$ as a function of several variables, rather than as a constant. The choice of $\delta_j$'s is also sensitive as smaller $\delta_j$'s would make the potential function range too large, while larger $\delta_j$'s would cause more frequent decrease of ``small'' variables.  In either case, the convergence time would increase. }%, however, due to the different algorithm parameters and different objectives, the proofs are not completely the same.
%%%%%%%%%%%%%%%%%%%%%%%%%%%%%%%%%%%%%%%%%%%%%%%%%%%%%%%%%%%%%%%%%%%%%%%%%%%%%%%%%%%%%%%%%%%%
\paragraph{Decrease of Small Variables.}

The following lemma is also needed for the convergence analysis. It shows that if some variable $x_j$  decreases by less than a multiplicative factor $(1-\beta_2)$, i.e., $x_j<\frac{\delta_j}{1-\beta_2}$ and $x_j$ decreases, then $x_j$ must be part of at least one approximately tight constraint. This lemma will be used later to show that in any round the increase in the potential due to the decrease of ``small'' variables is dominated by the decrease of ``large'' variables (i.e., the variables that decrease by a multiplicative factor $(1-\beta_2)$). %The proof of Lemma \ref{lemma:small-x-tight-yi} is provided in Appendix \ref{appendix:proofs-small-updates}.

\begin{lemma}\label{lemma:small-x-tight-yi}
Consider the rounds that happen after the initial $\tau_1 = O(\frac{1}{\beta_2}\ln(nA_{\max}))$ rounds. If in some round there is  a variable $x_j<\frac{\delta_j}{1-\beta_2}$ that decreases, then in the same round for some $i$ with $A_{ij} \neq 0$ it holds that: $y_i(x) \geq \frac{\sum_{l=1}^m A_{lj}y_l(x)}{mA_{\max}}$ and $\sum_{k=1}^n A_{ik}x_k > 1 - \frac{\varepsilon}{2}$.
\end{lemma}
\iffullpaper
\begin{proof}
Suppose that some $x_j<\frac{\delta_j}{1-\beta_2}$ triggers a decrease over the $j^\text{th}$ coordinate.
The first part of the Lemma is easy to show, simply by using the argument that at least one term of a summation must be higher than the average, i.e., there exists at least one $i$ with $A_{ij}\neq 0$ such that:
\begin{align*}
y_i(x) A_{ij} \geq \frac{\sum_{l=1}^m A_{lj}y_l(x)}{m}
\quad\Rightarrow\quad y_i \geq \frac{\sum_{l=1}^m A_{lj}y_l(x)}{mA_{\max}}.
\end{align*}

For the second part, as $x_j < \frac{\delta_j}{1-\beta_2}$, we have that:
\begin{align*}
{x_j}^\alpha y_i(x) \geq \frac{{x_j}^\alpha\sum_{l=1}^m A_{lj}y_l(x)}{mA_{\max}} 
\quad\Rightarrow\quad y_i(x) > \frac{(1-\beta_2)^\alpha}{{\delta_j}^\alpha}\frac{{x_j}^\alpha\sum_{l=1}^m A_{lj}y_l(x)}{mA_{\max}}.
\end{align*}

Since $x_j$ decreases, we have that ${x_j}^{\alpha}\sum_{l=1}^m y_i(x)A_{lj}\geq w_j(1+\gamma)$, and therefore  
$
y_i(x)>\frac{w_j}{{\delta_j}^{\alpha}}\frac{(1+\gamma)(1-\beta_2)^\alpha}{m{A_{\max}}}.
$
Moreover, as $y_i(x)=C\cdot e^{\kappa(\sum_{k=1}^nA_{ik}x_k - 1)}$, and $C = \frac{w_j}{{\delta_j}^{\alpha}}$, it follows that:
\begin{equation}
e^{\kappa(\sum_{k=1}^nA_{ik}x_k - 1)} > \frac{(1+\gamma)(1-\beta_2)^\alpha}{m{A_{\max}}}.\label{eq:y-bound-1}
\end{equation}
Observe that for $\alpha \leq 1$:
\begin{equation}\label{eq:eps-8-1}
(1+\gamma)(1-\beta_2)^\alpha\geq(1+\gamma)(1-\beta_2)>\left(1+\frac{\varepsilon}{4}\right)\left(1-\frac{\varepsilon}{20(\kappa+1)}\right)>1>\sqrt{\varepsilon},
\end{equation}
while for $\alpha > 1$, since $\varepsilon\alpha \leq \frac{9}{10}$:
\begin{align}
(1+\gamma)(1-\beta_2)^\alpha \geq (1+\gamma)(1-\alpha\beta_2) \geq \left(1+\gamma\right)\left(1 - \frac{\gamma\varepsilon\alpha}{5}\right)\geq 1>\sqrt{\varepsilon} \label{eq:eps-8-1-2},
\end{align}
where we have used the generalized Bernoulli's inequality for $(1-\beta_2)^\alpha \geq (1-\alpha\beta_2)$ \cite{mitrinovic1970analytic}, and then $\beta_2 = \beta = \frac{\gamma}{5(\kappa+\alpha)}<\frac{\gamma\varepsilon}{5}$.
Recalling that $\kappa = \dfrac{1}{\varepsilon}\ln\left(\dfrac{CmA_{\max}}{\varepsilon w_{\min}}\right)$, and combining (\ref{eq:y-bound-1}) with (\ref{eq:eps-8-1}) and (\ref{eq:eps-8-1-2}):
\begin{equation*}
\left(\frac{\varepsilon w_{\min}}{CmA_{\max}}\right)^{\frac{1-\sum_{k=1}^nA_{ik}x_k}{\varepsilon}}>\frac{\sqrt{\varepsilon}}{m{A_{\max}}}.
\end{equation*}
Finally, as $C \geq 2 w_{\max}nmA_{\max}$, it follows that 
$
\frac{ w_{\min}\varepsilon}{CmA_{\max}}\leq\frac{\varepsilon w_{\min}}{2w_{\max}nm^2{A_{\max}}^2}< \left(\frac{\sqrt{\varepsilon}}{m{A_{\max}}}\right)^2<1,
$
which gives:
\begin{equation*}
\frac{1-\sum_{k=1}^nA_{ik}x_k}{\varepsilon} < \frac{1}{2} \Leftrightarrow \sum_{k=1}^nA_{ik}x_k > 1- \frac{\varepsilon}{2}.
\end{equation*}
\end{proof}
\fi
 
%%%%%%%%%%%%%%%%%%%%%%%%%%%%%%%%%%%%%%%%%%%%%%%%%%%%%%%%%%%%%%%%%%%%%%%%%%%%%%%%%%%%%%%%%%%%
\paragraph{Potential.}

We use the following potential function to analyze the convergence time:
\begin{align*}
\Phi(x) = p_\alpha(x) - \frac{1}{\kappa}\sum_{i=1}^m y_i(x),
\end{align*}
where $p_\alpha(x) = \sum_{j=1}^n w_jf_{\alpha}({x_j})$ and $f_\alpha$ is defined by (\ref{eq:f-alpha}).
The potential function is strictly concave and its partial derivative with respect to any variable $x_j$ is:
\begin{align}
\frac{\partial \Phi(x)}{\partial x_j} = \frac{w_j}{{x_j}^{\alpha}}-\sum_{i=1}^m y_i(x)A_{ij}
= \frac{w_j}{{x_j}^{\alpha}}\left(1 - \frac{{x_j}^{\alpha}\sum_{i=1}^my_i(x)A_{ij}}{w_j}\right).\label{eq:potential-derivatives}
\end{align}

The following fact (given in a similar form in \cite{AwerbuchKhandekar2009}), which follows directly from the Taylor series representation of concave functions, will be useful for the potential increase analysis:
\begin{fact}\label{fact:taylor}
For a differentiable concave function $f:\mathbb{R}^n\rightarrow\mathbb{R}$ and any two points $x^0, x^1 \in \mathbb{R}^n$:
\begin{align*}
\sum_{j=1}^n \frac{\partial f(x^0)}{\partial x_j}(x_j^1-x_j^0)\geq f(x^1)-f(x^0)\geq \sum_{j=1}^n \frac{\partial f(x^1)}{\partial x_j}(x_j^1-x_j^0).
\end{align*}
\end{fact}

Using Fact \ref{fact:taylor} and (\ref{eq:potential-derivatives}), we show the following lemma:

\begin{lemma}\label{lemma:potential-increase}
Starting with a feasible solution and throughout the course of the algorithm, the potential function $\Phi(x)$ never decreases. Letting $x^0$ and $x^1$ denote the values of $x$ before and after a round update, respectively, the potential function increase is lower-bounded as:
\begin{equation*}
\Phi(x^1) -\Phi(x^0) \geq \sum_{j=1}^n w_j \frac{\big|x_j^1 - x_j^0\big|}{(x_j^1)^{\alpha}}\Big|1 - \frac{(x_j^1)^{\alpha}\sum_{i=1}^my_i(x^1)A_{ij}}{w_j} \Big|. \label{eq:general-potential-increase}
\end{equation*}
\end{lemma}
\iffullpaper
\begin{proof}
Since $\Phi$ is concave, using Fact \ref{fact:taylor} and (\ref{eq:potential-derivatives}) it follows that:
\begin{align}
\Phi(x^1) - \Phi(x^0) \geq \sum_{j=1}^n w_j \frac{x_j^1 - x_j^0}{(x_j^1)^{\alpha}}\left(1 - \frac{(x_j^1)^{\alpha}\sum_{i=1}^my_i(x^1)A_{ij}}{w_j} \right).\label{eq:potential-increase-before-bounds}
\end{align}
If $x_j^1 = x_j^0$, then the term in the summation (\ref{eq:potential-increase-before-bounds}) corresponding to the change in $x_j$ is equal to zero, and $x_j$ has no contribution to the sum in (\ref{eq:potential-increase-before-bounds}). 

If $x_j^1 - x_j^0 > 0$, then, as $x_j$ increases over the observed round, it must be $\frac{(x_j^0)^{\alpha}\sum_{i=1}^my_i(x^0)A_{ij}}{w_j}\leq 1-\gamma$. By the choice of the parameters, $\frac{(x_j^1)^{\alpha}\sum_{i=1}^my_i(x^1)A_{ij}}{w_j}\leq \left(1+\frac{\gamma}{4}\right)\Big(\frac{(x_j^0)^{\alpha}\sum_{i=1}^my_i(x^0)A_{ij}}{w_j}\Big)$, and therefore
\begin{align}
\frac{(x_j^1)^{\alpha}\sum_{i=1}^my_i(x^1)A_{ij}}{w_j}\leq \left(1+\frac{\gamma}{4}\right)(1-\gamma) = 1 - \frac{3}{4}\gamma - \frac{\gamma^2}{4}
<1 - \frac{3}{4}\gamma. \label{eq:alpha<1-zj-increase}
\end{align}
It follows that $1 - \frac{(x_j^1)^{\alpha}\sum_{i=1}^my_i(x^1)A_{ij}}{w_j}>\frac{3}{4}\gamma>0$, and therefore 
\begin{equation*}
w_j \frac{x_j^1 - x_j^0}{(x_j^1)^{\alpha}}\left(1 - \frac{(x_j^1)^{\alpha}\sum_{i=1}^my_i(x^1)A_{ij}}{w_j} \right) = w_j \frac{\big|x_j^1 - x_j^0\big|}{(x_j^1)^{\alpha}}\Bigg|1 - \frac{(x_j^1)^{\alpha}\sum_{i=1}^my_i(x^1)A_{ij}}{w_j} \Bigg|.
\end{equation*}

Finally, if $x_j^1-x_j^0<0$, then it must be $\frac{(x_j^0)^{\alpha}\sum_{i=1}^my_i(x^0)A_{ij}}{w_j}\geq 1+\gamma$. By the choice of the parameters, $\frac{(x_j^1)^{\alpha}\sum_{i=1}^my_i(x^1)A_{ij}}{w_j}\geq \left(1-\frac{\gamma}{4}\right)\Big(\frac{(x_j^0)^{\alpha}\sum_{i=1}^my_i(x^0)A_{ij}}{w_j}\Big)$, implying
\begin{align}
\frac{(x_j^1)^{\alpha}\sum_{i=1}^my_i(x^1)A_{ij}}{w_j}\geq \left(1-\frac{\gamma}{4}\right)(1+\gamma) = 1 + \frac{3}{4}\gamma - \frac{\gamma^2}{4}
>1 + \frac{1}{2}\gamma. \label{eq:alpha<1-zj-decrease}
\end{align}
We get that $1 - \frac{(x_j^1)^{\alpha}\sum_{i=1}^my_i(x^1)A_{ij}}{w_j}<-\frac{1}{2}\gamma<0$, and therefore 
\begin{equation*}
w_j \frac{x_j^1 - x_j^0}{(x_j^1)^{\alpha}}\left(1 - \frac{(x_j^1)^{\alpha}\sum_{i=1}^my_i(x^1)A_{ij}}{w_j} \right) = w_j \frac{\big|x_j^1 - x_j^0\big|}{(x_j^1)^{\alpha}}\left|1 - \frac{(x_j^1)^{\alpha}\sum_{i=1}^my_i(x^1)A_{ij}}{w_j} \right|,
\end{equation*}
completing the proof.
\end{proof}
\fi

\iffullpaper
%%%%%%%%%%%%%%%%%%%%%%%%%%%%%%%%%%%%%%%%%%%%%%%%%%%%%%%%%%%%%%%%%%%%%%%%%%%%%%%%%%%%%%%%%%%%
\subsection{Proof of Theorem \ref{thm:convergence-alpha<1}}\label{section:alpha<1}

The outline of the proof is as follows. We first derive a lower bound on the potential increase (Lemma \ref{lemma:potential-increase-alpha<1}), which will motivate the definition of a stationary round. Then, for the appropriate definition of a stationary round we will first show that in any stationary round, solution is $O(\varepsilon)-$approximate. Then, to complete the proof, we will show in any non-stationary round there is a sufficiently large increase in the potential function, which, combined with the bounds on the potential value will yield the result.

The following lemma lower-bounds the increase in the potential function in any round of the algorithm. 

\begin{lemma}\label{lemma:potential-increase-alpha<1}
If $\alpha<1$ and $\Phi(x^0)$, $x^0$, $y(x^0)$ and $\Phi(x^1)$, $x^1$, $y(x^1)$ denote the values of $\Phi$, $x$, and $y$ before and after a round, respectively, and $S^- = \{j: x_j \text{ decreases}\}$, then if $x^0$ is feasible:
\begin{enumerate}
\item $\Phi(x^1) - \Phi(x^0)\geq \Omega(\beta^2\gamma/\ln(1/\delta_{\min}))\sum_{j\in S^-}w_j \frac{{(x_j^0)}^{1-\alpha}}{1-\alpha}$;
\item $\Phi(x^1) - \Phi(x^0) \geq \Omega(\beta)\left((1-\gamma)\sum_{j=1}^n w_j{(x_j^0)^{1-\alpha}} - \sum_{i=1}^m y_i(x^0) \sum_{j=1}^n A_{ij}x_j^0\right)$;
\item $\Phi(x^1) - \Phi(x^0) \geq \Omega\Big(\frac{\beta^2}{\ln(1/\delta_{\min})}\Big)\left(\sum_{i=1}^m y_i(x^0) \sum_{j=1}^n A_{ij}x_j^0 - (1+\gamma)\sum_{j=1}^n w_j{(x_j^0)^{1-\alpha}}\right)$.
\end{enumerate}
\end{lemma}
\begin{proof}
$\quad$\\
\noindent\textbf{Proof of 1.} Observe that for $j\in S^-$, $x_j^1 = \max\{\delta_j, (1-\beta_2)x_j^0\}$. From the proof of Lemma \ref{lemma:potential-increase}, we have that: 
\begin{equation*}\Phi(x^1)-\Phi(x^0)\geq \sum_{j\in S^-} w_j \frac{x_j^0 - x_j^1}{(x_j^1)^{\alpha}}\left( \frac{(x_j^1)^{\alpha}\sum_{i=1}^my_i(x^1)A_{ij}}{w_j} - 1\right).
\end{equation*}
The proof that
\begin{equation*}
\sum_{j\in S^-} w_j {(x_j^0)}^{1-\alpha}\left( \frac{(x_j^1)^{\alpha}\sum_{i=1}^my_i(x^1)A_{ij}}{w_j} - 1\right) = \Theta\Bigg( \sum_{\{j\in S^-:x_j^0 \geq \frac{\delta_j}{1-\beta_2} \}} w_j {(x_j^0)}^{1-\alpha}\left( \frac{(x_j^1)^{\alpha}\sum_{i=1}^my_i(x^1)A_{ij}}{w_j} - 1\right)\Bigg)
\end{equation*}
is implied by the proof of part 3 of this lemma (see below). 
For each $j\in S^-$, we have that:
\begin{equation*}
\left(\frac{(x_j^1)^{\alpha}\sum_{i=1}^my_i(x^1)A_{ij}}{w_j} - 1\right) \geq (1+\gamma)(1-\gamma/4)-1 > \gamma/2,
\end{equation*}
Therefore:
\begin{align*}
\Phi(x^1) - \Phi(x^0) &\geq \Omega(\gamma) \sum_{j\in S^-} w_j \frac{\beta_2 x_j^0}{(1-\beta_2)(x_j^0)^{\alpha}}\\
&= \Omega\left(\frac{\beta_2\gamma}{1-\beta_2}\right)\sum_{j\in S^-} w_j {(x_j^0)}^{1-\alpha}.
\end{align*}

\noindent\textbf{Proof of 2.} Let $S^+$ denote the set of $j$'s such that $x_j$ increases in the current round. Then, recalling that for $j\in S^+$ $\frac{(x_j^0)^{\alpha}\sum_{i=1}^my_i(x^0)A_{ij}}{w_j}\leq 1- \gamma$ and that from the choice of parameters $\frac{(x_j^1)^{\alpha}\sum_{i=1}^my_i(x^1)A_{ij}}{w_j} \leq (1+\gamma/4)\frac{(x_j^0)^{\alpha}\sum_{i=1}^my_i(x^0)A_{ij}}{w_j}$:
\begin{align}
\Phi(x^1) - \Phi(x^0) &\geq \sum_{j=1}^n w_j \frac{x_j^1 - x_j^0}{(x_j^1)^{\alpha}}\left(1 - \frac{(x_j^1)^{\alpha}\sum_{i=1}^my_i(x^1)A_{ij}}{w_j} \right)\notag\\
&\geq \sum_{j\in S^+} w_j \frac{x_j^1 - x_j^0}{(x_j^1)^{\alpha}}\left(1 - \frac{(x_j^1)^{\alpha}\sum_{i=1}^my_i(x^1)A_{ij}}{w_j} \right)\notag\\
&\geq \sum_{j\in S^+} w_j \frac{x_j^1 - x_j^0}{(x_j^1)^{\alpha}}\left(1 - (1+\gamma/4)\frac{(x_j^0)^{\alpha}\sum_{i=1}^my_i(x^0)A_{ij}}{w_j} \right)\notag\\
&\geq \sum_{j\in S^+} w_j \frac{x_j^1 - x_j^0}{(x_j^1)^{\alpha}}\left((1 - \gamma) - \frac{(x_j^0)^{\alpha}\sum_{i=1}^my_i(x^0)A_{ij}}{w_j} \right). \notag%\label{eq:mul-increase-x}
\end{align}
Since $j\in S^+$, $x_j^1 = (1+\beta)x_j^0$, it follows that 
\begin{equation*}\Phi(x^1) - \Phi(x^0) \geq \frac{\beta}{(1+\beta)^{\alpha}}\sum_{j\in S^+} w_j (x_j^0)^{-\alpha}\left((1 - \gamma) - \frac{(x_j^0)^{\alpha}\sum_{i=1}^my_i(x^0)A_{ij}}{w_j} \right).
\end{equation*}
Observing that for any $x_j\notin S^+$ we have that $(1 - \gamma) - \frac{(x_j^0)^{\alpha}\sum_{i=1}^my_i(x^0)A_{ij}}{w_j}<0$, we get:
\begin{align*}
\Phi(x^1) - \Phi(x^0) &\geq \frac{\beta}{(1+\beta)^{\alpha}}\sum_{j=1}^n w_j (x_j^0)^{1-\alpha}\left((1 - \gamma) - \frac{(x_j^0)^{\alpha}\sum_{i=1}^my_i(x^0)A_{ij}}{w_j} \right)\\
&= \Omega(\beta)\left((1-\gamma)\sum_{j=1}^n w_j {(x_j^0)^{1-\alpha}} - \sum_{j=1}^nx_j^0\sum_{i=1}^m y_i(x^0) A_{ij} \right).
\end{align*}

\noindent\textbf{Proof of 3.} Let $S^-$ denote the set of $j$'s such that $x_j$ decreases in the current round. In this case not all the $x_j$'s with $j\in S^-$ decrease by a multiplicative factor $(1-\beta_2)$, since for $j\in S^-$: $x_j^1 = \max\{(1-\beta_2)x_j^0, \delta_j\}$. We will first lower-bound the potential increase over $x_j$'s that decrease multiplicatively: $\{j: j\in S^- \wedge x_j^0(1-\beta_2)\geq \delta_j\}$, so that $x_j^1 = x_j^0(1-\beta_2)$. Recall that for $j\in S^-$: $\frac{(x_j^0)^{\alpha}\sum_{i=1}^my_i(x^0)A_{ij}}{w_j}\geq 1+ \gamma$ and $\frac{(x_j^1)^{\alpha}\sum_{i=1}^my_i(x^1)A_{ij}}{w_j} \geq (1-\gamma/4\frac{\beta}{\ln(1/\delta_{\min})})\frac{(x_j^0)^{\alpha}\sum_{i=1}^my_i(x^0)A_{ij}}{w_j}\geq(1-\gamma/4)\frac{(x_j^0)^{\alpha}\sum_{i=1}^my_i(x^0)A_{ij}}{w_j}$. It follows that:
\begin{align}
\Phi(x^1) - \Phi(x^0) &\geq \frac{\beta_2}{(1-\beta_2)^{\alpha}}\sum_{\{j: j\in S^- \wedge x_j^0(1-\beta)\geq \delta_j\}} w_j (x_j^0)^{1-\alpha}\left(\frac{(x_j^1)^{\alpha}\sum_{i=1}^my_i(x^1)A_{ij}}{w_j} - 1\right)\notag\\
&\geq {\beta_2}\sum_{\{j: j\in S^- \wedge x_j^0(1-\beta_2)\geq \delta_j\}} w_j (x_j^0)^{1-\alpha}\left((1-\gamma/4)\frac{(x_j^0)^{\alpha}\sum_{i=1}^my_i(x^0)A_{ij}}{w_j}- 1\right)\notag\\
&=\Omega\Big(\frac{\beta^2}{\ln(1/\delta_{\min})}\Big)\sum_{\{j: j\in S^- \wedge x_j^0(1-\beta_2)\geq \delta_j\}} w_j (x_j^0)^{1-\alpha}\left(\frac{(x_j^0)^{\alpha}\sum_{i=1}^my_i(x^0)A_{ij}}{w_j}- (1+\gamma)\right).\label{eq:pot-increase-bound-multiplicative}
\end{align}
Next, we prove that the potential increase due to decrease of $x_j$ such that $\{j: j\in S^- \wedge x_j^0(1-\beta_2)< \delta_j\}$ is dominated by the potential increase due to $x_k$'s that decrease multiplicatively by the factor $(1-\beta_2)$. 

Choose any $x_j$ such that $\{j: j\in S^- \wedge x_j^0(1-\beta_2)< \delta_j\}$, and let $\xi_j(x^0) = \frac{(x_j^0)^{\alpha}\sum_{l=1}^mA_{lj}y_i(x^0)}{w_j}$. From Lemma \ref{lemma:small-x-tight-yi}, there exists at least one $i$ with $A_{ij}\neq 0$, such that:
\begin{equation}
y_i \geq \frac{w_j (x_j^0)^\alpha}{w_j (x_j^0)^\alpha}\cdot \frac{\sum_{i=1}^m y_i(x^0) A_{ij}}{mA_{\max}} > \frac{1}{mA_{\max}}\frac{w_j(1-\beta_2)^\alpha}{{\delta_j}^{\alpha}}\xi_j(x^0)\geq \frac{1-\beta_2}{mA_{\max}}\frac{w_j}{{\delta_j}^{\alpha}}\xi_j(x^0), \label{eq:yi-lower-bound-zj} \quad\text{ and},
\end{equation}
\begin{equation}
\sum_{k=1}^n A_{ik}x_k^0 > 1-\frac{\varepsilon}{2} \label{eq:x-tight-constraint}.
\end{equation}
From (\ref{eq:x-tight-constraint}), there exists at least one $p$ such that $A_{ip}\neq 0$ and \begin{equation}
A_{ip}x_p^0>\frac{1-\frac{\varepsilon}{2}}{n}. \label{eq:large-xp}
\end{equation}
Since $x_p^0 \in (0, 1]$ and $\alpha \in (0, 1)$, using (\ref{eq:large-xp}), we have that $A_{ip} (x_p^0)^{\alpha}\geq A_{ip}x_p^0 > \frac{1-\frac{\varepsilon}{2}}{n}$. Recalling (\ref{eq:yi-lower-bound-zj}):
\begin{align}
(x_p^0)^{\alpha}\sum_{l=1}^m A_{lp}y_l(x^0) &\geq (x_p^0)^{\alpha}A_{ip}y_i(x^0)\notag\\
&\geq \frac{1-\frac{\varepsilon}{2}}{n}\cdot \frac{1-\beta_2}{mA_{\max}}\frac{w_j}{{\delta_j}^{\alpha}}\xi_j(x^0)\notag.
\end{align}
Recalling that $\frac{w_j}{{\delta_j}^{\alpha}} = C \geq 2w_{\max}n^2 m A_{\max}$, it further follows that:
\begin{align}
(x_p^0)^{\alpha}\sum_{l=1}^m A_{lp}y_l(x^0) &\geq 2\left(1-\frac{\varepsilon}{2}\right)(1-\beta_2)\cdot n \cdot w_{\max}\cdot\xi_j(x^0).\label{eq:alpha<1--main-cond-for-small-var-inc}
\end{align}
Because $\varepsilon\leq \frac{1}{6}$ and $\beta_2 < \beta = \frac{\gamma}{5(\kappa+1)} = \frac{\varepsilon}{20(\kappa+1)}<\frac{\varepsilon}{20}$, it follows that $ 2\left(1-\frac{\varepsilon}{2}\right)(1-\beta_2)>1$. Therefore:
\begin{equation}
\frac{(x_p^0)^{\alpha}\sum_{l=1}^m A_{lp}y_l(x^0) }{w_p}\geq \frac{(x_p^0)^{\alpha}\sum_{l=1}^m A_{lp}y_l(x^0) }{w_{\max}} >n\cdot \xi_j(x^0) = n\cdot \frac{(x_j^0)^{\alpha}\sum_{l=1}^mA_{lj}y_i(x^0)}{w_j}. \label{eq:alpha<1--small-var-inc-cond1}
\end{equation}

As $\alpha<1$, we have that ${\delta_j}^{\alpha}>\delta_j$, and $\frac{w_j}{\delta_j}>\frac{w_j}{{\delta_j}^\alpha}=C$. Similar to (\ref{eq:yi-lower-bound-zj}), we can lower-bound $y_i$ as:
\begin{equation}
y_i(x) \geq \frac{1-\beta_2}{mA_{\max}}\cdot\frac{w_j}{\delta_j}\cdot\frac{x_j^0\sum_i y_i(x)A_{ij}}{w_j} > \frac{1-\beta_2}{mA_{\max}}\cdot\frac{w_j}{{\delta_j}^{\alpha}}\cdot\frac{x_j^0\sum_i y_i(x)A_{ij}}{w_j}.\label{eq:yi-lower-bound-zj-2}
\end{equation}
Then, recalling $A_{ip}x_p^0 > \frac{1-\frac{\varepsilon}{2}}{n}$, and using (\ref{eq:yi-lower-bound-zj-2}), it is simple to show that:
\begin{equation}
x_p^0\sum_l y_l(x^0)A_{lp} >  n\cdot x_j^0\sum_{l=1}^mA_{lj}y_l(x^0). \label{eq:alpha<1--small-var-inc-cond2}
\end{equation}
As $\xi_j(x^0) \geq (1+\gamma)$ and $x_p^0 > \frac{\delta_p}{1-\beta_2}$, it immediately follows from (\ref{eq:alpha<1--small-var-inc-cond1}) that $x_p$ decreases by a factor $(1-\beta_2)$. 

In the rest of the proof we show that (\ref{eq:alpha<1--small-var-inc-cond1}) and (\ref{eq:alpha<1--small-var-inc-cond2}) imply that the increase in the potential due to the decrease of variable $x_p$ dominates the increase in the potential due to the decrease of variable $x_j$ by at least a factor $n$. This result then further implies that the increase in the potential due to the decrease of variable $x_p$ dominates the increase in the potential due to the decrease of \emph{all} small $x_k$'s that appear in the constraint $i$ ($x_k$'s are such that $A_{ik}\neq 0$, $x_k^0 < \frac{\delta_k}{1-\beta_2}$, and $\frac{(x_k^0)^{\alpha}\sum_l y_l(x)A_{lk}}{w_k}\geq 1+\gamma$).   

Consider the following two cases: $w_p (x_p^0)^{1-\alpha} \geq (w_j x_j^0)^{1-\alpha}$ and $w_p (x_p^0)^{1-\alpha} < (w_j x_j^0)^{1-\alpha}$.

\noindent\textbf{Case 1: $w_p (x_p^0)^{1-\alpha} \geq (w_j x_j^0)^{1-\alpha}$}. Then, using (\ref{eq:alpha<1--small-var-inc-cond1}):
\begin{align}
w_p (x_p^0)^{1-\alpha}\left(\frac{(x_p^0)^{\alpha}\sum_{l=1}^m A_{lp}y_l(x^0) }{w_p} -(1+\gamma)\right)&\geq (w_j x_j^0)^{1-\alpha}\left(\frac{(x_p^0)^{\alpha}\sum_{l=1}^m A_{lp}y_l(x^0) }{w_p} -(1+\gamma)\right)\notag\\
& \geq (w_j x_j^0)^{1-\alpha}\left(n\cdot\frac{(x_j^0)^{\alpha}\sum_{l=1}^m A_{lj}y_l(x^0) }{w_j} -(1+\gamma)\right)\notag\\
&\geq n\cdot (w_j x_j^0)^{1-\alpha}\left(\frac{(x_j^0)^{\alpha}\sum_{l=1}^m A_{lj}y_l(x^0) }{w_j} -(1+\gamma)\right). \label{eq:alpha<1-dominant-inc-1}
\end{align}
\noindent\textbf{Case 2: $w_p (x_p^0)^{1-\alpha} < (w_j x_j^0)^{1-\alpha}$}. Then, using (\ref{eq:alpha<1--small-var-inc-cond2}):
\begin{align}
w_p (x_p^0)^{1-\alpha}\left(\frac{(x_p^0)^{\alpha}\sum_{l=1}^m A_{lp}y_l(x^0) }{w_p} -(1+\gamma)\right)&= x_p^0\sum_{l=1}^m A_{lp}y_l(x^0) - (1+\gamma)w_p (x_p^0)^{1-\alpha}\notag\\
&\geq x_p^0\sum_{l=1}^m A_{lp}y_l(x^0) - (1+\gamma)w_j (x_j^0)^{1-\alpha}\notag\\
&\geq n\cdot x_j^0\sum_{l=1}^m A_{lj}y_l(x^0) - (1+\gamma)w_j (x_j^0)^{1-\alpha}\notag\\
&\geq n\cdot (w_j x_j^0)^{1-\alpha}\left(\frac{(x_j^0)^{\alpha}\sum_{l=1}^m A_{lj}y_l(x^0) }{w_j} -(1+\gamma)\right). \label{eq:alpha<1-dominant-inc-2}
\end{align}

Combining (\ref{eq:alpha<1-dominant-inc-1}) and (\ref{eq:alpha<1-dominant-inc-2}) with (\ref{eq:pot-increase-bound-multiplicative}), it follows that:
\begin{equation*}
\Phi(x^1) - \Phi(x^0) \geq \Omega(\beta_2)\sum_{j\in S^-} w_j (x_j^0)^{1-\alpha}\left(\frac{(x_j^0)^{\alpha}\sum_{i=1}^my_i(x^0)A_{ij}}{w_j}- (1+\gamma)\right).
\end{equation*}
Finally, since for $j\notin S^-$: $\left(\frac{(x_j^0)^{\alpha}\sum_{i=1}^my_i(x^0)A_{ij}}{w_j}- (1+\gamma)\right) < 0$:
\begin{align*}
\Phi(x^1) - \Phi(x^0) &\geq \Omega(\beta_2)\sum_{j=1}^n w_j (x_j^0)^{1-\alpha}\left(\frac{(x_j^0)^{\alpha}\sum_{i=1}^my_i(x^0)A_{ij}}{w_j}- (1+\gamma)\right)\\
&= \Omega\Big(\frac{\beta^2}{\ln(1/\delta_{\min})}\Big)\left(\sum_{j=1}^nx_j^0\sum_{i=1}^my_i(x^0)A_{ij} - (1+\gamma)\sum_{j=1}^n w_j (x_j^0)^{1-\alpha}\right),
\end{align*}
completing the proof.
\end{proof}
Parts 2 and 3 of Lemma \cite{AwerbuchKhandekar2009} appear in a somewhat similar form in \cite{AwerbuchKhandekar2009}. However, part 3 requires significant additional results for bounding the potential change due to decrease of small $x_j$'s (i.e., $x_j$'s that are smaller than $\frac{\delta_j}{1-\beta}$) that were not needed in \cite{AwerbuchKhandekar2009}. The rest of the results in this paper are new.    

Consider the following definition of a stationary round:

\begin{definition}\label{def:alpha<1-stationary-round}
(Stationary round.) Let $S^- = \{j: x_j \text{ decreases}\}$. A round is stationary if it happens after the initial $\tau_0 + \tau_1$ rounds, where $\tau_0 = \frac{1}{\beta}\ln(\frac{1}{\delta_{\min}})$ and $\tau_1 = \frac{1}{\beta_2}\ln(nA_{\max})$, and both of the following two conditions hold:
\begin{enumerate}
\item $\sum_{j\in S^-} w_j {x_j}^{1-\alpha}\leq \gamma \sum_{j=1}^n w_j {x_j}^{1-\alpha}$, and
\item $\sum_{j=1}^n x_j \sum_{i=1}^m y_i(x)A_{ij} \leq (1+5\gamma/4)\sum_{j=1}^n w_j {x_j}^{1-\alpha}$.
\end{enumerate}
\end{definition}

In the rest of the proof, we first show that in any stationary round, we have an $O(\varepsilon)-$approximate solution, while in any non-stationary round, the potential function increases substantially.

We first prove the following lemma, which we will then be used in bounding the duality gap.
\begin{lemma}\label{lemma:alpha<1-lower-bound-xi-j}
After the initial $\tau_0 + \tau_1$ rounds, where $\tau_0 = \frac{1}{\beta}\ln(\frac{1}{\delta_{\min}})$ and $\tau_1 = \frac{1}{\beta_2}\ln(nA_{\max})$, in each round of the algorithm: $\xi_j(x)\equiv \frac{{x_j}^{\alpha}\sum_iy_i(x)A_{ij}}{w_j} > 1-\frac{5\gamma}{4}$, $\forall j$.
\end{lemma}
\begin{proof}
Suppose without loss of generality that the algorithm starts with a feasible solution. This assumption is w.l.o.g. because, from Lemma \ref{lemma:self-stabilization}, after at most $\tau_1$ rounds the algorithm reaches a feasible solution, and from Lemma \ref{lemma:feasibility}, once the algorithm reaches a feasible solution, it always maintains a feasible solution. 

Choose any $j$. Using the same argument as in the proof of Lemma \ref{lemma:feasibility}, after at most $\frac{1}{\beta}\ln(\frac{1}{\delta_j})\leq \tau_0$ rounds, there exists at least one round in which $\xi_j(x) > 1- \gamma$ (otherwise $x_j>1$, which is a contradiction).

Observe that in any round for which $\xi_j(x) \leq 1- \gamma$, $x_j$ increases by a factor $1+\beta_1 = 1+\beta$. Therefore, the maximum number of consecutive rounds in which $\xi_j(x) \leq 1-\gamma$ is at most $\frac{1}{\beta}\ln(\frac{1}{\delta_j})\leq \tau_0$, otherwise $x_j$ would increase to a value larger than 1, making $x$ infeasible, which is a contradiction due to Lemma \ref{lemma:feasibility}. The maximum amount by which $\xi_j(x)$ can decrease in any round is bounded by a factor $1-\frac{\gamma}{4}\cdot \frac{\beta}{\ln(1/\delta_{\min})} = 1-\frac{\gamma}{4}\cdot \frac{1}{\tau_0}$. Therefore, using the generalized Bernoulli's inequality, it follows that in any round:
\begin{equation*}
\xi_j(x) \geq (1-\gamma)\cdot \Big(1-\frac{\gamma}{4}\cdot \frac{1}{\tau_0}\Big)^{\tau_0}\geq (1-\gamma)\cdot \Big(1-\frac{\gamma}{4}\Big)>1 - \frac{5\gamma}{4}.
\end{equation*}
\end{proof}
A simple corollary of Lemma \ref{lemma:alpha<1-lower-bound-xi-j} is that:
\begin{corollary}\label{cor:alpha<1-lower-comp}
After the initial $\tau_0 + \tau_1$ rounds, where $\tau_0 = \frac{1}{\beta}\ln(\frac{1}{\delta_{\min}})$ and $\tau_1 = \frac{1}{\beta_2}\ln(nA_{\max})$, in each round of the algorithm: $\sum_j x_j \sum_i y_i(x)A_{ij} > \big(1 - \frac{5\gamma}{4}\big)\sum_j w_j {x_j}^{1-\alpha}$.
\end{corollary}
\begin{proof}
From Lemma \ref{lemma:alpha<1-lower-bound-xi-j}, after the initial $\tau_0 + \tau_1$ rounds, it always holds $\xi_j(x) \equiv \frac{{x_j}^{\alpha}\sum_i y_i(x)A_{ij}}{w_j}\geq 1-\frac{5\gamma}{4}$, $\forall j$. Multiplying both sides of the inequality by $w_j {x_j}^{1-\alpha}$, $\forall j$ and summing over $j$, the result follows.
\end{proof}

Recall that $p_{\alpha}(x)\equiv\sum_j w_j f_{\alpha}(x_j)$ denotes the primal objective. The following lemma states that any stationary round holds an $(1+6\varepsilon)$-approximate solution.
\begin{lemma}\label{lemma:alpha<1-stationary-round}
In any stationary round: $p(x^*)\leq (1+6\varepsilon)p(x)$, where $x^*$ is the optimal solution to $(P_{\alpha})$.
\end{lemma}
\begin{proof}
Since, by definition, a stationary round can only happen after the initial $\tau_0 + \tau_1$ rounds, we have that $x$ in that round is feasible, and also from Lemma \ref{lemma:approx-comp-slack}: $\sum_i y_i \leq (1+3\varepsilon)\sum_j x_j \sum_i y_i(x)A_{ij}$. Therefore, recalling Eq. (\ref{eq:duality-gap-alpha}) for the duality gap and denoting $\xi_j(x) = \frac{{x_j}^{\alpha}\sum_i y_i(x)A_{ij}}{w_j}$, we have that:
\begin{align}
p(x^*)-p(x)\leq G(x, y(x)) &= \sum_j w_j \frac{{x_j}^{1-\alpha}}{1-\alpha}\left(\xi_j^{-\frac{1-\alpha}{\alpha}}-1\right)+\sum_i y_i(x) - \sum_j w_j {x_j}^{1-\alpha}\xi_j^{-\frac{1-\alpha}{\alpha}}\notag\\
&=\sum_j w_j \frac{{x_j}^{1-\alpha}}{1-\alpha}\left(\alpha \xi_j^{-\frac{1-\alpha}{\alpha}}-1\right) + \sum_i y_i(x)\notag\\
&\leq \sum_j w_j \frac{{x_j}^{1-\alpha}}{1-\alpha}\left(\alpha \xi_j^{-\frac{1-\alpha}{\alpha}}-1\right) + (1+3\varepsilon)\sum_j x_j \sum_i y_i(x)A_{ij}. \label{eq:alpha<1-duality-gap-bound}
\end{align}

From Lemma \ref{lemma:alpha<1-lower-bound-xi-j}, $\xi_j > 1-\frac{5\gamma}{4}$, $\forall j$. Partition the indices of all the variables as follows:
\begin{equation*}
S_1 = \left\{j: \xi_j \in \left(1-\frac{5\gamma}{4}, 1+\frac{5\gamma}{4}\right)\right\}, \quad S_2 = \left\{j: \xi_j \geq 1+\frac{5\gamma}{4}\right\}.
\end{equation*}
Then, using (\ref{eq:alpha<1-duality-gap-bound}):
\begin{equation*}
p(x^*)-p(x)\leq G_1(x) + G_2(x),
\end{equation*}
where:
\begin{equation*}
G_1(x) = \sum_{j\in S_1} w_j \frac{{x_j}^{1-\alpha}}{1-\alpha}\left(\alpha \xi_j^{-\frac{1-\alpha}{\alpha}}-1\right) + (1+3\varepsilon)\sum_{j\in S_1}  x_j \sum_i y_i(x)A_{ij}
\end{equation*}
and 
\begin{equation*}
G_2(x) = \sum_{j\in S_2} w_j \frac{{x_j}^{1-\alpha}}{1-\alpha}\left(\alpha \xi_j^{-\frac{1-\alpha}{\alpha}}-1\right) + (1+3\varepsilon)\sum_{j\in S_2}  x_j \sum_i y_i(x)A_{ij}.
\end{equation*}
The rest of the proof follows by upper-bounding $G_1(x)$ and $G_2(x)$.

\noindent\textbf{Bounding $G_1(x)$.} Observing that $\forall j$: $x_j\sum_i y_i(x)A_{ij} = w_j {x_j}^{1-\alpha}\xi_j$, we can write $G_1(x)$ as:
\begin{align}
G_1(x) = \sum_{j\in S_1} w_j \frac{{x_j}^{1-\alpha}}{1-\alpha}\left(\alpha \xi_j^{-\frac{1-\alpha}{\alpha}} + (1+3\varepsilon)(1-\alpha)\xi_j-1\right).\label{eq:alpha<1-G-1}
\end{align}
Denote $r(\xi_j) = \alpha \xi_j^{-\frac{1-\alpha}{\alpha}} + (1+3\varepsilon)(1-\alpha)\xi_j-1$. It is simple to verify that $r(\xi_j)$ is a convex function. Since $\xi_j \in \left(1-\frac{5\gamma}{4}, 1+\frac{5\gamma}{4}\right)$, $\forall j\in S_1$, it follows that $r(\xi_j) < \max \{r(1-5\gamma/4), r(1+5\gamma/4)\}$. Now:
\begin{align}
r(1-5\gamma/4) &= \alpha \Big(1-\frac{5\gamma}{4}\Big)^{-\frac{1-\alpha}{\alpha}} + (1-\alpha)(1+3\varepsilon)\Big(1-\frac{5\gamma}{4}\Big)-1\notag\\
&< \alpha \Big(1-\frac{5\gamma}{4}\Big)^{-\frac{1-\alpha}{\alpha}} + (1-\alpha)(1+3\varepsilon) - 1.\notag
\end{align}
If $\frac{1-\alpha}{\alpha}\leq 1$, then as $(1-5\gamma/4)^{-1}\leq (1+2\gamma)$, it follows that $(1-5\gamma/4)^{-\frac{1-\alpha}{\alpha}}\leq 1+2\gamma$. Therefore:
\begin{align}
r(1-5\gamma/4)&< \alpha(1+2\gamma) + (1-\alpha)(1+3\varepsilon) - 1\notag\\
&= 2\gamma \alpha + 3\cdot(1-\alpha)\varepsilon = \alpha \frac{\varepsilon}{2} + 3\cdot(1-\alpha)\varepsilon\notag\\
&= 3\varepsilon\left(1-\frac{5}{6}\alpha\right).\label{eq:alpha<1-r-1}
\end{align}
If $\frac{1-\alpha}{\alpha}>1$, then (using generalized Bernoulli's inequality and $\varepsilon \leq \frac{\alpha}{1-\alpha}$):
\begin{align}
r(1-5\gamma/4) &< \alpha \frac{1}{(1-5\gamma/4)^{\frac{1-\alpha}{\alpha}}} + (1-\alpha)(1+3\varepsilon) - 1\notag\\
&\leq \alpha \frac{1}{1-\frac{5\gamma}{4}\cdot\frac{1-\alpha}{\alpha}} + (1-\alpha)(1+3\varepsilon) - 1\notag\\
&\leq \alpha \Big(1+\frac{5\gamma}{4}\cdot \frac{1-\alpha}{\alpha}\Big) + (1-\alpha)(1+3\varepsilon) - 1\notag\\
&\leq (1-\alpha)\Big(\frac{5\gamma}{4} + 3\varepsilon\Big)\notag\\
&< 4\varepsilon(1-\alpha).\label{eq:alpha<1-r-2}
\end{align}
On the other hand:
\begin{align}
r(1+5\gamma) &= \alpha \Big(1+\frac{5\gamma}{4}\Big)^{-\frac{1-\alpha}{\alpha}} + (1-\alpha)(1+3\varepsilon)\Big(1+\frac{5\gamma}{4}\Big)-1\notag\\ 
&< \alpha + (1-\alpha)(1+4\varepsilon) - 1\notag\\
&=4\varepsilon(1-\alpha). \label{eq:alpha<1-r-3}
\end{align}
Combining (\ref{eq:alpha<1-r-1})--(\ref{eq:alpha<1-r-3}) with (\ref{eq:alpha<1-G-1}):
\begin{equation}
G_1(x) < 4\varepsilon\cdot \sum_{j\in S_1} w_j \frac{{x_j}^{1-\alpha}}{1-\alpha}. \label{eq:alpha<1-G-1-1}
\end{equation}

\noindent\textbf{Bounding $G_2(x)$.} Because the round is stationary and $S_2 \subseteq S^-$, we have that: $\sum_{j\in S_2}w_j {x_j}^{1-\alpha} \leq \gamma \sum_{j=1}^n w_j {x_j}^{1-\alpha}$. Using the second part of the stationary round definition and that $\sum_{j\in S_2}x_j\sum_{i=1}^m y_i(x)A_{ij}>(1-5\gamma/4)\sum_{j\in S_2}w_j {x_j}^{1-\alpha}$ (follows from Lemma \ref{lemma:alpha<1-lower-bound-xi-j}):
\begin{align}
\sum_{j\in S_2} x_j \sum_{i=1}^m y_i(x)A_{ij} &= \sum_{k=1}^m x_k\sum_{i=1}^m y_i(x)A_{ik} - \sum_{l\notin S_2} x_l \sum_{l=1}^m y_l(x)A_{lk}\notag\\
&\leq (1+5\gamma/4)\sum_{k=1}^n w_k {x_k}^{1-\alpha} - (1-5\gamma/4) \sum_{l\notin S_2} w_l {x_l}^{1-\alpha}\notag\\
&\leq (1+5\gamma/4)\sum_{j\in S_2} w_j {x_j}^{1-\alpha} + \frac{5\gamma}{2}\sum_{l\notin S_2} w_l {x_l}^{1-\alpha}\notag\\
&\leq \gamma(1+5\gamma/4)\sum_{k=1}^n w_k {x_k}^{1-\alpha} + \frac{5\gamma}{2}\sum_{k=1}^n w_k {x_k}^{1-\alpha}\notag\\
&< 4\gamma \sum_{k=1}^n w_k {x_k}^{1-\alpha} = \varepsilon \sum_{k=1}^n w_k {x_k}^{1-\alpha}. \label{eq:alpha<1-xjy-bound}
\end{align}
Above, first inequality follows from $\sum_{k=1}^m x_k\sum_{i=1}^m y_i(x)A_{ik} \leq (1+5\gamma/4)\sum_{k=1}^n w_k {x_k}^{1-\alpha}$ (part 2 of the stationary round definition) and Corollary \ref{cor:alpha<1-lower-comp}. Second inequality follows by breaking the left summation into two summations: those with $j\in S_2$ and those with $l\notin S_2$. The third inequality follows from $S_2 \subseteq S$ and part 1 of the stationary round definition.

Observe that as $\xi_j \geq 1+5\gamma/4>1$, we have that $\xi_j^{-\frac{1-\alpha}{\alpha}}<1$. Using (\ref{eq:alpha<1-xjy-bound}), it follows that:
\begin{align}
G_2(x) & = \sum_{j\in S_2} w_j \frac{{x_j}^{1-\alpha}}{1-\alpha}\left(\alpha \xi_j^{-\frac{1-\alpha}{\alpha}}-1\right) + (1+3\varepsilon)\sum_{j\in S_2}  x_j \sum_i y_i(x)A_{ij}\notag\\
&< (\alpha-1) \sum_{j\in S_2} w_j \frac{{x_j}^{1-\alpha}}{1-\alpha} + (1+3\varepsilon)\sum_{j\in S_2}x_j\sum_i y_i(x)A_{ij}\notag\\
&\leq (\alpha-1) \sum_{j\in S_2} w_j \frac{{x_j}^{1-\alpha}}{1-\alpha} + \varepsilon (1+3\varepsilon)\sum_{k=1}^n w_k {x_k}^{1-\alpha}\notag\\
&\leq -(1-\alpha)\sum_{j\in S_2} w_j \frac{{x_j}^{1-\alpha}}{1-\alpha} + \frac{3}{2}\varepsilon (1-\alpha) \sum_{k=1}^n w_k \frac{{x_k}^{1-\alpha}}{1-\alpha}\notag\\
&< \frac{3}{2}\varepsilon (1-\alpha) \sum_{k=1}^n w_k \frac{{x_k}^{1-\alpha}}{1-\alpha}\notag\\
&< 2{\varepsilon}\sum_{k=1}^n w_k \frac{{x_k}^{1-\alpha}}{1-\alpha}. \label{eq:alpha<1-G-2}
\end{align}

Finally, combining (\ref{eq:alpha<1-G-1-1}) and (\ref{eq:alpha<1-G-2}):
\begin{align*}
p(x^*) - p(x) &< \Big(4\varepsilon + 2\varepsilon\Big) \sum_{j\in S_1} w_j \frac{{x_j}^{1-\alpha}}{1-\alpha}\\
&= 6\varepsilon p(x).
\end{align*}
\end{proof}

\begin{proofof}{Theorem \ref{thm:convergence-alpha<1}}
From Lemma \ref{lemma:alpha<1-stationary-round}, in any stationary round: $p(x^*) \leq p(x) (1 + 6\varepsilon)$. Therefore, to prove the theorem, it suffices to show that there are at most $O\left(\frac{1}{\alpha^2\varepsilon^5}\ln^2\left(\wratio{mnA_{\max}}\right)\ln^2\left(\wratio\frac{mnA_{\max}}{\varepsilon}\right)\right)$ non-stationary rounds in total, where $\wratio = {w_{\max}}/{ w_{\min}}$, because we can always run the algorithm for $\varepsilon' = \varepsilon/6$ to get an $\varepsilon-$approximation, and this would only affect the constant in the convergence time. 

To bound the number of non-stationary rounds, we will show that the potential increases by a ``large enough'' multiplicative value in all the non-stationary rounds in which the potential is not too ``small''. For the non-stationary rounds in which the value of the potential is ``small'', we show that the potential increases by a large enough value so that there can be only few such rounds. 

In the rest of the proof, we assume that the initial $\tau_0 + \tau_1 $ rounds have passed, so that $x$ is feasible, and the statement of Lemma \ref{lemma:approx-comp-slack} holds. This does not affect the overall bound on the convergence time, as 
\begin{align}
\tau_0 + \tau_1 &= \frac{1}{\beta}\ln\Big(\frac{1}{\delta_{\min}}\Big) + \frac{1}{\beta_2}\ln(nA_{max}) = O\left( \frac{1}{\beta^2}\ln(nA_{\max})\ln\Big(\frac{1}{\delta_{\min}}\Big)\right)\notag\\
&= O\left(\frac{1}{\alpha\varepsilon^4}\ln(nA_{\max})\ln^2\left(\wratio\frac{mnA_{\max}}{\varepsilon}\right)\ln\left(\wratio mnA_{\max}\right)\right).\label{eq:tau1+tau0-bound}
\end{align}

To bound the minimum and the maximum values of the potential $\Phi$, we will bound $\sum_j w_j \frac{{x_j}^{1-\alpha}}{1-\alpha}$ and $\frac{1}{\kappa}\sum_i y_i(x)$. Recall that $\Phi(x) = \sum_j w_j \frac{{x_j}^{1-\alpha}}{1-\alpha} - \frac{1}{\kappa}\sum_i y_i(x)$. 

Since $\delta_j = \left(\frac{w_{j}}{2w_{\max}n^2mA_{\max}}\right)^{\frac{1}{\alpha}}\geq \left(\frac{w_{\min}}{2w_{\max}n^2mA_{\max}}\right)^{\frac{1}{\alpha}}$, $x$ is always feasible, and $x_j\leq 1$, $\forall j$, we have that:
\begin{equation}
\frac{W}{1-\alpha}\cdot\left(\frac{w_{\min}}{2w_{\max}n^2mA_{\max}}\right)^{\frac{1-\alpha}{\alpha}}\leq\sum_j w_j \frac{{x_j}^{1-\alpha}}{1-\alpha}\leq \frac{W}{1-\alpha}, \label{eq:alpha<1-p-alpha-bound}
\end{equation}
and
\begin{equation}
0<\frac{1}{\kappa}\sum_i y_i(x)\leq \frac{Cm}{\kappa}. \label{eq:alpha<1-sum-yi-bound}
\end{equation}
Thus, we have:
\begin{align}
\Phi_{\min} &\geq -\frac{1}{\kappa}\sum_i y_i(x)\notag\\
&\geq -\frac{1}{\kappa}\cdot m\cdot C\notag\\
&\geq - O(m^2n^2 A_{\max}w_{\max}), \label{eq:alpha<1-phi-min}
\end{align}
and
\begin{equation}
\Phi_{\max} \leq \sum_{j=1}^n w_j \frac{1}{1-\alpha} = \frac{W}{1-\alpha}. \label{eq:alpha<1-phi-max}
\end{equation}

Recall from Lemma \ref{lemma:potential-increase} that the potential never decreases. We consider the following three cases for the value of the potential:

\noindent \textbf{Case 1: $\Phi_{\min}\leq \Phi \leq -\Theta(\frac{ w_{\min}}{A_{\max}})$}. Since in this case $\Phi < 0$, we have that $\sum_i y_i(x) > \kappa \sum_j w_j \frac{{x_j}^{1-\alpha}}{1-\alpha}$. From Lemma \ref{lemma:approx-comp-slack}, $\sum_j x_j \sum_i y_i(x)A_{ij}\geq (1-3\varepsilon)\sum_j w_j \frac{{x_j}^{1-\alpha}}{1-\alpha}$, thus implying: 
\begin{equation}\sum_j x_j \sum_i y_i(x)A_{ij}\geq \frac{1-3\varepsilon}{\kappa}\sum_j w_j \frac{{x_j}^{1-\alpha}}{1-\alpha}\geq 2\cdot\sum_j w_j \frac{{x_j}^{1-\alpha}}{1-\alpha}, \label{eq:alpha<1-neg-phi-inc}
\end{equation}
as $\kappa \geq \frac{1}{\varepsilon}$ and $\varepsilon\leq \frac{1}{6}$. Combining Part 3 of Lemma \ref{lemma:potential-increase-alpha<1} and (\ref{eq:alpha<1-neg-phi-inc}), the potential increases by at least:
\begin{align*}
\Omega\left(\frac{\beta^2}{\ln(1/\delta_{\min})}\right)\sum_jx_j \sum_i y_i(x) A_{ij} = \left(\frac{\beta^2}{\ln(1/\delta_{\min})}\right)\sum_i y_i(x) &= \left(\frac{\beta^2}{\ln(1/\delta_{\min})}\cdot\kappa\right)(-\Phi(x)) \\
&= \Omega\left(\frac{\gamma^2}{\kappa \ln(1/\delta_{\min})}\right)(-\Phi(x)).
\end{align*}
Since the potential never decreases, there can be at most 
\begin{align*}
O\left(\frac{\kappa\ln(1/\delta_{\min})}{\gamma^2}\ln\left(\frac{-\Phi_{\min}}{ w_{\min}/A_{\max}}\right)\right) = O\left(\frac{1}{\alpha}\frac{1}{\varepsilon^3}\ln^2\left(\wratio nmA_{\max}\right)\ln\left(\wratio\frac{nmA_{\max}}{\varepsilon}\right)\right)
\end{align*}
Case 1 rounds.

\noindent\textbf{Case 2: $-O\big(\frac{ w_{\min}}{A_{\max}}\big) < \Phi \leq O\Big(\frac{W}{1-\alpha}\cdot\big(\frac{w_{\min}}{2w_{\max}n^2mA_{\max}}\big)^{\frac{1-\alpha}{\alpha}}\Big)$.} From Lemma \ref{lemma:approx-comp-slack}, there exists at least one $i$ such that $\sum_j A_{ij}x_j \geq 1-(1+1/\kappa)\varepsilon$. Since $A_{ij}\leq A_{\max}$ $\forall i, j$, it is also true that $\sum_{j}x_j \geq \frac{1-(1+1/\kappa)\varepsilon}{A_{\max}}$, and as ${x_j}^{1-\alpha}\geq x_j$ and $\kappa \geq \frac{1}{\varepsilon}$, it follows that $\sum_{j}w_j{x_j}^{1-\alpha}\geq (1-\varepsilon(1+\varepsilon))\left(\frac{ w_{\min}}{A_{\max}}\right)$. From (\ref{eq:alpha<1-p-alpha-bound}), we also have $\sum_j w_j \frac{{x_j}^{1-\alpha}}{1-\alpha}\geq \frac{W}{1-\alpha}\cdot\left(\frac{w_{\min}}{2w_{\max}n^2mA_{\max}}\right)^{\frac{1-\alpha}{\alpha}}$. Therefore:
\begin{equation}
\sum_j w_j \frac{{x_j}^{1-\alpha}}{1-\alpha} \geq \max\left\{(1-\varepsilon(1+\varepsilon))\frac{1}{1-\alpha}\cdot \frac{ w_{\min}}{A_{\max}},\; \frac{W}{1-\alpha}\cdot\left(\frac{w_{\min}}{2w_{\max}n^2mA_{\max}}\right)^{\frac{1-\alpha}{\alpha}}\right\}. \label{eq:alpha<1-p-alpha-case2-bound}
\end{equation}
If $\Phi \leq \frac{1}{10}\cdot\max\left\{(1-\varepsilon(1+\varepsilon))\frac{1}{1-\alpha}\cdot \frac{ w_{\min}}{A_{\max}},\; \frac{W}{1-\alpha}\cdot\left(\frac{w_{\min}}{2w_{\max}n^2mA_{\max}}\right)^{\frac{1-\alpha}{\alpha}}\right\}$, then 
\begin{align*}
\sum_i y_i(x)  &\geq  \frac{9}{10}\kappa\cdot\frac{1}{1-\alpha}\sum_{j}w_j{x_j}^{\alpha}\\
&\geq \frac{9}{10}\kappa\cdot \max\left\{(1-\varepsilon(1+\varepsilon))\frac{1}{1-\alpha}\frac{ w_{\min}}{A_{\max}},\; \frac{W}{1-\alpha}\cdot\left(\frac{w_{\min}}{2w_{\max}n^2mA_{\max}}\right)^{\frac{1-\alpha}{\alpha}}\right\}.
\end{align*}
From Lemma \ref{lemma:approx-comp-slack}, 
\begin{align*}
\sum_i y_i(x) \sum_j A_{ij}x_j&\geq (1-3\varepsilon)\sum_i y_i(x) \\
&\geq (1-3\varepsilon)\frac{9}{10}\kappa\cdot \max\left\{(1-\varepsilon(1+\varepsilon))\frac{1}{1-\alpha}\cdot \frac{ w_{\min}}{A_{\max}},\; \frac{W}{1-\alpha}\cdot\left(\frac{w_{\min}}{2w_{\max}n^2mA_{\max}}\right)^{\frac{1-\alpha}{\alpha}}\right\}.
\end{align*}
From the third part of Lemma \ref{lemma:potential-increase-alpha<1}, the potential increases additively by at least 
\begin{equation*}
\Omega\left(\frac{\beta^2\kappa}{\ln(1/\delta_{\min})}\right)\cdot \max\left\{\frac{1}{1-\alpha}\cdot \frac{ w_{\min}}{A_{\max}},\; \frac{W}{1-\alpha}\cdot\left(\frac{w_{\min}}{2w_{\max}n^2mA_{\max}}\right)^{\frac{1-\alpha}{\alpha}}\right\},
\end{equation*}
and, therefore, $\Phi = \Omega\left(\frac{W}{1-\alpha}\cdot\big(\frac{w_{\min}}{2w_{\max}n^2mA_{\max}}\big)^{\frac{1-\alpha}{\alpha}}\right)$ after at most
\begin{align*}
O\left(\frac{\ln(1/\delta_{\min})\kappa}{\gamma^2}\right) = O\left(\frac{1}{\alpha}\frac{1}{\varepsilon^3}\ln\left(\wratio nmA_{\max}\right)\ln\left(\wratio\frac{nmA_{\max}}{\varepsilon}\right) \right)
\end{align*}
rounds.

\noindent\textbf{Case 3: $\Omega\Big(\frac{W}{1-\alpha}\cdot\big(\frac{w_{\min}}{2w_{\max}n^2mA_{\max}}\big)^{\frac{1-\alpha}{\alpha}}\Big)\leq \Phi \leq \frac{W}{1-\alpha}$.} In this case, $\Phi = O\left(\sum_j w_j \frac{{x_j}^{1-\alpha}}{1-\alpha}\right)$. If the round is stationary, then from Lemma \ref{lemma:alpha<1-stationary-round}, $p(x^*)\leq (1+6\varepsilon)p(x)$. If the round is not stationary, then from Definition \ref{def:alpha<1-stationary-round}, either:
\begin{enumerate}
\item $\sum_{k\in S^-}w_k{x_k}^{1-\alpha} > \gamma \sum_{j=1}^n w_j {x_j}^{1-\alpha}$, or
\item $\sum_{j=1}^n x_j \sum_{i=1}^m y_i(x)A_{ij} > (1+\frac{5\gamma}{4}) \sum_{j=1}^n w_j {x_j}^{1-\alpha}$.
\end{enumerate}
If the former is true, then using the first part of Lemma \ref{lemma:potential-increase-alpha<1}, the potential increases by at least $\Omega\left(\frac{\beta^2\gamma}{\ln(1/\delta_{\min})}\right)\cdot \sum_j w_j {x_j}^{1-\alpha} = \Omega\left(\frac{\beta^2\gamma}{\ln(1/\delta_{\min})}\right)\cdot (1-\alpha)\Phi$. If the latter is true, from the third part of Lemma \ref{lemma:potential-increase-alpha<1}, the potential increases by at least $\Omega\left(\frac{\beta^2\gamma}{\ln(1/\delta_{\min})}\right)\cdot \sum_j w_j {x_j}^{1-\alpha} = \Omega\left(\frac{\beta^2\gamma}{\ln(1/\delta_{\min})}\right)\cdot (1-\alpha)\Phi$.
It follows that there are at most 
\begin{align*}
O&\left(\frac{1}{1-\alpha}\cdot\frac{\ln(1/\delta_{\min})}{\beta^2 \gamma}\ln\left(\frac{\frac{W}{1-\alpha}}{\frac{W}{1-\alpha}\cdot\big(\frac{w_{\min}}{2w_{\max}n^2mA_{\max}}\big)^{\frac{1-\alpha}{\alpha}}}\right)\right)\\ 
&= O\left(\frac{1}{\alpha^2}\frac{1}{\varepsilon^5}\ln^2\left(\wratio\cdot{mnA_{\max}}\right)\ln^2\left(\wratio\cdot\frac{mnA_{\max}}{\varepsilon}\right)\right)
\end{align*}
non-stationary Case 3 rounds.

Combining the three cases with the bound on $\tau_0+ \tau_1$ (\ref{eq:tau1+tau0-bound}), the total convergence time is at most:
\begin{align*}
O\left(\frac{1}{\alpha^2\varepsilon^5}\ln^2\left(\wratio{mnA_{\max}}\right)\ln^2\left(\wratio\cdot\frac{mnA_{\max}}{\varepsilon}\right)\right)
\end{align*}
rounds, as claimed.
\end{proofof}

%%%%%%%%%%%%%%%%%%%%%%%%%%%%%%%%%%%%%%%%%%%%%%%%%%%%%%%%%%%%%%%%%%%%%%%%%%%%%%%%%%%%%%%%%%%%
\subsection{Proof of Theorem \ref{thm:convergence-alpha=1}}\label{section:alpha=1}

The proof outline for the convergence of \textsc{$\alpha$-FairPSolver} in the $\alpha = 1$ case is as follows. First, we show that in any round it cannot be the case that only ``small'' $x_j$'s (i.e., $x_j$'s that are smaller than $\frac{\delta_j}{1-\beta}$) decrease. In fact, we show that the increase in the potential due to updates of ``small'' variables is dominated by the increase in the potential due to those variables that decrease multiplicatively by a factor $(1-\beta_2) = (1-\beta)$ (Lemmas \ref{lemma:mul-increase-prop} and \ref{lemma:potential-increase-proportional}). We then define a stationary round and show that: (i) in any non-stationary round the potential increases significantly, and (ii) in any stationary round, the solution $x$ at the beginning of the round provides an additive $5W\varepsilon$--approximation to the optimum objective value.

\begin{lemma}\label{lemma:mul-increase-prop}
Starting with a feasible solution, in any round of the algorithm: 
\begin{enumerate}
\item $\sum_{\{k\in S^- : x_k\geq \frac{\delta_k}{1-\beta}\}} x_k \sum_{i=1}^m y_i(x)A_{ik}\geq \frac{1}{2}\sum_{j\in S^-} x_j \sum_{i=1}^m y_i(x)A_{ij}$.
\item $\sum_{\{k\in S^- : x_k\geq \frac{\delta_k}{1-\beta}\}} \frac{x_k \sum_{i=1}^m y_i(x)A_{ik}}{w_k}\geq \frac{1}{2}\sum_{j\in S^-} \frac{x_j \sum_{i=1}^m y_i(x)A_{ij}}{w_j}$.
\end{enumerate}
\end{lemma}

\begin{proof}
Fix any round, and let $x^0, y(x^0)$ and $x^1, y(x^1)$ denote the values of $x, y$ at the beginning and at the end of the round, respectively. If for all $j\in S^-$ $x_j^0 \geq \frac{\delta_j}{1-\beta}$, there is nothing to prove. 

Suppose that there exists some $x_j^0<\frac{\delta_j}{1-\beta}$ that decreases. Then from Lemma \ref{lemma:small-x-tight-yi} there exists at least one $i\in\{1,...,m\}$ such that $A_{ij}\neq 0$, and: 
\begin{itemize}
\item $\sum_{k=1}^nA_{ik}x_k^0 > 1- \frac{\varepsilon}{2}$, and
\item $y_i(x) \geq \frac{\sum_{l=1}^my_l(x^0)A_{lj}}{mA_{\max}} > (1-\beta)\frac{w_j}{\delta_j}\frac{1}{mA_{\max}}\frac{x_j^0\sum_{l=1}^my_l(x^0)A_{lj}}{w_j}$.
\end{itemize}
Since $\sum_{k=1}^nA_{ik}x_k^0 > 1- \frac{\varepsilon}{2}$, there exists at least one $p$ such that $A_{ip}x_p^0 > \frac{1-\frac{\varepsilon}{2}}{n}$. 
Recalling that $C = \frac{w_j}{\delta_j}\geq 2w_{\max}n^2mA_{\max}$:

\begin{align}
(x_p^0)A_{ip}y_i(x^0)&>C\cdot \frac{(1-\beta)}{m {A_{\max}}}\cdot\frac{1-\frac{\varepsilon}{2}}{n}\cdot\frac{x_j^0\sum_{l=1}^my_l(x^0)A_{lj}}{w_j} \notag\\
&>2w_{\max}n^2m{A_{\max}}\cdot \frac{(1-\beta)}{m {A_{\max}}}\cdot\frac{1-\frac{\varepsilon}{2}}{n}\cdot \frac{x_j^0\sum_{l=1}^my_l(x^0)A_{lj}}{w_j}\notag\\
&\geq  2nw_{\max}(1-\beta)\left(1-\frac{\varepsilon}{2}\right)\cdot \frac{x_j^0\sum_{l=1}^m y_l(x^0)A_{lj}}{w_j}\notag\\
&\geq nw_{\max}\frac{x_j^0\sum_{l=1}^m y_l(x^0)A_{lj}}{w_j}. \label{eq:alpha=1-n-times-inc} 
\end{align}

Since $x_j$ decreases, it must be $\frac{x_j^0\sum_{l=1}^my_l(x^0)A_{lj}}{w_j}\geq 1+\gamma$. Using (\ref{eq:alpha=1-n-times-inc}):
\begin{equation*}
\frac{x_p^0\sum_{l=1}^m y_l(x^0)A_{lp}}{w_p}\geq \frac{(x_p^0)A_{ip}y_i(x^0)}{w_{\max}}\geq n \frac{x_j^0\sum_{l=1}^m y_l(x^0)A_{lj}}{w_j} \geq 1+\gamma,
\end{equation*}
and, therefore, $x_p$ decreases as well. 
Moreover, since (\ref{eq:alpha=1-n-times-inc}) implies
\begin{equation*}
x_p^0\sum_{l=1}^m y_l(x^0)A_{lp} \geq \sum_{\{j\in S^-: x_j<\frac{\delta_j}{1-\beta}\wedge A_{ij}\neq 0\}} \frac{w_{\max}}{w_j}x_j^0\sum_{l=1}^m y_l(x^0)A_{lj} \geq \sum_{\{j\in S^-: x_j<\frac{\delta_j}{1-\beta}\wedge A_{ij}\neq 0\}} x_j^0\sum_{l=1}^m y_l(x^0)A_{lj},
\end{equation*}
the proof of the first part of the lemma follows. The second part follows from (\ref{eq:alpha=1-n-times-inc}) as well, since:
\begin{align*}
\frac{x_p^0\sum_{l=1}^my_l(x^0)A_{lp}}{w_p} &\geq \frac{(x_p^0)A_{ip}y_i(x^0)}{w_{\max}}\\
&\geq n \frac{x_j^0\sum_{l=1}^m y_l(x^0)A_{lj}}{w_j},
\end{align*}
which, given that $x_j$ was chosen arbitrarily, implies:
\begin{align*}
\frac{x_p^0\sum_{l=1}^my_l(x^0)A_{lp}}{w_p} &\geq \sum_{\{j\in S^-: x_k<\frac{\delta_k}{1-\beta}\wedge A_{ik}\neq 0\}} \frac{x_k^0\sum_{l=1}^m y_l(x^0)A_{lj}}{w_k}.
\end{align*}
\end{proof}

\begin{lemma}\label{lemma:potential-increase-proportional}
Let $x^0, y(x^0)$ and $x^1, y(x^1)$ denote the values of $x, y$ at the beginning and at the end of any fixed round, respectively. If $x^0$ is feasible, then the potential increase in the round is at least:
\begin{enumerate}
\item \label{item:prop-1} $\Phi(x^1) - \Phi(x^0) \geq \Omega(\beta\gamma)\sum_{j\in S^+} w_j$;
\item \label{item:prop-2} $\Phi(x^1) - \Phi(x^0) \geq \Omega(\beta)\left((1-\gamma)W - \sum_{j=1}^nx_j^0\sum_{i=1}^m y_i(x^0) A_{ij}\right)$.
\item \label{item:prop-3} $\Phi(x^1) - \Phi(x^0) \geq \Omega(\beta)\left(\sum_{j=1}^nx_j^0\sum_{i=1}^m y_i(x^0) A_{ij} - (1+\gamma)W\right)$.
\end{enumerate}
\end{lemma}
\begin{proof}
$\quad$\\
\noindent\textbf{Proof of \ref{item:prop-1}:}
Recall that:
\begin{equation*}
\Phi(x^1) - \Phi(x^0) \geq \sum_{j=1}^n w_j \frac{x_j^1 - x_j^0}{x_j^1}\Big( 1- \frac{x_j^1\sum_{i=1}^m y_i(x^1)A_{ij}}{w_j}\Big) \geq \sum_{j\in S^+} w_j \frac{x_j^1 - x_j^0}{x_j^1}\Big( 1- \frac{x_j^1\sum_{i=1}^m y_i(x^1)A_{ij}}{w_j}\Big).
\end{equation*}
Let $\xi_j(x^1) = \frac{x_j^1\sum_{i=1}^m y_i(x^1)A_{ij}}{w_j}$, $\xi_j(x^0) = \frac{x_j^0\sum_{i=1}^m y_i(x^0)A_{ij}}{w_j}$.

If $j\in S^+$, then $x_j^1 = (1+\beta)x_j^0$ and $\xi_j(x^0) \leq 1-\gamma$. Since from the choice of parameters $\xi_j$ increases by at most a factor of $1+\gamma/4$, it follows that: $\xi_j(x^1) \leq (1-\gamma)(1+\gamma/4)\leq 1 -\frac{3}{4}\gamma$, which gives $1 - \xi_j(x^1) \geq \frac{3}{4}\gamma$. Therefore:
\begin{equation*}
\Phi(x^1) - \Phi(x^0) \geq \frac{\beta}{1+\beta}\cdot \frac{3}{4}\gamma \cdot \sum_{j\in S^+} w_j.
\end{equation*}

\noindent\textbf{Proof of \ref{item:prop-2}:} The proof is equivalent to the proof of the second part of Lemma \ref{lemma:potential-increase-alpha<1} and is omitted.
 
\noindent\textbf{Proof of \ref{item:prop-3}:} Using that for $j\in S^-$ we have that $\frac{x_j^0\sum_{i=1}^m y_i(x^0)A_{ij}}{w_j}\geq 1 + \gamma$ and $x_j^1 = \max\{(1-\beta)x_j^0, \delta_j\}$, we can lower bound the increase in the potential as:
\begin{align}
\Phi(x^1) - \Phi(x^0) &\geq \sum_{\{j\in S^-:x_j^0\geq \frac{\delta_j}{1-\beta}\}} w_j \frac{x_j^1 - x_j^0}{x_j^1}\Big(1 - \frac{x_j^1 \sum_{i=1}^m y_{i}(x^1)A_{ij}}{w_j}\Big)\notag\\
& = \frac{\beta}{1-\beta} \sum_{\{j\in S^-:x_j^0\geq \frac{\delta_j}{1-\beta}\}} w_j \Big(\frac{x_j^1 \sum_{i=1}^m y_i (x^1)A_{ij}}{w_j} - 1\Big)\notag\\
&\geq \frac{\beta}{1-\beta} \sum_{\{j\in S^-:x_j^0\geq \frac{\delta_j}{1-\beta}\}} w_j \Big((1-\gamma/4)\frac{x_j^0 \sum_{i=1}^m y_i(x^0)A_{ij}}{w_j} - 1\Big)\notag\\
&\geq \frac{\beta}{1-\beta}(1-\gamma/4) \sum_{\{j\in S^-:x_j^0\geq \frac{\delta_j}{1-\beta}\}} w_j\Big(\frac{x_j^0 \sum_{i=1}^m y_i(x^0)A_{ij}}{w_j} - (1+\gamma)\Big). \label{eq:delta-phi-j-in-S-}
\end{align}
Now consider $k\in S^-$ such that $x_k^0 < \frac{\delta_k}{1-\beta}$. From the proof of Lemma \ref{lemma:mul-increase-prop}, for each such $x_k$ there exists a constraint $i$ and a variable $x_p \geq \frac{\delta_p}{1-\beta}$ with $p\in S^-$ such that $A_{ik}\neq 0$, $A_{ip\neq 0}$, $x_p^0 \sum_l y_l(x^0)A_{lp}\geq n\cdot x_k^0 \sum_l y_l(x^0)A_{lk}$, and $\frac{x_p^0 \sum_l y_l(x^0)A_{lp}}{w_p}\geq n\cdot \frac{x_k^0 \sum_l y_l(x^0)A_{kp}}{w_k}$. If $w_k \leq w_p$ then
\begin{align}
w_p \Big( \frac{x_p^0 \sum_l y_l(x^0)A_{lp}}{w_p} - (1+\gamma) \Big)&\geq w_k  \Big(n\cdot \frac{x_k^0 \sum_l y_l(x^0)A_{kp}}{w_k} - (1+\gamma) \Big)\notag\\
&\geq n \cdot w_k  \Big( \frac{x_k^0 \sum_l y_l(x^0)A_{kp}}{w_k} - (1+\gamma) \Big)\notag.
\end{align}
On the other hand, if $w_k>w_p$, then:
\begin{align*}
w_p \Big( \frac{x_p^0 \sum_l y_l(x^0)A_{lp}}{w_p} - (1+\gamma) \Big)&= ( x_p^0 \sum_l y_l(x^0)A_{lp} - (1+\gamma)w_p )\\
&> n\cdot x_k^0 \sum_l y_l(x^0)A_{kp} - (1+\gamma)w_k\\
&\geq n\cdot w_k \Big( \frac{x_k^0 \sum_l y_l(x^0)A_{kp}}{w_k} - (1+\gamma) \Big).
\end{align*}
It follows from (\ref{eq:delta-phi-j-in-S-}) that:
\begin{align*}
\Phi(x^1) - \Phi(x^0) & \geq \frac{\beta}{1-\beta}\frac{1-\gamma/4}{2} \sum_{j\in S^-} w_j \Big(\frac{x_j^0 \sum_{i=1}^m y_i(x^0)A_{ij}}{w_j} - (1+\gamma)\Big). 
\end{align*}
Finally, since for $j\notin S^-$ we have that $\frac{x_j^0 \sum_{i=1}^m y_i(x^0)A_{ij}}{w_j}< 1 + \gamma$:
\begin{align*}
\Phi(x^1) - \Phi(x^0) & \geq \frac{\beta}{1-\beta}\frac{1-\gamma/4}{2} \sum_{j=1}^n w_j \Big(\frac{x_j^0 \sum_{i=1}^m y_i(x^0)A_{ij}}{w_j} - (1+\gamma)\Big)\\
&= \Omega(\beta) \Big( \sum_{j=1}^n x_j^0 \sum_{i=1}^m y_i(x^0)A_{ij} - (1+\gamma)\sum_{j=1}^n w_j \Big). 
\end{align*}
\end{proof}

Consider the following definition of a stationary round:

\begin{definition}\label{def:alpha=1-stat-round}
A round is stationary if it happens after the initial $\tau_0 + \tau_1$ rounds, where $\tau_0 = \frac{1}{\beta}\ln(1/\delta_{\min})$, $\tau_1 = \frac{1}{\beta}\ln(nA_{\max})$ and if both of the following conditions hold:
\begin{itemize}
\item $\sum_{j\in S^+}w_j \leq W /\tau_0$;
\item $(1-2\gamma)W\leq \sum_{j=1}^n x_j \sum_{i=1}^m y_(x)A_{ij}\leq (1+2\gamma)W$.
\end{itemize}
\end{definition}

We first show that in any non-stationary round there is a sufficient progress towards the $\varepsilon-$approximate solution.
\begin{lemma}\label{lemma:non-stat-round-alpha=1}
In any non-stationary round the potential function increases by at least $\Omega(\beta\gamma\cdot W/\tau_0)$.
\end{lemma}
\begin{proof}
A round is non-stationary if either of the two conditions from Definition \ref{def:alpha=1-stat-round} does not hold. If the first condition does not hold, then from the first part of Lemma \ref{lemma:potential-increase-proportional}, the potential increases by $\Omega(\beta\gamma\cdot W/\tau_0)$. If the second condition does not hold, then from either the second or the third part of Lemma \ref{lemma:potential-increase-proportional} the potential increases by at least $\Omega(\beta\gamma W) \geq \Omega(\beta\gamma\cdot W/\tau_0)$.
\end{proof}

Before proving that in every non-stationary round, the solution is $O(\varepsilon)-$approximate, we will need the following intermediary lemma.
\begin{lemma}\label{lemma:cond-lower-bound}
Starting with a feasible solution and after at most $\tau_0 = \frac{1}{\beta}\ln\left(\frac{1}{\delta_{\min}}\right)$ rounds, in any round of the algorithm:
\begin{equation*}
\min_j \frac{x_j \sum_{i=1}^m y_i(x)A_{ij}}{w_j} \geq (1-\gamma)^{\tau_0}.
\end{equation*}
\end{lemma}
\begin{proof}
First, we claim that after the algorithm reaches a feasible solution it takes at most $\tau_0 + 1$ additional rounds for each agent $j$ to reach a round in which $\frac{x_j \sum_{i=1}^m y_i(x)A_{ij}}{w_j} > 1-\gamma$. Suppose not, and pick any agent $k$ for which in each of the $\tau_0 + 1$ rounds following the first round that holds a feasible solution: $
\frac{x_k\sum_{i=1}^m y_i(x)A_{ik}}{w_k}\leq 1-\gamma$. Then $x_k$ increases in each of the rounds and after $\frac{1}{\beta}\ln(\frac{1}{\delta_k})\leq \tau_0$ rounds we have $x_k \geq 1$. Therefore, after at most $\tau_0+1$ rounds the solution becomes infeasible, which is a contradiction (due to Lemma \ref{lemma:feasibility}).

Now choose any $x_j$ and observe $\xi_j = \frac{x_j \sum_{i=1}^m y_i(x)A_{ij}}{w_j}$ over the rounds that happen after the first $\tau_0 + 1$ rounds. The maximum number of consecutive rounds for which $\xi_j\leq 1 -\gamma$ is $\tau_j = \frac{1}{\beta}\ln(\frac{1}{\delta_j})\leq \tau_0$, otherwise we would have $x_j>1$, a contradiction. Since in any round, due to the choice of the algorithm parameters, $\xi_j$ decreases by at most a factor of $1-\gamma/4$, the minimum value that $\xi_j$ can take is at least $(1-\gamma)(1-\gamma/4)^{\tau_j/2}>(1-\gamma)^{\tau_0}$, thus completing the proof.
\end{proof}

Now we are ready to prove that a solution in a stationary round is $O(\varepsilon)-$approximate.
\begin{lemma}\label{lemma:alpha=1-stat-round}
In any stationary round: $p_1(x^*) - p_1(x)\leq 5\varepsilon W$, where $x^*$ is the optimal solution.
\end{lemma}
\begin{proof}
Since, due to Definition \ref{def:alpha=1-stat-round}, a stationary round can only happen after the initial $\tau_0 + \tau_1$ rounds, we have that in any stationary round the solution is feasible (Lemmas \ref{lemma:feasibility} and \ref{lemma:self-stabilization}) and approximate complementary slackness (Lemma \ref{lemma:approx-comp-slack}) holds. 

Recall the expression \iffalse(\ref{eq:duality-gap-proportional})\fi for the duality gap:
\begin{equation*}
G_1(x, y) =  - \sum_{j=1}^n w_j \ln\left(\frac{x_j\sum_{i=1}^m y_i A_{ij}}{w_j}\right)+\sum_{i=1}^m y_i -W.
\end{equation*}
From the second part of Lemma \ref{lemma:approx-comp-slack}:
\begin{equation*}
\sum_{i=1}^m y_i\leq (1+3\varepsilon)\sum_{j=1}^nx_j\sum_{i=1}^m y_i A_{ij}.
\end{equation*}
Therefore:
\begin{equation*}
G_1(x, y) \leq  - \sum_{j=1}^n w_j \ln\left(\frac{x_j\sum_{i=1}^m y_i A_{ij}}{w_j}\right)+(1+3\varepsilon)\sum_{j=1}^nx_j\sum_{i=1}^m y_i A_{ij} -W.
\end{equation*}
Since the round is stationary, we have that $\sum_{j=1}^nx_j\sum_{i=1}^m y_i A_{ij}\leq (1+2\gamma)W$, which gives:
\begin{equation}
G_1(x, y) \leq  - \sum_{j=1}^n w_j \ln\left(\frac{x_j\sum_{i=1}^m y_i A_{ij}}{w_j}\right)+4\varepsilon W. \label{eq:alpha=1-duality-gap-bound}
\end{equation}
Let $\xi_j = \frac{x_j\sum_{i=1}^m y_i A_{ij}}{w_j}$. The remaining part of the proof is to bound $-\sum_{j=1}^n w_j\ln(\xi_j)\leq -\sum_{j:\xi_j < 1} w_j \ln(\xi_j)$. 
For $\xi_j \in(1-\gamma, 1)$, we have that $-w_j\ln(\xi_j)\leq \gamma w_j$. To bound the remaining terms, we will use Lemma \ref{lemma:cond-lower-bound} and the bound of the sum of the weights $w_j$ for which $\xi_j \in S^+$ (that is, $\xi_j \leq 1-\gamma$). It follows that:
\begin{align}
-\sum_{j=1}^n w_j\ln(\xi_j)&\leq -\sum_{k:\xi_k\in(1-\gamma, 1)}w_k \ln(\xi_k) - \sum_{l\in S^+}w_l\ln(\xi_l)\notag\\
&\leq \gamma \sum_{k:\xi_k\in(1-\gamma, 1)}w_k -\ln\Big((1-\gamma)^{\tau_0}\Big)\cdot \sum_{l\in S^+}w_l \quad(\text{from Lemma \ref{lemma:cond-lower-bound}})\notag\\
&\leq \gamma W + \tau_0 {\gamma} \cdot\frac{W}{\tau_0}\notag\\
&= 2\gamma W \notag\\
&= \frac{\varepsilon}{2}W. \label{eq:bound-on-ln-terms}
\end{align}
Combining (\ref{eq:alpha=1-duality-gap-bound}) and (\ref{eq:bound-on-ln-terms}), and recalling that $p_1(x^*)-p_1(x)\leq G_1(x, y(x))$, the result follows.
\end{proof}

\begin{proofof}{Theorem \ref{thm:convergence-alpha=1}}
Consider the values of the potential in the rounds following the initial $\tau_0 + \tau_1$ rounds, where $\tau_0 = \frac{1}{\beta}\ln(1/\delta_{\min})$, $\tau_1 = \frac{1}{\beta}\ln(nA_{\max})$ (so that the solution $x$ is feasible in each round and the approximate complementary slackness holds). Observe that $\tau_0+\tau_1 = o\Big(\frac{\ln^2\left(nmA_{\max}\wratio\right)\ln^2\left(\frac{nmA_{\max}}{\varepsilon}\wratio\right)}{\varepsilon^5}\Big)$.

We start by bounding the minimum and the maximum values that the potential can take. Recall (from Lemma \ref{lemma:potential-increase}) that the potential never decreases.

Due to Lemma \ref{lemma:feasibility}, $x_j\in [\delta_j, 1]$, $\forall j$, and therefore we can bound the two summations in the potential as:
\begin{equation}
\sum_j w_j \ln(x_j)\geq \sum_j w_j \ln(\delta_j) = - O\Big(W\cdot\ln\Big(\frac{w_{\max}}{ w_{\min}}nmA_{\max}\Big)\Big), \label{eq:alpha=1-pot-lbound-1}
\end{equation}
\begin{equation}
\sum_j w_j \ln(x_j)\leq \sum_j w_j \ln(1) \leq 0, \label{eq:alpha=1-pot-ubound-1}
\end{equation}
\begin{equation}
-\frac{1}{\kappa} \sum_i y_i(x) \geq -\frac{mC}{\kappa}\cdot e^{0} > -mC = -O(w_{\max}n^2m^2A_{\max}), \label{eq:alpha=1-pot-lbound-2}
\end{equation}
and
\begin{equation}
-\frac{1}{\kappa} \sum_i y_i(x) < -\frac{mC}{\kappa}\cdot e^{-\kappa} < 0. \label{eq:alpha=1-pot-ubound-2}
\end{equation}

From (\ref{eq:alpha=1-pot-lbound-1}) and (\ref{eq:alpha=1-pot-lbound-2}):
\begin{equation}
\Phi_{\min} \geq -O(w_{\max}n^2m^2A_{\max}). \label{eq:alpha=1-phi-min}
\end{equation}

On the other hand, from (\ref{eq:alpha=1-pot-ubound-1}) and (\ref{eq:alpha=1-pot-ubound-2}):
\begin{equation}
\Phi_{\max} <0. \label{eq:alpha=1-phi-max}
\end{equation}

Consider the following two cases:

\noindent\textbf{Case 1: $\frac{1}{\kappa}\sum_i y_i(x)\geq W\cdot \ln\Big(e\cdot\frac{w_{\max}}{ w_{\min}}nmA_{\max}\Big)$.}  Then $\frac{1}{\kappa}\sum_i y_i(x)\leq -\Phi(x)\leq \frac{2}{\kappa}\sum_i y_i(x)$ and $\frac{1}{\kappa}\sum_i y_i(x)\geq W$. From the third part of Lemma \ref{lemma:approx-comp-slack}, we have that $\sum_j x_j \sum_iy_i(x)A_{ij}\geq (1-3\varepsilon)\sum_i y_i(x) \geq 2W$. Thus using the Part \ref{item:prop-2} of Lemma \ref{lemma:potential-increase-proportional}, we get that the potential increases by 
\begin{equation*}
\Omega(\beta)\cdot \sum_j x_j \sum_i y_i(x)A_{ij} = \Omega\left(\beta\cdot \sum_i y_i(x)\right)= \Omega(\beta\kappa)\cdot (-\Phi(x)).
\end{equation*}
Finally, since $\beta\kappa = \Theta(\gamma)$, there can be at most $O\left(\frac{1}{\gamma}\ln\left(\frac{\wratio nmA_{\max}}{W\ln(\wratio nmA_{\max})}\right)\right)$ Case 1 rounds.

\noindent\textbf{Case 2: $\frac{1}{\kappa}\sum_i y_i(x) < W\cdot \ln\Big(e\cdot\frac{w_{\max}}{ w_{\min}}nmA_{\max}\Big)$.}  Then $-2W\cdot \ln\Big(e\cdot\frac{w_{\max}}{ w_{\min}}nmA_{\max}\Big) < \Phi(x) < 0$. From Lemma \ref{lemma:alpha=1-stat-round}, if a round is stationary, then $p(x^*) - p(x)\leq 5\varepsilon W$. If a round is non-stationary, from Lemma \ref{lemma:non-stat-round-alpha=1}, the potential increases (additively) by at least $\Omega(\beta\gamma\cdot W/\tau_0)$. Therefore, the maximum number of non-stationary rounds is at most:
\begin{align*}
O\left(\frac{W\ln(nm{A_{\max}}w_{\max}/ w_{\min})}{\beta\gamma W/\tau_0}\right) & = O\left(\frac{1}{\beta^2\gamma}\cdot \ln^2\left(\wratio nmA_{\max}\right)\right)\\
&= O\left(\frac{\ln^2\left(\wratio nmA_{\max}\right)\ln^2\left(\wratio\frac{nmA_{\max}}{\varepsilon}\right)}{\varepsilon^5}\right).
\end{align*}
Combining the results for the Case 1 and Case 2, the theorem follows by invoking \textsc{$\alpha$-FairPSolver} for the approximation parameter $\varepsilon' = \varepsilon/5$.
\end{proofof}

%%%%%%%%%%%%%%%%%%%%%%%%%%%%%%%%%%%%%%%%%%%%%%%%%%%%%%%%%%%%%%%%%%%%%%%%%%%%%%%%%%%%%%%%%%%%
\subsection{Proof of Theorem \ref{thm:convergence-alpha>1}}\label{section:alpha>1}
The outline of the proof of Theorem \ref{thm:convergence-alpha>1} is as follows. First, we show that in any round of the algorithm the variables that decrease by a multiplicative factor $(1-\beta_2)$ dominate the potential increase due to \emph{all the variables} that decrease (Lemma \ref{lemma:alpha>1-mul-increase-over-S-}). 
This result is then used in Lemma \ref{lemma:potential-increase-alpha>1} to show the appropriate lower bound on the potential increase. Observe that for $\alpha > 1$ the objective function $p_\alpha(x)$, and, consequently, the potential function $\Phi(x)$ is negative for any feasible $x$. To yield a poly-logarithmic convergence time in $\wratio, m, n$, and $A_{\max}$, the idea is to show that the negative potential $-\Phi(x)$ decreases by some multiplicative factor whenever $x$ is not a ``good'' approximation to $x^*$ -- the optimal solution to $(P_\alpha)$. This idea, combined with the fact that the potential never decreases (and therefore $-\Phi(x)$ never increases) and with upper and lower bounds on the potential then leads to the desired convergence time. 

\begin{lemma}\label{lemma:alpha>1-mul-increase-over-S-}
In any round of the algorithm in which the solution $x^0$ at the beginning of the round is feasible:
\begin{equation*}\sum_{\left\{j:j\in S^- \wedge x_j^0 \geq \frac{\delta_j}{1-\beta}\right\}} x_j^0\sum_{i=1}^m y_i(x^0)A_{ij}\geq \frac{1}{2} \sum_{j\in S^-}x_j^0\sum_{i=1}^m y_i(x^0)A_{ij};
\end{equation*}
and
\begin{align*}\sum_{\left\{j:j\in S^- \wedge x_j^0 \geq \frac{\delta_j}{1-\beta}\right\}} \left(x_j^0\sum_{i=1}^m y_i(x^0)A_{ij} - (1+\gamma)w_j(x_j^0)^{1-\alpha}\right)\geq \frac{1}{2} \sum_{j\in S^-}\left(x_j^0\sum_{i=1}^m y_i(x^0)A_{ij} - (1+\gamma)w_j(x_j^0)^{1-\alpha}\right).
\end{align*}
\end{lemma}
\begin{proof}
If $x_j^0 \geq \frac{\delta_j}{1-\beta}$, $\forall j$, there is nothing to prove, so assume that there exists at least one $j$ with $x_j^0 < \frac{\delta_j}{1-\beta}$. 
The proof proceeds as follows. First, we show that for each $j$ for which $x_j$ decreases by a factor less than $(1-\beta)$ there exists at least one $x_p$ that appears in at least one constraint $i$ in which $x_j$ appears and decreases by a factor $(1-\beta)$. We then proceed to show that $x_p$ is in fact such that 
\begin{align*}x_p^0\sum_{l=1}^m y_l(x^0)A_{lp} = \Omega(n)x_j^0\sum_{l=1}^m y_l(x^0)A_{lj}
\end{align*}
and 
\begin{align*}x_p^0\sum_{l=1}^m y_l(x^0)A_{lp} - (1+\gamma)w_p(x_p^0)^{1-\alpha}= \Omega(n) \left(x_j^0\sum_{l=1}^m y_l(x^0)A_{lj} - (1+\gamma)w_j(x_j^0)^{1-\alpha}\right).
\end{align*}
This will then imply that the terms $x_p^0\sum_{l=1}^m y_l(x^0)A_{lp}$ and $x_p^0\sum_{l=1}^m y_l(x^0)A_{lp} - (1+\gamma)w_p(x_p^0)^{1-\alpha}$ dominate \emph{the sum} of all the terms corresponding to $x_j$'s with $A_{ij}\neq 0$ and $x_j<\frac{\delta_j}{1-\beta}$, thus completing the proof.

From Lemma \ref{lemma:small-x-tight-yi}, for each $j\in S^-$ with $x_j < \frac{\delta_j}{1-\beta}$ there exists at least one constraint $i$ such that:
\begin{itemize}
\item $\sum_{k=1}^n A_{ik}x_k^0 > 1 - \frac{\varepsilon}{2}$, and
\item $y_i(x^0) \geq \frac{\sum_{l=1}^m y_l(x^0)A_{lj}}{mA_{\max}} \Rightarrow y_i(x^0) > (1-\beta)^\alpha\frac{1}{mA_{\max}}\frac{w_j}{{\delta_j}^{\alpha}}\frac{(x_j^0)^{\alpha}\sum_{l=1}^m y_i(x^0)A_{lj}}{w_j}$.
\end{itemize}
Therefore, there exists at least one $x_p$ with $A_{ip}\neq 0$ such that $A_{ip}x_p^0>\frac{1-\frac{\varepsilon}{2}}{n}$, which further gives $A_{ip}(x_p^0)^{\alpha}>\frac{(1-\frac{\varepsilon}{2})^{\alpha}}{n^\alpha}\cdot {A_{ip}}^{1-\alpha}\geq\frac{(1-\frac{\varepsilon}{2})^{\alpha}}{n^\alpha}\cdot {A_{\max}}^{1-\alpha}$, where the last inequality follows from $1\leq A_{ip}\leq A_{\max}$ and $\alpha>1$. Combining the inequality for $A_{ip}(x_p^0)^{\alpha}$ with the inequality for $y_i(x^0)$ above:

\begin{align}
(x_p^0)^{\alpha}\sum_{l=1}^m y_l(x^0)A_{lp} &\geq (x_p^0)^{\alpha}A_{ip}y_i(x^0)\notag\\
&\geq \frac{(1-\frac{\varepsilon}{2})^{\alpha}}{n^\alpha}\cdot {A_{\max}}^{1-\alpha} (1-\beta)^\alpha\frac{1}{mA_{\max}}\frac{w_j}{{\delta_j}^{\alpha}}\frac{(x_j^0)^{\alpha}\sum_{l=1}^m y_l(x^0)A_{lj}}{w_j}\notag\\
&= C\cdot \frac{(1-\frac{\varepsilon}{2})^{\alpha}}{n^\alpha m {A_{\max}}^{\alpha}}(1-\beta)^\alpha\frac{(x_j^0)^{\alpha}\sum_{l=1}^m y_l(x^0)A_{lj}}{w_j}\quad (\text{from } C= \frac{w_j}{\delta_j^{\alpha}})\notag\\
&\geq 2n w_{\max} (1-\beta)^\alpha\left(1-\frac{\varepsilon}{2}\right)^{\alpha}\frac{(x_j^0)^{\alpha}\sum_{l=1}^m y_l(x^0)A_{lj}}{w_j}\quad (\text{from } C\geq 2w_{\max} n^{\alpha+1}m {A_{\max}}^{2\alpha-1})\notag.
\end{align}

Using the generalized Bernoulli's inequality: $\left(1-\frac{\varepsilon}{2}\right)^{\alpha}>1-\frac{\varepsilon\alpha}{2}$ and $(1-\beta)^\alpha>(1-\beta\alpha)$ \cite{mitrinovic1970analytic}, and recalling that $\varepsilon\alpha\leq \frac{9}{10}$, $\beta \leq \frac{\gamma\varepsilon}{5} = \frac{\varepsilon^2}{20}\leq \frac{\varepsilon}{120}$, we further get:
\begin{align}
(x_p^0)^{\alpha}\sum_{l=1}^m y_l(x^0)A_{lp} &\geq 2nw_{\max}\left(1 - \frac{9}{10\cdot120}\right)\left(1 - \frac{9}{20}\right)\cdot \frac{(x_j^0)^{\alpha}\sum_{l=1}^m y_l(x^0)A_{lj}}{w_j}\notag\\
&\geq n w_{\max}\frac{(x_j^0)^{\alpha}\sum_{l=1}^m y_l(x^0)A_{lj}}{w_j}\notag, 
\end{align}
which further implies:
\begin{equation}
\frac{(x_p^0)^{\alpha}\sum_{l=1}^m y_l(x^0)A_{lp}}{w_p} \geq n\cdot \frac{(x_j^0)^{\alpha}\sum_{l=1}^m y_l(x^0)A_{lj}}{w_j},\label{eq:condition-n-dominance}
\end{equation}
as $w_p\leq w_{\max}$. Since $x_j$ decreases, $\frac{(x_j^0)^{\alpha}\sum_{l=1}^m y_l(x^0)A_{lj}}{w_j}\geq 1 + \gamma$, and therefore $x_p$ decreases as well.

Using similar arguments, as $A_{ip}x_p^0>\frac{1-\frac{\varepsilon}{2}}{n}$ and recalling that $y_i(x^0)\geq \frac{1}{m A_{\max}}\sum_{l=1}^mA_{lj}y_l(x^0) >\frac{1}{m A_{\max}}\frac{1-\beta}{\delta_j}\cdot x_j^0\sum_{l=1}^mA_{lj}y_l(x^0)$:

\begin{align}
x_p^0\sum_{l=1}^m y_l(x^0)A_{lp} &\geq x_p^0 A_{ip}y_i(x^0)\geq \frac{1-\frac{\varepsilon}{2}}{n}\frac{1}{m A_{\max}}\frac{1-\beta}{\delta_j}\cdot x_j^0\sum_{l=1}^mA_{lj}y_l(x^0)\notag\\
&\geq n x_j^0\sum_{l=1}^mA_{lj}y_l(x^0), \label{eq:slackness-n-dominance}
\end{align}
as $\delta_j \leq \frac{1}{2^{1/\alpha}n^2mA_{\max}}$ and $2^{1/\alpha}(1-\frac{\varepsilon}{2})(1-\beta)\geq 2^{\frac{10}{9}\varepsilon}(1-\frac{\varepsilon}{2})(1-\frac{\varepsilon^2}{20})\geq 1$ (since $\varepsilon \in (0, 1/6]$).

From (\ref{eq:slackness-n-dominance}), it follows that 
\begin{equation*}x_p^0\sum_{l=1}^m y_l(x^0)A_{lp} \geq \sum_{\{k \in S^-: x_k <\frac{\delta_k}{1-\beta}\wedge A_{ik}\neq 0\}}x_k^0\sum_{l=1}^m y_l(x^0)A_{lk},
\end{equation*}
which further implies the first part of the lemma.

For the second part, consider the following two cases:

\noindent\textbf{Case 1:} $w_p(x_p^0)^{1-\alpha} \geq w_j(x_j^0)^{1-\alpha}$. Then:
\begin{align*}
x_p^0\sum_{l=1}^m y_l(x^0)A_{lp} - (1+\gamma)w_p(x_p^0)^{1-\alpha} &= w_p(x_p^0)^{1-\alpha}\left(\frac{(x_p^0)^{\alpha}\sum_{l=1}^m y_l(x^0)A_{lp}}{w_p} - (1+\gamma)\right)\\
&\geq w_j(x_j^0)^{1-\alpha}\left(\frac{(x_p^0)^{\alpha}\sum_{l=1}^m y_l(x^0)A_{lp}}{w_p} - (1+\gamma)\right)\\
&\geq w_j(x_j^0)^{1-\alpha}\left(n\frac{(x_j^0)^{\alpha}\sum_{l=1}^m y_l(x^0)A_{lj}}{w_j} - (1+\gamma)\right)\quad (\text{from (\ref{eq:condition-n-dominance})}) \\
&\geq n w_j(x_j^0)^{1-\alpha}\left(\frac{(x_j^0)^{\alpha}\sum_{l=1}^m y_l(x^0)A_{lj}}{w_j} - (1+\gamma)\right)\\
&= n\left(x_j^0\sum_{l=1}^m y_l(x^0)A_{lj} - (1+\gamma)w_j(x_j^0)^{1-\alpha}\right),
\end{align*}
implying the second part of the lemma.

\noindent\textbf{Case 2:} $w_p(x_p^0)^{1-\alpha} < w_j(x_j^0)^{1-\alpha}$. Then:
\begin{align*}
x_p^0\sum_{l=1}^m y_l(x^0)A_{lp} - (1+\gamma)w_p(x_p^0)^{1-\alpha} &> x_p^0\sum_{l=1}^m y_l(x^0)A_{lp} - (1+\gamma)w_j(x_j^0)^{1-\alpha}\\
&\geq n x_j^0\sum_{l=1}^m y_l(x^0)A_{lj} - (1+\gamma)w_j(x_j^0)^{1-\alpha}\quad (\text{from (\ref{eq:slackness-n-dominance})}) \\
&\geq n \left( x_j^0\sum_{l=1}^m y_l(x^0)A_{lj} - (1+\gamma)w_j(x_j^0)^{1-\alpha} \right),
\end{align*}
thus implying the second part of the lemma and completing the proof.
\end{proof}
%%%%%%%%%%%%%%%%%%%%%%%%%%%%%%%%%%%%%%%%%%%%%%%%%%%%%%%%%%%%%%%%%%%
The following lemma lower-bounds the increase in the potential, in each round.
\begin{lemma}\label{lemma:potential-increase-alpha>1}
Let $x^0$ and $x^1$ denote the values of $x$ before and after any fixed round, respectively, and let $S^+ = \{j: x_j^1 > x_j^0\}$, $S^- = \{j: x_j^1 < x_j^0\}$. The potential increase in the round is lower bounded as:
\begin{enumerate}[noitemsep, topsep=5pt]
\item $\Phi(x^1) - \Phi(x^0) \geq \Omega(\beta\gamma)\sum_{j\in\{S^+\cup S^-\}} x_j^0\sum_{i=1}^m y_i(x^0)A_{ij}$;
\item $\Phi(x^1) - \Phi(x^0) \geq \Omega\left(\frac{\beta}{(1-\beta)^{\alpha}}\right) \left(\sum_{j=1}^n x_j^0 \sum_{i=1}^m y_i(x^0) - (1+\gamma) \sum_{j=1}^nw_j (x_j^0)^{1-\alpha}\right)$;
\item $\Phi(x^1) - \Phi(x^0) \geq \Omega\left(\frac{\beta}{(1+\beta)^{\alpha}}\right) \left( (1-\gamma)\sum_{j=1}^n w_j (x_j^0)^{1-\alpha} - \sum_{j=1}^n x_j^0 \sum_{i=1}^m y_i(x^0)\right)$.
\end{enumerate}
\end{lemma}
\begin{proof}
$\quad$\\
\noindent\textbf{Proof of 1.} From Lemma \ref{lemma:potential-increase}:
\begin{equation*}
\Phi(x^1) - \Phi(x^0) \geq \sum_{j=1}^n w_j \frac{|x_j^1 - x_j^0|}{(x_j^1)^{\alpha}}\left|1 - \frac{(x_j^1)^{\alpha}\sum_{i=1}^m y_i(x^1)A_{ij}}{w_j}\right|.
\end{equation*}
Let $\xi_j(x^1) = \frac{(x_j^1)^{\alpha}\sum_{i=1}^m y_i(x^1) A_{ij}}{w_j}$. 
From the proof of Lemma \ref{lemma:potential-increase}, if $x_j^1 - x_j^0 > 0$, then $1 - \xi_j(x^1)\geq \frac{3}{4}\gamma \geq \frac{3}{4}\gamma \xi_j(x^1)$, as $0 < \xi_j(x^1) \leq 1 - \frac{3}{4}\gamma$. If $x_j^ 1 - x_j^0 < 0$, then $1 - \xi_j(x^1) \leq - \frac{\gamma}{2}$, which implies $1\leq \xi_j(x^1)(1 +\frac{\gamma}{2})^{-1}$, and thus $1 - \xi_j(x^1) \leq \xi_j(x^1)((1 + \gamma/2)^{-1} - 1) = \xi_j(x^1) \frac{-\gamma/2}{1 + \gamma/2} <-\xi_j(x^1)\frac{\gamma/2}{3/2}= -\frac{\gamma}{3}\xi_j(x^1)$.
Therefore: $|1-\xi_j(x^1)|\geq \frac{\gamma}{3}\xi_j(x^1) \Leftrightarrow\left|1 - \frac{(x_j^1)^{\alpha}\sum_{i=1}^m y_i(x^1)A_{ij}}{w_j}\right|\geq \frac{\gamma}{3}\frac{(x_j^1)^{\alpha}\sum_{i=1}^m y_i(x^1)A_{ij}}{w_j}$, which further gives:
\begin{equation*}
\Phi(x^1) - \Phi(x^0) \geq \sum_{j=1}^n w_j \frac{|x_j^1 - x_j^0|}{(x_j^1)^{\alpha}} \frac{(x_j^1)^{\alpha}\sum_{i=1}^m y_i(x^1)A_{ij}}{w_j} \geq \frac{\gamma}{3}\sum_{j=1}^n |x_j^1 - x_j^0|\cdot \sum_{i=1}^m y_i(x^1)A_{ij}.
\end{equation*}
If $j\in S^+$, then $x_j^1 = (1+\beta)x_j^0$, and therefore $|x_j^1 - x_j^0|\cdot \sum_{i=1}^m y_i(x^1)A_{ij} = \left(1 - \frac{1}{1+\beta}\right)x_j^1\sum_{i=1}^m y_i(x^1)A_{ij}\geq \left(1 - \frac{\gamma}{4}\right)\frac{\beta}{1+\beta}x_j^0\sum_{i=1}^m y_i(x^0)A_{ij}$.

Similarly, if $j\in S^-$ and $x_j^0 \geq \frac{\delta_j}{1-\beta}$, then $x_j^1 = (1-\beta)x_j^0$, and therefore $|x_j^1 - x_j^0|\cdot \sum_{i=1}^m y_i(x^1)A_{ij} = \left(\frac{1}{1-\beta}-1\right)x_j^1\sum_{i=1}^m y_i(x^1)A_{ij}\geq \left(1 - \frac{\gamma}{4}\right)\frac{\beta}{1-\beta}x_j^0\sum_{i=1}^m y_i(x^0)A_{ij}$. Using part 1 of Lemma \ref{lemma:alpha>1-mul-increase-over-S-}:
\begin{equation*}
\Phi(x^1) - \Phi(x^0) \geq \frac{\gamma}{6}\frac{\beta}{1+\beta}\sum_{j\in \{S^+\cup S^-\}} x_j^0\sum_{i=1}^m y_i(x^0)A_{ij}.
\end{equation*}

\noindent\textbf{Proof of 2:} Consider $j\in S^-$ such that $x_j^0 \geq \frac{\delta_j}{1-\beta}$. Then $x_j^1 = (1-\beta)x_j^0$, $\frac{(x_j^1)^{\alpha}\sum_{i=1}^m y_i(x^1)A_{ij}}{w_j}\geq (1+\gamma)$, and using Lemma \ref{lemma:potential-increase}:
\begin{align*}
\Phi(x^1) - \Phi(x^0) &\geq \sum_{\{j\in S^-: x_j^0 \geq \frac{\delta_j}{1-\beta}\}} w_j \frac{|x_j^1 - x_j^0|}{(x_j^1)^{\alpha}}\left|1 - \frac{(x_j^1)^{\alpha}\sum_{i=1}^m y_i(x^1)A_{ij}}{w_j}\right|\\
&\geq \sum_{\{j\in S^-: x_j^0 \geq \frac{\delta_j}{1-\beta}\}} w_j \frac{\beta}{(1-\beta)^{\alpha}}(x_j^0)^{1-\alpha} \left(\frac{(x_j^1)^{\alpha}\sum_{i=1}^m y_i(x^1)A_{ij}}{w_j} - 1 \right)\\
&\geq \frac{\beta}{(1-\beta)^{\alpha}}\sum_{\{j\in S^-: x_j^0 \geq \frac{\delta_j}{1-\beta}\}}w_j(x_j^0)^{1-\alpha}\left((1-\gamma/4)\frac{(x_j^0)^{\alpha}\sum_{i=1}^m y_i(x^0)A_{ij}}{w_j} - 1 \right)\\
&\geq(1-\gamma/4) \frac{\beta}{(1-\beta)^{\alpha}} \sum_{\{j\in S^-: x_j^0 \geq \frac{\delta_j}{1-\beta}\}}w_j(x_j^0)^{1-\alpha}\left(\frac{(x_j^0)^{\alpha}\sum_{i=1}^m y_i(x^0)A_{ij}}{w_j} - (1+\gamma) \right)\\
&=(1-\gamma/4) \frac{\beta}{(1-\beta)^{\alpha}} \sum_{\{j\in S^-: x_j^0 \geq \frac{\delta_j}{1-\beta}\}}\left( x_j^0 \sum_{i=1}^m y_i(x^0) - (1+\gamma) w_j (x_j^0)^{1-\alpha}\right).
\end{align*}
Using the second part of Lemma \ref{lemma:alpha>1-mul-increase-over-S-} and the fact that for $k\notin S^-$: $\frac{(x_k^0)^{\alpha}\sum_{i=1}^m y_i(x^0)A_{ik}}{w_k}<(1+\gamma)$, we get the desired result:
\begin{equation*}
\Phi(x^1) - \Phi(x^0) \geq \frac{1}{2}(1-\gamma/4) \frac{\beta}{(1-\beta)^{\alpha}}\left(\sum_{j=1}^n x_j^0 \sum_{i=1}^m y_i(x^0) - (1+\gamma) \sum_{j=1}^nw_j (x_j^0)^{1-\alpha}\right).
\end{equation*}

\noindent\textbf{Proof of 3:} The proof is equivalent to the proof of Lemma \ref{lemma:potential-increase-alpha<1}, part 2, and is omitted for brevity. 
\end{proof}

Consider the following definition of a stationary round:
\begin{definition}\label{def:stationary-round}
(Stationary round.) A round is stationary, if both:
\begin{enumerate}[topsep = 5pt]
\item $\sum_{j\in\{S^+\cup S^-\}} x_j^0\sum_{i=1}^m y_i(x)A_{ij} \leq \gamma \sum_{j=1}^n w_j {(x_j^0)}^{1-\alpha}$, and
\item $(1-2\gamma)\sum_{j=1}^n w_j {(x_j^0)}^{1-\alpha}\leq \sum_{j=1}^n x_j^0\sum_{i=1}^m y_i(x^0)A_{ij}$
\end{enumerate}
hold, where $S^+ = \{j: x_j^1 > x_j^0\}$, $S^- = \{j: x_j^1 < x_j^0\}$. Otherwise, the round is non-stationary.
\end{definition}

%%%%%%%%%%%%%%%%%%%%%%%%%%%%%%%%%%%%%%%%%%%%%%%%%%%%%%%%%%%%

The following two technical propositions are used in Lemma \ref{lemma:alpha>1-stationary-near-opt} for bounding the duality gap in stationary rounds.
%%%%%%%%%%%%%%%%%%%%%%%%%%%%%%%%%%%%%%%%%%%%%%%%%%%%%%%%%%%%%%%%%%%%%%%%%%%%%%%%%
\begin{proposition}\label{prop:alpha>1-duality-gap}
After the initial the initial $\tau_0 + \tau_1$ rounds, where $\tau_0 = \frac{1}{\beta}\ln(1/\delta_{\min})$, $\tau_1 = \frac{1}{\beta}\ln(nA_{\max})$, it is always true that $G_\alpha(x, y(x))\leq \sum_{j=1}^n w_j\frac{x_j^{1-\alpha}}{\alpha-1}\Big(1 + (1+3\varepsilon)(\alpha-1)\xi_j - \alpha\xi_j^{\frac{\alpha-1}{\alpha}} \Big)$, where $\xi_j = \frac{x_j^\alpha \sum_i y_i(x)A_{ij}}{w_j}$.
\end{proposition}
\begin{proof}
Recall from (\ref{eq:duality-gap-alpha}) that the duality gap for $x, y$ in $(P_\alpha)$ is given as: 
\begin{equation*}
G_{\alpha}(x, y)=\sum_{j=1}^n w_j\frac{x_j^{1-\alpha}}{1-\alpha}\left(\left(\frac{w_j}{{x_j}^{\alpha}\sum_{i=1}^m y_i A_{ij}}\right)^{\frac{1-\alpha}{\alpha}}-1\right) +\sum_{i=1}^m y_i - \sum_{j=1}^n  w_j x_j^{1-\alpha}\cdot \left(\frac{{x_j}^{\alpha}\sum_{j=1}^n A_{ij}y_i}{w_j}\right)^{\frac{\alpha-1}{\alpha}}.
\end{equation*}
From Lemma \ref{lemma:approx-comp-slack}, after at most initial $\tau_0+\tau_1$ rounds:
\begin{align*}
\sum_{i=1}^m y_i&\leq (1+3\varepsilon)\sum_{j=1}^nx_j\sum_{i=1}^m y_i A_{ij}\\
&=(1+3\varepsilon)\sum_{j=1}^n w_j{x_j}^{1-\alpha}\cdot\left(\frac{{x_j}^{\alpha}\sum_{i=1}^m y_i A_{ij}}{w_j}\right),
\end{align*}
and letting $\xi_j = \frac{{x_j}^{\alpha}\sum_{i=1}^m y_i A_{ij}}{w_j}$, we get:
\begin{align*}
G_{\alpha}(x, y)&\leq \sum_{j=1}^n w_j\frac{x_j^{1-\alpha}}{1-\alpha}\left(\xi_j^{\frac{\alpha-1}{\alpha}}-1 + (1+3\varepsilon)(1-\alpha)\xi_j -(1-\alpha)\xi_j^{\frac{\alpha-1}{\alpha}}\right)\\
&= \sum_{j=1}^n w_j\frac{x_j^{1-\alpha}}{1-\alpha}\left(\alpha\xi_j^{\frac{\alpha-1}{\alpha}}+ (1+3\varepsilon)(1-\alpha)\xi_j -1 \right)\\ 
&= \sum_{j=1}^n w_j\frac{x_j^{1-\alpha}}{\alpha-1}\left(1 + (1+3\varepsilon)(\alpha-1)\xi_j - \alpha\xi_j^{\frac{\alpha-1}{\alpha}} \right).
\end{align*}
\end{proof}

\begin{proposition}\label{prop:alpha>1-r-bound}
Let $r_\alpha(\xi_j) =\left(1 + (1+3\varepsilon)(\alpha-1)\xi_j - \alpha\xi_j^{\frac{\alpha-1}{\alpha}} \right)$, where $\xi_j = \frac{{x_j}^{\alpha}\sum_{i=1}^m y_i A_{ij}}{w_j}$. If $\alpha > 1$ and
$ \xi_j\in (1-\gamma, 1+\gamma)$ $\forall j\in \{1,...,n\}$, 
then $r_\alpha(\xi_j)\leq \varepsilon(3\alpha-2)$.
\end{proposition}
\begin{proof}

Observe the first and the second derivative of $r_\alpha(\xi_j)$:
\begin{align*}
\frac{d r_\alpha(\xi_j)}{d \xi_j} 
&= (\alpha-1)(1+ 3\varepsilon - \xi_j^{-1/\alpha});\\
\frac{d^2 r_\alpha(\xi_j)}{d{\xi_j}^2}&= \frac{1}{\alpha}(\alpha-1){\xi_j}^{-1/\alpha-1}.
\end{align*}
As $\xi_j>0$, $r(\xi_j)$ is convex for $\alpha>1$, and therefore: 
$
r(\xi_j)\leq \max\{r(1-\gamma), r(1+\gamma)\}.
$ 
We have that:
\begin{align*}
r(1-\gamma) = r(1-\varepsilon/4) &= 1- \left(1-\frac{\varepsilon}{4}\right)((1-\alpha)(1+3\varepsilon)+\alpha(1-\varepsilon/4)^{-1/\alpha})\\
&\leq 1 - \left(1-\frac{\varepsilon}{4}\right)(1-\alpha + 3\varepsilon(1 -  \alpha) + \alpha(1+\varepsilon/4)^{1/\alpha})\\
&\leq 1 - \left(1-\frac{\varepsilon}{4}\right)(1+ \varepsilon/4+ 3\varepsilon(1-\alpha)) \quad (\text{from }(1+\varepsilon/4)^{1/\alpha}\geq 1+ {\varepsilon}/({4\alpha}))\\
&= 1 - 1 - \frac{\varepsilon}{4} + 3\varepsilon(\alpha-1)+\frac{\varepsilon}{4}(1+\varepsilon/4-3\varepsilon(\alpha-1))\\
&= \frac{\varepsilon^2}{16} + 3\varepsilon(\alpha-1)\left(1-\frac{\varepsilon}{4}\right)\\
&\leq \varepsilon(3\alpha-2).
\end{align*}
On the other hand:
\begin{align*}
r(1+\gamma) = r(1+\varepsilon/4) &= 1- \left(1+\frac{\varepsilon}{4}\right)((1-\alpha)(1+3\varepsilon)+\alpha(1+\varepsilon/4)^{-1/\alpha})\\
&\leq 1- \left(1+\frac{\varepsilon}{4}\right)(1-\alpha + 3\varepsilon - 3\varepsilon \alpha + \alpha(1-\varepsilon/4)^{1/\alpha})\\
&\leq 1- \left(1+\frac{\varepsilon}{4}\right)\left(1+\frac{11}{4}\varepsilon-3\varepsilon\alpha\right)\\
&\leq 1 - \left(1+\frac{11}{4}\varepsilon - 3\varepsilon\alpha\right)\\
&\leq \varepsilon(3\alpha-2),
\end{align*}
completing the proof.
\end{proof}

%%%%%%%%%%%%%%%%%%%%%%%%%%%%%%%%%%%%%%%%%%%%%%%%%%%%%%%%%%%%%%%%%%%%%%%%%%
The following lemma states that in any stationary round current solution is an $(1+\varepsilon(4\alpha-1))$-approximate solution.
\begin{lemma}\label{lemma:alpha>1-stationary-near-opt}
In any stationary round that happens after the initial the initial $\tau_0 + \tau_1$ rounds, where $\tau_0 = \frac{1}{\beta}\ln(1/\delta_{\min})$, $\tau_1 = \frac{1}{\beta}\ln(nA_{\max})$, we have that $p_\alpha(x^*) - p_\alpha(x)\leq \varepsilon(4\alpha-1)(-p_\alpha(x))$, where $x^*$ is the optimal solution to $(P_\alpha)$ and $x$ is the solution at the beginning of the round.
\end{lemma}
\begin{proof}
Observe that for any $k\notin \{S^+ \cup S^-\}$ (by the definition of $S^+$ and $S^-$) we have that $1-\gamma<\frac{x_k^{\alpha}\sum_{i=1}^m y_i(x)A_{ik}}{w_k}<1+\gamma$, which is equivalent to:
\begin{equation}
(1-\gamma)w_k x_k^{1-\alpha}< x_k\sum_{i=1}^m y_i(x)A_{ik} < (1+\gamma)w_k x_k^{1-\alpha} \quad \forall k\notin \{S^+\cup S^-\}.\label{eq:tight-stationary-xk}
\end{equation}

Using stationarity and (\ref{eq:tight-stationary-xk}):
\begin{align}
(1-2\gamma)\sum_{j=1}^nw_jx_j^{1-\alpha} &\leq \sum_{j=1}^n x_j\sum_{i=1}^n y_i(x)A_{ij}\notag\\
&= \sum_{l\in\{S^+\cup S^-\}} x_l\sum_{i=1}^n y_i(x)A_{il} +  \sum_{k\notin\{S^+\cup S^-\}} x_k\sum_{i=1}^n y_i(x)A_{ik}\notag\\
&< \gamma \sum_{j=1}^n w_j x_j^{1-\alpha} + (1+\gamma)\sum_{k\notin\{S^+\cup S^-\}} w_k x_k^{1-\alpha}.\label{eq:bound-on-all-wj-xj}
\end{align}
Since $\sum_{l\in\{S^+\cup S^-\}}w_{l}x_l^{1-\alpha} = \sum_{j=1}^n w_j x_j^{1-\alpha} - \sum_{k\notin\{S^+\cup S^-\}} w_k x_k^{1-\alpha}$, using (\ref{eq:bound-on-all-wj-xj}):
\begin{align*}
(1-2\gamma)\sum_{l\in\{S^+\cup S^-\}}w_{l}x_l^{1-\alpha} &< \gamma \sum_{j=1}^n w_j x_j^{1-\alpha} + (1+\gamma)\sum_{k\notin\{S^+\cup S^-\}} w_k x_k^{1-\alpha} - (1-2\gamma)\sum_{k\notin\{S^+\cup S^-\}} w_k x_k^{1-\alpha}\\
&= \gamma \sum_{j=1}^n w_j x_j^{1-\alpha} + 3\gamma \sum_{k\notin\{S^+\cup S^-\}} w_k x_k^{1-\alpha}\\
&\leq 4\gamma \sum_{j=1}^n w_j x_j^{1-\alpha},
\end{align*}
and therefore:
\begin{equation}
\sum_{l\in\{S^+\cup S^-\}}w_{l}x_l^{1-\alpha}<\frac{4\gamma}{1-2\gamma} \sum_{j=1}^n w_j x_j^{1-\alpha} < 5\gamma \sum_{j=1}^n w_j x_j^{1-\alpha}, \label{eq:tight-active-objective-part} 
\end{equation}
as $\gamma = \frac{\varepsilon}{4}$ and $\varepsilon\leq \frac{1}{6}$.

As $p_\alpha(x^*)-p_\alpha(x)\leq G(x, y(x))$, from Proposition \ref{prop:alpha>1-duality-gap}:
\begin{align*}
p_\alpha(x^*) - p_\alpha(x) \leq& \sum_{j=1}^n w_j\frac{x_j^{1-\alpha}}{\alpha-1}\left(1 + (1+3\varepsilon)(\alpha-1)\xi_j -\alpha\xi_j^{\frac{\alpha-1}{\alpha}} \right)\\
=& \sum_{k\notin\{S^+\cup S^-\}} w_k\frac{x_k^{1-\alpha}}{\alpha-1}\left(1 + (1+3\varepsilon)(\alpha-1)\xi_k -\alpha\xi_k^{\frac{\alpha-1}{\alpha}} \right)\\ 
&+ \sum_{l\in\{S^+\cup S^-\}} w_l\frac{x_l^{1-\alpha}}{\alpha-1}\left(1 + (1+3\varepsilon)(\alpha-1)\xi_l -\alpha\xi_l^{\frac{\alpha-1}{\alpha}} \right).
\end{align*}
From Proposition \ref{prop:alpha>1-r-bound}:
\begin{align}
\sum_{k\notin\{S^+\cup S^-\}} w_k\frac{x_k^{1-\alpha}}{\alpha-1}\left(1 + (1+3\varepsilon)(\alpha-1)\xi_k -\alpha\xi_k^{\frac{\alpha-1}{\alpha}} \right) &\leq \varepsilon(3\alpha - 2)\sum_{k\notin\{S^+\cup S^-\}} w_k\frac{x_k^{1-\alpha}}{\alpha-1}\notag\\
&\leq \varepsilon(3\alpha - 2) \sum_{j=1}^n w_j\frac{x_j^{1-\alpha}}{\alpha-1}\notag\\
&= \varepsilon(3\alpha - 2)(-p_{\alpha}(x)). \label{eq:approx-for-passive-x}
\end{align}

Observe $\sum_{l\in\{S^+\cup S^-\}} w_l\frac{x_l^{1-\alpha}}{\alpha-1}\left(1 + (1+3\varepsilon)(\alpha-1)\xi_l -\alpha\xi_l^{\frac{\alpha-1}{\alpha}} \right)$. Since $\alpha>1$, each $w_l\frac{x_l^{1-\alpha}}{\alpha-1}>0$, and therefore:
\begin{align*}
\sum_{l\in\{S^+\cup S^-\}} w_l\frac{x_l^{1-\alpha}}{\alpha-1} & \left(1 + (1+3\varepsilon)(\alpha-1)\xi_l -\alpha\xi_l^{\frac{\alpha-1}{\alpha}} \right)\\ 
&\leq \sum_{l\in\{S^+\cup S^-\}} w_l\frac{x_l^{1-\alpha}}{\alpha-1}\left((1+3\varepsilon)(\alpha-1)\xi_l + 1 \right)\\
&= \sum_{l\in\{S^+\cup S^-\}} w_l\frac{x_l^{1-\alpha}}{\alpha-1}\left((1+3\varepsilon)(\alpha-1)\frac{x_l^{\alpha}\sum_{i=1}^m y_i(x)A_{il}}{w_l} + 1 \right)\\
&= (1+3\varepsilon)\sum_{l\in\{S^+\cup S^-\}} x_l\sum_{i=1}^m y_i(x)A_{il} + \sum_{l\in\{S^+\cup S^-\}} w_l\frac{x_l^{1-\alpha}}{\alpha-1}.
\end{align*}
Now, from stationarity $\sum_{l\in\{S^+\cup S^-\}} x_l\sum_{i=1}^m y_i(x)A_{il}<\gamma \sum_{j=1}^n w_j x_j^{1-\alpha}$ and using (\ref{eq:tight-active-objective-part}) we get:
\begin{align}
\sum_{l\in\{S^+\cup S^-\}} w_l\frac{x_l^{1-\alpha}}{\alpha-1}\left(1 + (1+3\varepsilon)(\alpha-1)\xi_j -\alpha\xi_j^{\frac{\alpha-1}{\alpha}} \right)&< \sum_{j=1}^n w_j \frac{x_j^{1-\alpha}}{\alpha-1}(\gamma(1+3\varepsilon)(\alpha-1) + 5\gamma)\notag\\
%&\leq -p_{\alpha}(x)\left(\frac{3}{2}\gamma(\alpha-1)+5\gamma\right)\notag\\
&\leq -p_{\alpha}(x)\left(\frac{3\varepsilon}{8}\alpha +\varepsilon \right). \label{eq:approx-for-active-x}
\end{align}
Finally, combining (\ref{eq:approx-for-passive-x}) and (\ref{eq:approx-for-active-x}): 
$
p_\alpha(x^*) - p_\alpha(x) < \varepsilon(4\alpha - 1)(-p_{\alpha}(x)).
$
\end{proof}
%%%%%%%%%%%%%%%%%%%%%%%%%%%%%%%%%%%%%%%%%%%%%%%%%%%%%%%%%%%%%%%%%%%%%%%%%%%%%%%%%%%%%%%%%%%
The following two lemmas are used for lower-bounding the potential increase in non-stationary rounds.

\begin{lemma}\label{lemma:alpha>1-large-sum-yi}
Consider any non-stationary round that happens after the initial $\tau_0 + \tau_1$ rounds, where $\tau_0 = \frac{1}{\beta}\ln(1/\delta_{\min})$, $\tau_1 = \frac{1}{\beta}\ln(nA_{\max})$. Let $x^0$ and $x^1$ denote the values of $x$ before and after the round update. If $\frac{1}{\kappa}\sum_i y_(x^0) \geq -\sum_j w_j \frac{{(x_j^0)}^{1-\alpha}}{1-\alpha}$, then $\Phi(x^1) - \Phi(x^0)\geq \Omega(\gamma^3)(-\Phi(x^0))$.
\end{lemma}
\begin{proof}
Observe that as $\frac{1}{\kappa}\sum_i y_(x^0) \geq -\sum_j w_j \frac{{(x_j^0)}^{1-\alpha}}{1-\alpha}$,
\begin{equation*}
-\Phi(x_0)\leq 2\cdot\frac{1}{\kappa}\sum_i y_(x^0)\leq \frac{2(1-3\varepsilon)}{\kappa} \sum_{j=1}x_j^0 \sum_{i=1}^m y_i(x^0)A_{ij},
\end{equation*}
where the last inequality follows from Lemma \ref{lemma:approx-comp-slack}.

Since the round is not stationary, we have that either: 
\begin{enumerate}
\item $\sum_{j\in S^-\cup S^+}x_j^0\sum_i y_i(x)A_{ij} > \gamma \sum_{j=1}^n w_j (x_j^0)^{1-\alpha}$, or
\item $(1-2\gamma)\sum_{j=1}^n w_j(x_j^0)^{1-\alpha} > \sum_{j=1}^n x_j^0 \sum_{i=1}^m y_i(x^0)A_{ij}$.
\end{enumerate}
\noindent\textbf{Case 1: $\sum_{j\in S^-\cup S^+}x_j^0\sum_i y_i(x)A_{ij} > \gamma \sum_{j=1}^n w_j (x_j^0)^{1-\alpha}$.} If: 
\begin{equation*}\sum_{j=1}^n x_j^0 \sum_{i=1}^m y_i(x^0)\leq (1+2\gamma)\sum_{j=1}^n w_j (x_j^0)^{1-\alpha},
\end{equation*}
then
\begin{equation*}
\sum_{j\in S^-\cup S^+}x_j^0\sum_i y_i(x)A_{ij} > \frac{\gamma}{1+2\gamma}\sum_{j=1}x_j^0 \sum_{i=1}^m y_i(x^0)A_{ij} = \Omega(\gamma)\sum_{j=1}x_j^0 \sum_{i=1}^m y_i(x^0)A_{ij},
\end{equation*}
and, from the first part of Lemma \ref{lemma:potential-increase-alpha>1}, the potential increase is lower bounded as:
\begin{align}
\Phi(x^1) - \Phi(x^0) &\geq \Omega(\beta\gamma^2) \sum_{j=1}x_j^0 \sum_{i=1}^m y_i(x^0)A_{ij}\notag\\
&= \Omega(\beta\kappa\gamma^2) (-\Phi(x^0))\notag\\
&= \Omega(\gamma^3)(-\Phi(x^0)). \notag%\label{eq:alpha>1-y-inc-1}
\end{align}
On the other hand, if:
\begin{equation*}
\sum_{j=1}^n x_j^0 \sum_{i=1}^m y_i(x^0)>(1+2\gamma)\sum_{j=1}^n w_j (x_j^0)^{1-\alpha},
\end{equation*}
then, from the second part of Lemma \ref{lemma:potential-increase-alpha>1}:
\begin{align}
\Phi(x^1) - \Phi(x^0) &\geq \Omega(\beta\gamma)\sum_{j=1}x_j^0 \sum_{i=1}^m y_i(x^0)A_{ij}\notag\\
&= \Omega(\beta\gamma\kappa)(-\Phi(x^0))\notag\\
&= \Omega(\gamma^2)(-\Phi(x^0)). \notag%\label{eq:alpha>1-y-inc-2}
\end{align}

\noindent\textbf{Case 2: $(1-2\gamma)\sum_{j=1}^n w_j(x_j^0)^{1-\alpha} > \sum_{j=1}^n x_j^0 \sum_{i=1}^m y_i(x^0)A_{ij}$.} Then, using the third part of Lemma \ref{lemma:potential-increase-alpha>1}:
\begin{align}
\Phi(x^1) - \Phi(x^0) &\geq \Omega\left(\frac{\beta}{(1+\beta)^{\alpha}}\gamma\right)\sum_{j=1}x_j^0 \sum_{i=1}^m y_i(x^0)A_{ij}\notag\\
&= \Omega\left({\beta}\gamma\right)\sum_{j=1}x_j^0 \sum_{i=1}^m y_i(x^0)A_{ij}\notag\\
&= \Omega(\beta\gamma\kappa)(-\Phi(x^0))\notag\\
&= \Omega(\gamma^2)(-\Phi(x^0)), \notag%\label{eq:alpha>1-y-inc-3}
\end{align}
where in the second line we have used that $\frac{\beta}{(1+\beta)^{\alpha}} = \Theta(\beta)$. This can be shown using the generalized Bernoulli's inequality and $\varepsilon\alpha\leq\frac{9}{10}$ as follows:
\begin{equation*}
\frac{1}{(1+\beta)^{\alpha}}\geq (1-2\beta)^{\alpha}\geq 1-2\alpha\beta = 1 - \frac{\alpha}{k+\alpha}\cdot\frac{\varepsilon}{10}\geq 1 - \frac{9}{100} = \Theta(1).
\end{equation*}
\end{proof}

\begin{lemma}\label{lemma:alpha>1-mul-inc-non-stat}
Consider any non-stationary round that happens after the initial $\tau_0 + \tau_1$ rounds, where $\tau_0 = \frac{1}{\beta}\ln(1/\delta_{\min})$, $\tau_1 = \frac{1}{\beta}\ln(nA_{\max})$. Let $x^0$ and $x^1$ denote the values of $x$ before and after the round update. If $\frac{1}{\kappa}\sum_i y_(x^0) < -\sum_j w_j \frac{{(x_j^0)}^{1-\alpha}}{1-\alpha}$, then $\Phi(x^1) - \Phi(x^0) \geq \Omega\left(\beta\gamma^2\right)(\alpha-1)(-\Phi(x^0))$.
\end{lemma}
\begin{proof}
Observe that as $\frac{1}{\kappa}\sum_i y_(x^0) < -\sum_j w_j \frac{{(x_j^0)}^{1-\alpha}}{1-\alpha}$,
\begin{equation*}
-\Phi(x_0)\leq -2\sum_j w_j \frac{{(x_j^0)}^{1-\alpha}}{1-\alpha} = \frac{2}{\alpha-1} \sum_j w_j {{(x_j^0)}^{1-\alpha}}.
\end{equation*}

From the definition of a stationary round, we have either of the following two cases:

\noindent\textbf{Case 1:} $\sum_{j\in\{S^+\cup S^-\}}x_j\sum_{i=1}^m y_i(x)A_{ij}> \gamma \sum_{j=1}^n w_j x_j^{1-\alpha}$. 
From the first part of Lemma \ref{lemma:potential-increase-alpha>1}, the increase in the potential is: $\Phi(x^1) - \Phi(x^0) \geq \Omega\left(\beta\gamma^2\right)\sum_{j=1}^n w_j x_j^{1-\alpha}$. As $-\Phi(x^0)\leq \frac{2}{\alpha-1} \sum_j w_j {{(x_j^0)}^{1-\alpha}}$, the increase in the potential is at least:
\begin{align}
\Phi(x^1) - \Phi(x^0) &\geq \Omega(\beta\gamma^2)(\alpha-1)(-\Phi(x^0)).\notag%\label{eq:alpha>1-mul-pot-inc-1}
\end{align}

\noindent\textbf{Case 2:} $(1-2\gamma)\sum_{j=1}^n w_j x_j^{1-\alpha} > \sum_{j=1}^n x_j\sum_{i=1}^m y_i(x)A_{ij}$. Using part 3 of Lemma \ref{lemma:potential-increase-alpha>1}, the increase in the potential is then $\Phi(x^1) - \Phi(x^0)\geq \Omega\left(\frac{\beta}{(1+\beta)^{\alpha}}\gamma\right)\sum_{j=1}^n w_j x_j^{1-\alpha}$. Therefore, using that $\frac{\beta}{(1+\beta)^{\alpha}} = \Theta(\beta)$ as in the proof of Lemma \ref{lemma:alpha>1-large-sum-yi}:
\begin{align}
\Phi(x^1) - \Phi(x^0) &\geq \Omega(\beta\gamma)(\alpha-1)(-\Phi(x^0)).\notag
\end{align}

\end{proof}

\begin{proofof}{Theorem \ref{thm:convergence-alpha>1}}
We will bound the total number of non-stationary rounds that happen after the initial $\tau_0 + \tau_1$ rounds, where $\tau_0 = \frac{1}{\beta}\ln(1/\delta_{\min})$, $\tau_1 = \frac{1}{\beta}\ln(nA_{\max})$. The total convergence time is then at most the sum of $\tau_0+\tau_1$ rounds and the number of non-stationary rounds that happen after the initial $\tau_0 + \tau_1$ rounds, since, from Lemma \ref{lemma:alpha>1-stationary-near-opt}, in any stationary round: $p(x^*) - p(x)\leq \varepsilon(4\alpha-1)(-p(x))$.

Consider the non-stationary rounds that happen after the initial $\tau_0 + \tau_1$ rounds. 
As $x_j\in[\delta_j, 1]$, $\forall j$, it is simple to show that:
\begin{equation}
\frac{W}{\alpha-1}\leq \sum_j w_j \frac{{x_j}^{1-\alpha}}{\alpha-1}\leq \frac{W}{\alpha-1}\cdot 2{\wratio}^{\frac{\alpha-1}{\alpha}}n^{2(\alpha-1)}m^{\alpha-1}{A_{\max}}^{2\alpha -1}, \label{eq:alpha>1-palpha-bounds}
\end{equation}
and
\begin{equation}
0< \frac{1}{\kappa}\sum_{i} y_i(x)\leq \frac{mC}{\kappa}\leq \varepsilon m C. \label{eq:alpha>1-sum-yi-bounds}
\end{equation}

Recall that $\Phi(x) = - \sum_j w_j \frac{{x_j}^{1-\alpha}}{\alpha-1} - \frac{1}{\kappa}\sum_{i} y_i(x)$ and that the potential $\Phi(x)$ never decreases.

There can be two cases of non-stationary rounds: those in which $\sum_j w_j \frac{{x_j}^{1-\alpha}}{\alpha-1}$ dominates in the absolute value of the potential, and those in which $\frac{1}{\kappa}\sum_{i} y_i(x)$ dominates in the absolute value of the potential. We bound the total number of the non-stationary rounds in such cases as follows.

\noindent\textbf{Case 1: $\frac{1}{\kappa}\sum_{i} y_i(x) \geq \sum_j w_j \frac{{x_j}^{1-\alpha}}{\alpha-1}$.}  
From (\ref{eq:alpha>1-palpha-bounds}) and (\ref{eq:alpha>1-sum-yi-bounds}), in any such round, the negative potential is bounded as:
\begin{equation*}
\Omega\left( \frac{W}{\alpha-1}\right)\leq -\Phi(x)\leq O\left(\varepsilon{mC}\right).
\end{equation*}
Moreover, from Lemma \ref{lemma:alpha>1-large-sum-yi}, in each Case 1 non-stationary round, the potential increases by at least $\Omega(\gamma^3)(-\Phi(x))$. It immediately follows that there can be at most:
\begin{align}
O\left(\frac{1}{\gamma^3}\ln\left(\frac{\varepsilon{mC}}{\frac{W}{\alpha-1}}\right)\right)&= O\left(\frac{1}{\gamma^3}\ln\left((\alpha-1)\varepsilon\wratio nmA_{\max}\right)\right)\notag\\
&= O\left(\frac{1}{\varepsilon^3}\ln\left(\wratio nmA_{\max}\right)\right) \label{eq:alpha>1-conv-bound-1}
\end{align}
Case 1 non-stationary rounds, as $(\alpha-1)\varepsilon<\alpha\varepsilon\leq \frac{9}{10}$.

\noindent\textbf{Case 2: $\frac{1}{\kappa}\sum_{i} y_i(x) < \sum_j w_j \frac{{x_j}^{1-\alpha}}{\alpha-1}$.} From (\ref{eq:alpha>1-palpha-bounds}) and (\ref{eq:alpha>1-sum-yi-bounds}), in any such round, the negative potential is bounded as:
\begin{equation*}
\Omega\left(\frac{W}{\alpha-1}\right)\leq-\Phi(x)\leq O\left(\frac{W}{\alpha-1}\cdot{{\wratio}^{\frac{\alpha-1}{\alpha}}n^{2(\alpha-1)}m^{\alpha-1}{A_{\max}}^{2\alpha -1}}\right). 
\end{equation*}
Moreover, from Lemma \ref{lemma:alpha>1-mul-increase-over-S-}, in each such non-stationary round the potential increases by at least $\Omega\left(\beta\gamma^2\right)(\alpha-1)(-\Phi(x^0))$. Therefore, there can be at most:
\begin{align}
O\left(\frac{1}{\beta\gamma^2(\alpha-1)}\ln\left(\frac{\frac{W}{\alpha-1}\cdot{{\wratio}^{\frac{\alpha-1}{\alpha}}n^{2(\alpha-1)}m^{\alpha-1}{A_{\max}}^{2\alpha -1}}}{\frac{W}{\alpha-1}}\right)\right) &= O\left(\frac{1}{\beta\gamma^2}\ln({\wratio}^{\frac{1}{\alpha}} nmA_{\max})\right)\notag\\
&= O\left(\frac{1}{\varepsilon^4}\ln(\wratio nmA_{\max})\ln\left(\wratio\cdot\frac{nmA_{\max}}{\varepsilon}\right)\right)\label{eq:alpha>1-conv-bound-2}
\end{align}
Case 2 non-stationary rounds.

The total number of initial $\tau_0+\tau_1$ rounds can be bounded as:
\begin{align}
\tau_0+\tau_1 &= \frac{1}{\beta}\ln(1/\delta_{\min}) + \frac{1}{\beta}\ln(nA_{\max})\notag\\
&= O\left(\frac{1}{\varepsilon^2}\ln\left(\wratio nmA_{\max}\right)\ln\left(\wratio \cdot\frac{nmA_{\max}}{\varepsilon}\right)\right). \label{eq:alpha>1-conv-bound-3}
\end{align}

Combining (\ref{eq:alpha>1-conv-bound-1}), (\ref{eq:alpha>1-conv-bound-2}), and (\ref{eq:alpha>1-conv-bound-3}), the total convergence time is at most:
\begin{align*}
O\left(\frac{1}{\varepsilon^4}\ln\left(\wratio\cdot{nmA_{\max}}\right)\ln\left(\wratio\cdot\frac{nmA_{\max}}{\varepsilon}\right)\right).
\end{align*} 

Finally, running \textsc{$\alpha$-FairPSolver} for the approximation parameter $\varepsilon' = \varepsilon/(4\alpha-1)$, we get that in any stationary round $p_\alpha(x^*)-p_{\alpha}(x)\leq -\varepsilon p_{\alpha}(x)$, while the total number of non-stationary rounds is at most:
\begin{align*}
O\left(\frac{\alpha^4}{\varepsilon^4}\ln\left(\wratio\cdot{nmA_{\max}}\right)\ln\left(\wratio\cdot\frac{nmA_{\max}}{\varepsilon}\right)\right).
\end{align*} 
\end{proofof}

\fi
%%%%%%%%%%%%%%%%%%%%%%%%%%%%%%%%%%%%%%%%%%%%%%%%%%%%%%%%%%%%%%%%%%%%%%%%%%%%%%%%%%%%%%%%%%%%
%%%%%%%%%%%%%%%%%%%%%%%%%%%%%%%%%%%%%%%%%%%%%%%%%%%%%%%%%%%%%%%%%%%%%%%%%%%%%%%%%%%%%%%%%%%%
\subsection{Structural Properties of $\alpha-$Fair Allocations}
\paragraph{Lower Bound on the Minimum Allocated Value.}
Recall (from Section \ref{section:prelims}) 
that the optimal solution $x^*$ to $(P_\alpha)$ must lie in the positive orthant. We show in Lemma \ref{lemma:lower-bound} that not only does $x^*$ lie in the positive orthant, but the minimum element of $x^*$ can be bounded below as a function of the problem parameters. This lemma motivates the choice of parameters $\delta_j$ in \textsc{$\alpha$-FairPSolver} (Section \ref{section:algorithm}). 
%The proof is provided in Appendix \ref{appendix:lower-bound}.

\begin{lemma}\label{lemma:lower-bound}
Let $x^* = (x_1^*,...,x_n^*)$ be the optimal solution to $(P_\alpha)$. Then $\forall j\in \{1,...,n\}$:
\begin{itemize}
\itemsep0pt
\item $x_j^*\geq \big(\frac{w_j}{w_{\max}M}\min_{i: A_{ij}\neq 0}\frac{1}{n_i A_{ij}}\big)^{1/\alpha}$, if  $0<\alpha\leq1$,
\item $x_j^* \geq {{A_{\max}}}^{(1-\alpha)/\alpha}\big(\frac{w_j}{w_{\max}M}\big)^{1/\alpha}\min_{i: A_{ij}\neq 0}\frac{1}{n_i A_{ij}}$, if $\alpha>1$,
\end{itemize}
where $n_i = \sum_{j=1}^n \mathds{1}_{\{A_{ij}\neq 0\}}$\footnote{With the abuse of notation, $\mathds{1}_{\{e\}}$ is the indicator function of the expression $e$, i.e., 1 if $e$ holds, and 0 otherwise.} is the number of non-zero elements in the $i^{\text{th}}$ row of the constraint matrix $A$, and $M=\min\{m, n\}$.
\end{lemma}
\iffullpaper
\begin{proof}
Fix $\alpha$. Let:
\begin{equation*}\mu_j({\alpha}) = 
\begin{cases} \left(\frac{w_j}{w_{\max}M}\min_{i: A_{ij}\neq 0}\frac{1}{n_i A_{ij}}\right)^{1/\alpha}, & \mbox{if } \alpha\leq 1 \\
{{A_{\max}}}^{(1-\alpha)/\alpha}\left(\frac{w_j}{w_{\max}M}\right)^{1/\alpha}\min_{i: A_{ij}\neq 0}\frac{1}{n_i A_{ij}}, & \mbox{if } \alpha> 1 \end{cases}.
\end{equation*}

For the purpose of contradiction, suppose that $x^* = (x_1^*,...,x_n^*)$ is the optimal solution to $(P_\alpha)$, and $x_j^* < \mu_j(\alpha)$ for some fixed $j\in\{1,...,n\}$.  

To establish the desired result, we will need to introduce additional notation. We first break the set of (the indices of) constraints of the form $Ax\leq 1$ in which variable $x_j$ appears with a non-zero coefficient into two sets, $U$ and $T$:
\begin{itemize}
\item Let $U$ denote the set of the constraints from $(P_\alpha)$ that are not tight at the given optimal solution $x^*$, and are such that $A_{u, j}\neq 0$ for $u\in U$. Let $s_u = {1-\sum_{k=1}^n A_{uk}x_k}$ denote the slack of the constraint $u\in U$.
\item Let $T$ denote the set of tight constraints from $(P_\alpha)$ that are such that $A_{tj}\neq 0$ for $t\in T$. Observe that since $x^*$ is assumed to be optimal, $T\neq \emptyset$.
\end{itemize}

Let $\varepsilon_j = \min\left\{\mu_j(\alpha)-x_j^*, \min_{u\in U}s_u/A_{uj}\right\}$. Notice that by increasing $x_j$ to $x_j^* + \varepsilon_j$ none of the constraints from $U$ can be violated (although all the constraints in $T$ will; we deal with these violations in what follows).

In each constraint $t\in T$, there must exist at least one variable $x_k$ such that $x_k^*>\dfrac{1}{n_tA_{tk}}$, because $\sum_{l=1}^n A_{tl} x_{l}^*=1$, as each $t\in T$ is tight, and $
x_j^*<\mu_j(\alpha)\leq\min_{i: A_{ij}\neq 0}\frac{1}{n_i A_{ij}}\leq \frac{1}{n_t A_{tj}}.
$
Select one such $x_k$ in each constraint $t\in T$, and denote by $K$ the set of indices of selected variables. Observe that $|K|\leq |T|$ ($\leq M$), since an $x_k$ can appear in more than one constraint. 

For each $k\in K$, let $T_k$ denote the constraints in which $x_k$ is selected, and let 
\begin{equation}
\varepsilon_k = \max_{t\in T_k: A_{tk}\neq 0} \dfrac{A_{tj}\varepsilon_j}{A_{tk}}.\label{eq:epsilon-choice}
\end{equation}
If we increase $x_j$ by $\varepsilon_j$ and decrease $x_k$ by $\varepsilon_k$ $\forall k\in K$, each of the constraints $t\in T$ will be satisfied since, from (\ref{eq:epsilon-choice}) and from the fact that only one $x_k$ gets selected per constraint $t\in T$, $\varepsilon_j A_{tj}-\sum_{k\in K}\varepsilon_k A_{tk}\leq 0$. Therefore, to construct an alternative feasible solution $x'$, we set $x'_{j}=x_j^*+\varepsilon_j$, $x'_k = x_k^*-\varepsilon_k$ for $k\in K$, and $x'_l = x_l^*$ for all the remaining coordinates $l\in\{1,...,n\}\backslash (K \cup\{j\})$. 

Since $j$ is the only coordinate over which $x$ gets increased in $x'$, all the constraints $Ax'\leq 1$ are satisfied. For $x'$ to be feasible, we must have in addition that $x_k'\geq 0$ for $k\in K$. 
We show that $x_k' = x^*_k - \varepsilon_k \geq 0$ as follows:
\begin{align*}
\varepsilon_k &= \varepsilon_j \cdot \max_{t\in T_k: A_{tk}\neq 0}\frac{A_{tj}}{A_{tk}}\\
&\leq \mu_j(\alpha)\cdot \max_{t\in T_k: A_{tk}\neq 0}\frac{A_{tj}}{A_{tk}}\\
& \leq \min_{i: A_{ij}\neq 0}\frac{1}{n_i A_{ij}}\cdot \max_{t\in T_k: A_{tk}\neq 0}\frac{A_{tj}}{A_{tk}}\\
& \leq \max_{t\in T_k: A_{tk}\neq 0}\frac{1}{n_t A_{tj}}\frac{A_{tj}}{A_{tk}}\\
&\leq \max_{t\in T_k: A_{tk}\neq 0}\frac{1}{n_t A_{tk}}\\
&< x_k^*,
\end{align*}
where the second line follows from $\varepsilon_j \leq \mu_j(\alpha)-x_j^*\leq \mu_j(\alpha)$, and the last line follows from the choice of $x_k$.

The last part of the proof is to show that $\sum_{l=1}^n w_l\frac{x_{l}'-x_l^*}{{x_l^*}^{\alpha}}>0$, which contradicts the initial assumption that $x^*$ is optimal, by the definition of $\alpha$-fairness from Section \ref{section:prelims}. We have that:
\begin{align}
\sum_{l=1}^n w_l\frac{x_{l}'-x_l^*}{{x_l^*}^{\alpha}} &= w_j\frac{\varepsilon_j}{{x_j^*}^{\alpha}} - \sum_{k\in K}w_k \frac{\varepsilon_k}{{x_k^*}^{\alpha}}\notag\\
&= \sum_{k\in K} \left(w_j\frac{\varepsilon_j}{{x_j^*}^{\alpha}|K|} - w_k \frac{\varepsilon_k}{{x_k^*}^{\alpha}}\right)\notag\\
&= \sum_{k\in K} \left(\frac{w_j\varepsilon_j{x_k^*}^{\alpha} - w_k\varepsilon_k{x_j^*}^{\alpha}|K|}{{x_j^*}^{\alpha}{x_k^*}^{\alpha}|K|}\right) . \label{eq:sum-ind-terms}
\end{align}
Consider one term from the summation (\ref{eq:sum-ind-terms}). From the choice of $\varepsilon_k$'s, we know that for each $\varepsilon_k$ there exist $t\in T$ such that $\varepsilon_k = \dfrac{\varepsilon_jA_{tj}}{A_{tk}}$, and at the same time (by the choice of $x_k$) we have $x_k^* > \dfrac{1}{n_t A_{tk}}$, so that
\begin{equation}
w_j\varepsilon_j{x_k^*}^{\alpha}> w_j\dfrac{\varepsilon_k A_{tk}}{A_{tj}}\left(\frac{1}{A_{tk}n_t}\right)^{\alpha}> \frac{w_kw_j\varepsilon_k }{w_{\max}}\dfrac{A_{tk}}{A_{tj}}\left(\frac{1}{A_{tk}n_t}\right)^{\alpha}.\label{eq:xk-general}
\end{equation}

\noindent\textbf{Case 1.} Suppose first that $\alpha\leq 1$. Then ${x_k^*}^{\alpha} > \left(\frac{1}{A_{tk}n_t}\right)^{\alpha}\geq \frac{1}{A_{tk}n_t}$, as $A_{tk}\neq 0\Rightarrow A_{tk}\geq 1$.  Plugging into (\ref{eq:xk-general}), we have:
\begin{equation}
w_j\varepsilon_j{x_k^*}^{\alpha}> \frac{w_kw_j\varepsilon_k}{w_{\max}}\dfrac{ 1}{n_tA_{tj}}\label{eq:xk}.
\end{equation}

By the initial assumption, $x_j^* < \mu_j(\alpha) = \left(\frac{w_j}{w_{\max}M}\min_{i: A_{ij}\neq 0}\frac{1}{n_i A_{ij}}\right)^{1/\alpha}$, and therefore 
\begin{equation}
w_k\varepsilon_k{x_j^*}^{\alpha}|K|< \frac{w_k w_j \varepsilon_k }{w_{\max}}\frac{|K|}{M}\min_{i: A_{ij}\neq 0}\frac{1}{n_i A_{ij}}\leq \frac{w_k w_j \varepsilon_k }{w_{\max}}\frac{1}{n_t A_{tj}}, \label{eq:xj}
\end{equation}
since it must be $|K|\leq M$ ($=\min\{m, n\}$). From (\ref{eq:xk}) and (\ref{eq:xj}), we get that every term in the summation (\ref{eq:sum-ind-terms}) is strictly positive, which implies:
\begin{equation*}
\sum_{l=1}^n w_l\frac{x_{l}'-x_l^*}{{x_l^*}^{\alpha}}>0,
\end{equation*}
and therefore $x^*$ is not optimal.

\noindent\textbf{Case 2.} Now suppose that $\alpha>1$. Then 
\begin{equation*}
x_j^*<\mu_j(\alpha) = {{A_{\max}}}^{(1-\alpha)/\alpha}\left(\frac{w_j}{w_{\max}M}\right)^{1/\alpha}\min_{i: A_{ij}\neq 0}\frac{1}{n_i A_{ij}} \leq {{A_{\max}}}^{(1-\alpha)/\alpha}\left(\frac{w_j}{w_{\max}M}\right)^{1/\alpha}\frac{1}{n_t A_{tj}}.
\end{equation*}
Therefore:
\begin{align}
w_k\varepsilon_k{x_j^*}^{\alpha}|K|&< w_k\varepsilon_k \frac{w_j}{w_{\max}M} {A_{\max}}^{1-\alpha}\left(\frac{1}{n_tA_{tj}}\right)^{\alpha}|K|\notag\\
&\leq {w_k}\frac{w_j}{w_{\max}}{A_{\max}}^{1-\alpha}\varepsilon_k\left(\frac{1}{A_{tk}n_t}\right)^{\alpha}\dfrac{ {A_{tk}}^{\alpha}}{{A_{tj}}^{\alpha}}\notag\\
&={w_k}\frac{w_j}{w_{\max}}\frac{\varepsilon_k A_{tk}}{A_{tj}}\cdot\frac{(A_{tk}/A_{tj})^{\alpha-1}}{{A_{\max}}^{\alpha-1}}\left(\frac{1}{A_{tk}n_t}\right)^{\alpha}\notag\\
&\leq {w_k}\frac{w_j}{w_{\max}}\dfrac{\varepsilon_k A_{tk}}{A_{tj}}\left(\frac{1}{A_{tk}n_t}\right)^{\alpha},\label{eq:xj-alpha1}
\end{align}
as $|K|\leq M$, and $\frac{A_{tk}}{A_{tj}}\leq A_{\max}$ (since for any $i, j$: $1\leq A_{ij}\leq {A_{\max}}$).

Finally, from (\ref{eq:xk-general}) and (\ref{eq:xj-alpha1}) we get that every term in the summation (\ref{eq:sum-ind-terms}) is positive, which yields a contradiction.
\end{proof}
\fi

%%%%%%%%%%%%%%%%%%%%%%%%%%%%%%%%%%%%%%%%%%%%%%%%%%%%%%%%%%%%%%%%%%%%%%%%%%%%%%%%%%%%%%%%%%%%
\paragraph{Asymptotics of $\alpha-$Fair Allocations}

The following lemma states that for sufficiently small (but not too small) $\alpha$, the values of the linear and the $\alpha-$fair objectives at their respective optimal solutions are approximately the same. This statement will then lead to a conclusion that to $\varepsilon-$approximately solve an $\alpha-$fair packing problem for a very small $\alpha$, one can always use an $\varepsilon-$approximation packing LP algorithm. 
\begin{lemma}\label{lemma:LP-close-to-small-alpha-fair}
Let $(P_{\alpha})$ be an $\alpha-$fair packing problem with optimal solution $x^*$, and $(P_0)$ be the LP with the same constraints and the same weights $w$ as $(P_{\alpha})$ and an optimal solution $z^*$. Then if $\alpha \leq \frac{\varepsilon/4}{\ln(nA_{\max}/\varepsilon)}$, we have that $\sum_j w_j z_j^* \geq (1-3\varepsilon)\sum_j \frac{(x_j^*)^{1-\alpha}}{1-\alpha}$, where $\varepsilon\in(0, 1/6]$.
\end{lemma}
\iffullpaper
\begin{proof}
The proof outline is as follows. First, we show that the $\alpha-$fair objective $p_\alpha(x^*)$ can be upper-bounded by a linear objective as $p_\alpha(x^*)\equiv \sum_j w_j \frac{{x_j^*}^{1-\alpha}}{1-\alpha}\leq (1+O(\varepsilon))\sum_j w_j x_j^*$. Then, to complete the proof, we use the optimality of $z^*$ for the LP: $\sum_j w_j z_j^*\geq \sum_j w_j x_j^*$ ($\geq (1-O(\varepsilon)) \sum_j w_j \frac{{x_j^*}^{1-\alpha}}{1-\alpha}$ from the first part of the proof).

Let $g(x_j) = \frac{{x_j}^{1-\alpha}}{1-\alpha} - (1+\varepsilon)x_j$. Consider the case when $g(x_j)\leq 0$. Solving $g(x_j)\leq 0$ for $x_j$, we get that it should be 
\begin{equation}
x_j \geq \Big(\frac{1}{1-\alpha}\Big)^{1/\alpha}\cdot \Big(\frac{1}{1+\varepsilon}\Big)^{1/\alpha}.\label{eq:x-j-alpha-c}
\end{equation}

Choose $\alpha$ so that $\frac{1}{(1+\varepsilon)^{1/\alpha}} \leq \big(\frac{\varepsilon/4}{nA_{\max}}\big)$, which is equivalent to $\alpha \leq \frac{\ln(1+\varepsilon)}{\ln(4nA_{\max}/\varepsilon)}$. Then to have $g(x_j)\leq 0$, it suffices to have $x_j \geq \frac{\varepsilon}{nA_{\max}}$, because (i) $\big(\frac{1}{1-\alpha}\big)^{1/\alpha} \in [e, 4]$ for $\alpha \in [0, 1/2]$, where $e$ is the base of the natural logarithm, and (ii) $\frac{1}{(1+\varepsilon)^{1/\alpha}} \leq \big(\frac{\varepsilon/4}{nA_{\max}}\big)$ by the choice of $\alpha$.

Now, as $\alpha \leq \frac{\ln(1+\varepsilon)}{\ln(4nA_{\max}/\varepsilon)}$, summing over $j$ such that $x_j^*\geq \frac{\varepsilon}{nA_{\max}}$ we have:
\begin{align}
\sum_{j: x_j^*\geq \frac{\varepsilon}{nA_{\max}}}w_j \frac{(x_j^*)^{1-\alpha}}{1-\alpha} - (1+\varepsilon)\sum_{j: x_j^*\geq \frac{\varepsilon}{nA_{\max}}}w_j x_j^* = \sum_{j: x_j^*\geq \frac{\varepsilon}{nA_{\max}}}w_j g(x_j^*)
\leq 0 \label{eq:large-coord-approx}
\end{align}
Now we bound the rest of the terms in $p_\alpha(x^*)$, i.e., we consider $j: x_j^* < \frac{\varepsilon}{nA_{\max}}$. Observe that since $x_j = \frac{1}{nA_{\max}}$ for $j=\{1,...,n\}$ is a feasible solution to $(P_{\alpha})$ and $x^*$ is the optimal solution to $(P_{\alpha})$, we have that $\sum_j w_j \frac{(1/nA_{\max})^{1-\alpha}}{1-\alpha}\leq \sum_j w_j \frac{(x_j^*)^{1-\alpha}}{1-\alpha}$, which gives:
\begin{align}
\sum_{j: x_j^*< \frac{\varepsilon}{nA_{\max}}}w_j \frac{(x_j^*)^{1-\alpha}}{1-\alpha} &< \varepsilon^{1-\alpha}\sum_{j: x_j^*< \frac{\varepsilon}{nA_{\max}}}w_j\frac{(1/nA_{\max})^{1-\alpha}}{1-\alpha }\notag\\
& < \varepsilon^{1-\alpha} \sum_{j=1}^n w_j \frac{(x_j^*)^{1-\alpha}}{1-\alpha}\notag\\
&\leq 2\varepsilon \sum_{j=1}^n w_j \frac{(x_j^*)^{1-\alpha}}{1-\alpha} \notag.
\end{align}
Therefore:
\begin{align}
\sum_{j: x_j^*\geq \frac{\varepsilon}{nA_{\max}}}w_j \frac{(x_j^*)^{1-\alpha}}{1-\alpha} > (1-2\varepsilon) \sum_{j=1}^n w_j \frac{(x_j^*)^{1-\alpha}}{1-\alpha}.\label{eq:small-coord-approx}
\end{align}
Combining (\ref{eq:large-coord-approx}) and (\ref{eq:small-coord-approx}), we now get:
\begin{align}
\sum_{j=1}^n w_j \frac{(x_j^*)^{1-\alpha}}{1-\alpha} &< \frac{1+\varepsilon}{1-2\varepsilon}\cdot \sum_{j: x_j^*\geq \frac{\varepsilon}{nA_{\max}}}w_j x_j^*. \label{eq:alpha-LP}
\end{align}
Finally, since $z^*$ optimally solves $(P_0)$ (which has the same constraints and weights as $(P_{\alpha})$), we have that $x^*$ is feasible for $(P_0)$, and using (\ref{eq:alpha-LP}) and optimality of $z^*$, it follows that:
\begin{align}
\sum_{j=1}^n w_j z_j^* &\geq \sum_{j=1}^n w_j x_j^* \notag\\
&\geq \frac{1-2\varepsilon}{1+\varepsilon} \sum_{j=1}^n w_j \frac{(x_j^*)^{1-\alpha}}{1-\alpha}\notag\\
&\geq (1-3\varepsilon) \sum_{j=1}^n w_j \frac{(x_j^*)^{1-\alpha}}{1-\alpha},\notag
\end{align}
as claimed.
\end{proof}
\fi
Observing that for any $\alpha\in(0, 1)$, $\frac{(z_j^*)^{1-\alpha}}{1-\alpha}\geq z_j^*$ (since, due to the scaling, $z_j^*\in[0, 1]$), a simple corollary of Lemma \ref{lemma:LP-close-to-small-alpha-fair} is that an $\varepsilon-$approximation $z$ to $(P_0)$ ($\sum_j w_j z_j \geq (1-\varepsilon)\sum_j w_j z_j^*$) is also an $O(\varepsilon)-$approximation to $(P_{\alpha})$, for $\alpha \leq \frac{\varepsilon/4}{\ln(nA_{\max}/\varepsilon)}$. Thus, to find an $\varepsilon-$approximate solution for $\alpha \leq \frac{\varepsilon/4}{\ln(nA_{\max}/\varepsilon)}$, the packing LP algorithm of \cite{AwerbuchKhandekar2009} can be run, which means that there is a stateless distributed algorithm that converges in poly($\ln(\varepsilon^{-1}\wratio mnA_{\max})/\varepsilon$) time for $\alpha$ arbitrarily close to zero.

The following two lemmas show that when $\alpha$ is sufficiently close to 1, $(P_{\alpha})$ can be $\varepsilon-$approximated by $\varepsilon-$approximately solving $(P_1)$ with the same constraints and weights.

\begin{lemma}\label{lemma:alpha-close-to-1-below}
Let $x$ be an $\varepsilon-$approximate solution to a 1-fair packing problem $(P_1)$ returned by \textsc{$\alpha$-FairPSolver}. Then, for any $\alpha \in \left[1- 1/{\tau_0}, 1\right)$, where $\tau_0 = \frac{1}{\beta}\ln(\frac{1}{\delta_{\min}})$, $x$ is also a $2\varepsilon-$approximate solution to $(P_\alpha)$, where the only difference between $(P_1)$ and $(P_\alpha)$ is in the value of $\alpha$ in the objective.
\end{lemma}
\iffullpaper
\begin{proof}
Suppose that $x$ is a solution in some stationary round, provided by \textsc{$\alpha$-FairPSolver} run for $\alpha = 1$. Fix that round. It is clear that if $x$ is feasible in $(P_1)$, it is also feasible in $(P_\alpha)$, since all the constraints in $(P_1)$ and $(P_\alpha)$ are the same by the initial assumption. All that is required for a dual solution $y$ to be feasible is that $y_i \geq 0$, for all $i$, and therefore $y(x)$ is a feasible dual solution for $(P_\alpha)$. The rest of the proof follows by bounding the duality gap $G_\alpha(x, y(x))$. Recall from (\ref{eq:duality-gap-alpha}) that: 
\begin{align}\label{eq:duality-gap-alpha-recap}
G_{\alpha}(x, y(x)) = \sum_{j=1}^n w_j\frac{x_j^{1-\alpha}}{1-\alpha}\bigg(\Big(\frac{{x_j}^{\alpha}\sum_{i=1}^m y_i A_{ij}}{w_j}\Big)^{\frac{\alpha-1}{\alpha}}-1\bigg) +\sum_{i=1}^m y_i - \sum_{j=1}^n  w_j x_j^{1-\alpha}\cdot \Big(\frac{{x_j}^{\alpha}\sum_{i=1}^m A_{ij}y_i}{w_j}\Big)^{\frac{\alpha-1}{\alpha}}.
\end{align}

Since $x$ is a solution from a stationary round, from the second part of the definition of a stationary round (Definition \ref{def:alpha=1-stat-round}), we have that:
\begin{align*}
\sum_{j=1}^n x_j \sum_{i=1}^n y_i(x)A_{ij} \leq (1+2\gamma)\sum_{k=1}^n w_k.
\end{align*}
Further, from Lemma \ref{lemma:approx-comp-slack}:
\begin{align}
\sum_{i=1}^m y_i(x) \leq (1+3\varepsilon)\sum_{j=1}^n x_j \sum_{i=1}^n y_i(x)A_{ij} \leq (1+3\varepsilon)(1+2\gamma)\sum_{k=1}^n w_k. \label{eq:alpha-below-1-acs}
\end{align}

Next, we show that:
\begin{align}
{x_j}^{1-\alpha} \geq 1- \gamma,\quad \forall j. \label{eq:alpha-below=1-all-x-large}
\end{align}
Rearranging the terms and taking logarithms of both sides in (\ref{eq:alpha-below=1-all-x-large}), we obtain the  equivalent inequality $1-\alpha \leq \frac{\ln(1/(1-\gamma))}{\ln(1/x_j)}$. Recall from \textsc{$\alpha$-FairPSolver} that in every (except for, maybe, the first) round $x_j \geq \delta_j \geq \delta_{\min}$. As $\ln(1/(1-\gamma))\geq \gamma$, it therefore suffices to show that $1-\alpha \leq \frac{\gamma}{\ln(1/\delta_{\min})}$. But from the statement of the lemma, $1-\alpha\leq 1/\tau_0 < \frac{\gamma}{\ln(1/\delta_{\min})}$, completing the proof of (\ref{eq:alpha-below=1-all-x-large}).

Combining (\ref{eq:alpha-below-1-acs}) and (\ref{eq:alpha-below=1-all-x-large}), we get that:
\begin{align}\label{eq:alpha-below-1-bnd-1}
\sum_{i=1}^m y_i(x) \leq \frac{(1+3\varepsilon)(1+2\gamma)}{1-\gamma}\sum_{j=1}^n w_j {x_j}^{1-\alpha} \leq (1+5\varepsilon) \sum_{j=1}^n w_j {x_j}^{1-\alpha},
\end{align}
where the second inequality follows from $\varepsilon\leq 1/6$, $\gamma = \varepsilon/4$.

Using (\ref{eq:alpha-below-1-bnd-1}), we can bound the duality gap (Eq. (\ref{eq:duality-gap-alpha-recap})) as:
\begin{align}\label{eq:duality-gap-alpha-betterbound}
G_{\alpha}(x, y(x)) \leq \sum_{j=1}^n w_j \frac{{x_j}^{1-\alpha}}{1-\alpha} \left(\alpha\Big(\frac{{x_j}^{\alpha}\sum_{i=1}^m y_i A_{ij}}{w_j}\Big)^{\frac{\alpha-1}{\alpha}}-1 + (1-\alpha)(1+5\varepsilon)\right).
\end{align}
To complete the proof, recall from Lemma \ref{lemma:cond-lower-bound} that in any round of the algorithm, for all $j$: $ \frac{x_j \sum_{i=1}^m y_i(x)A_{ij}}{w_j} \geq (1-\gamma)^{\tau_0}$. As $\alpha < 1$ and $x_j \in [0, 1]$, $\forall j$, it holds that ${x_j}^{\alpha} \geq x_j$, $\forall j$, and therefore:
\begin{align}\label{eq:alpha-below-1-bnd-2}
\frac{{x_j}^{\alpha} \sum_{i=1}^m y_i(x)A_{ij}}{w_j} \geq (1-\gamma)^{\tau_0}, \quad \forall j.
\end{align}
Finally, recalling that $1-\alpha \leq 1/\tau_0$, and combining (\ref{eq:alpha-below-1-bnd-2}) with (\ref{eq:duality-gap-alpha-betterbound}), we get:
\begin{align}
G_{\alpha}(x, y(x))&\leq \sum_{j=1}^n w_j \frac{{x_j}^{1-\alpha}}{1-\alpha} \left(\alpha\Big(\frac{1}{1-\gamma}\Big)^{1/\alpha} -1 + (1-\alpha)(1+5\varepsilon) \right)\notag\\ 
&\leq \sum_{j=1}^n w_j \frac{{x_j}^{1-\alpha}}{1-\alpha}((1 + {2\gamma})^{1/\alpha} - 1  + (1-\alpha)(1+5\varepsilon))\notag\\
&\leq \sum_{j=1}^n w_j \frac{{x_j}^{1-\alpha}}{1-\alpha} (1 + \varepsilon - 1  + (1-\alpha)(1+5\varepsilon)) \notag\\
&\leq 2\varepsilon \sum_{j=1}^n w_j \frac{{x_j}^{1-\alpha}}{1-\alpha},\notag
\end{align}
where the third inequality follows from $\alpha \geq 1 - 1/\tau_0 \geq 1 - \frac{\gamma \varepsilon}{5}\geq 1 - \frac{\varepsilon^2}{20}$, and the fourth inequality follows from $1-\alpha < \varepsilon/2$ and $\varepsilon \leq 1/6$.
\end{proof}
\fi

\begin{lemma}\label{lemma:alpha-close-to-1-above}
Let $x$ be an $\varepsilon-$approximate solution to a 1-fair packing problem $(P_1)$ returned by \textsc{$\alpha$-FairPSolver}. Then, for any $\alpha \in (1, 1 + 1/{\tau_0}]$, where $\tau_0 = \frac{1}{\beta}\ln(\frac{1}{\delta_{\min}})$, $x$ is also a $2\varepsilon-$approximate solution to $(P_\alpha)$, where the only difference between $(P_1)$ and $(P_\alpha)$ is in the value of $\alpha$ in the objective.
\end{lemma}
\iffullpaper
\begin{proof}
Similar to the proof of Lemma \ref{lemma:alpha-close-to-1-below}, we will fix an $x$ from some stationary round of \textsc{$\alpha$-FairPSolver} run on $(P_1)$, and argue that the same $x$ $2\varepsilon-$approximates $(P_\alpha)$ by bounding the duality gap $G_{\alpha}(x, y(x))$, although we will need to use a different set of inequalities since now $\alpha > 1$. Similar to the proof of Lemma \ref{lemma:alpha-close-to-1-below}, as $x$ is (primal-)feasible for $(P_1)$, $x$ and $y(x)$ are primal- and dual-feasible for $(P_\alpha)$.

By the same token as in the proof of Lemma \ref{lemma:alpha-close-to-1-below}:
\begin{align*}
\sum_{i=1}^m y_i (x) \leq (1+3\varepsilon)(1+2\gamma)\sum_{j=1}^n w_j.
\end{align*}
As $\alpha > 1$ and $x_j \in (0, 1]$, $\forall j$, we have that ${x_j}^{1-\alpha}\geq 1$, $\forall j$, and therefore:
\begin{align}\label{eq:alpha-above-1-bnd-1}
\sum_{i=1}^m y_i (x) \leq (1+3\varepsilon)(1+2\gamma)\sum_{j=1}^n w_j {x_j}^{1-\alpha} \leq (1+4\varepsilon)\sum_{j=1}^n w_j {x_j}^{1-\alpha}.
\end{align}
Therefore, we can write for the duality gap:
\begin{align}
G_{\alpha}(x, y(x))&\leq \sum_{j=1}^n w_j \frac{{x_j}^{1-\alpha}}{1-\alpha} \left(\alpha\Big(\frac{{x_j}^{\alpha}\sum_{i=1}^m y_i A_{ij}}{w_j}\Big)^{\frac{\alpha-1}{\alpha}}-1 + (1-\alpha)(1+4\varepsilon)\right)\\
&= -\sum_{j=1}^n w_j \frac{{x_j}^{1-\alpha}}{1-\alpha} \left(-\alpha\Big(\frac{{x_j}^{\alpha}\sum_{i=1}^m y_i A_{ij}}{w_j}\Big)^{\frac{\alpha-1}{\alpha}} + 1 + (\alpha-1)(1+4\varepsilon)\right)
\label{eq:alpha-above-1-duality-gap}.
\end{align}
Notice that, as $\alpha >1$, the objective for $(P_\alpha)$, $\sum_{j=1}^n w_j \frac{{x_j}^{1-\alpha}}{1-\alpha}$, is now negative.

Using the same arguments as in the proof of Lemma \ref{lemma:alpha-close-to-1-below}, it is straightforward to show that ${x_j}^{\alpha-1}\geq 1-\gamma$, $\forall j$. From Lemma \ref{lemma:cond-lower-bound}, we have that $\frac{x_j\sum_{i}y_i(x)A_{ij}}{w_j}\geq (1-\gamma)^{\tau_0}$, $\forall j$, and therefore:
\begin{align}
\frac{{x_j}^{\alpha}\sum_{i=1}^m y_i(x)A_{ij}}{w_j} &= \frac{{x_j}^{1-\alpha}\cdot x_j\sum_{i=1}^my_i(x)A_{ij}}{w_j}\notag\\ 
&\geq (1-\gamma)^{\tau_0 + 1}. \label{eq:alpha-above-1-bnd-2}
\end{align}
Recalling that $\alpha - 1 \leq 1/\tau_0$ (by the statement of the lemma) and using (\ref{eq:alpha-above-1-bnd-2}), we have:
\begin{align}
\Big(\frac{{x_j}^{\alpha}\sum_{i=1}^m y_i A_{ij}}{w_j}\Big)^{\frac{\alpha-1}{\alpha}} &\geq (1-\gamma)^{(\tau_0+1)/(\tau_0(1+1/\tau_0))}\notag\\ 
&=(1-\gamma). \label{eq:alpha-above-1-bnd-3}
\end{align}
Finally, plugging (\ref{eq:alpha-above-1-bnd-3}) into (\ref{eq:alpha-above-1-duality-gap}), we have:
\begin{align}
G_{\alpha}(x, y(x))&\leq -\sum_{j=1}^n w_j \frac{{x_j}^{1-\alpha}}{1-\alpha} \left(-\alpha(1-\gamma) + 1 + (\alpha-1)(1+4\varepsilon)\right)\notag \\
%&\leq -\sum_{j=1}^n w_j \frac{{x_j}^{1-\alpha}}{1-\alpha} \left(\alpha\cdot\frac{9}{2}\varepsilon - 4\varepsilon\right)\notag \\
&= -\sum_{j=1}^n w_j \frac{{x_j}^{1-\alpha}}{1-\alpha} \left(\alpha\cdot\frac{1}{4}\varepsilon + 4\varepsilon(\alpha - 1)\right)\notag \\
&\leq -\varepsilon \sum_{j=1}^n w_j \frac{{x_j}^{1-\alpha}}{1-\alpha},
\end{align}
where the equality follows from $\gamma = \frac{\varepsilon}{4}$, and the last inequality follows from $\alpha - 1 \leq \frac{1}{\tau_0} < \frac{\varepsilon}{20}$.
\end{proof}
\fi

Finally, we consider the asymptotics of $\alpha-$fair allocations, as $\alpha$ becomes large. This result complements the result from \cite{MoWalrand2000} that states that $\alpha-$fair allocations approach the max-min fair one as $\alpha\rightarrow \infty$ by showing how fast the max-min fair allocation is reached as a function of $\alpha, \wratio, n$, and $A_{\max}$. First, for completeness, we provide the definition of max-min fairness.
\begin{definition}(Max-min fairness \cite{Bertsekas:1987:DN:12517}.) \label{def:max-min-fairness}
Let $\mathcal{R}\subset \mathbb{R}_+^n$ be a compact and convex set. A vector $x\in \mathcal{R}$ is max-min fair on $\mathcal{R}$ if for any vector $z\in \mathcal{R}$ it holds that: if for some $j\in\{1,...,n\}$ $z_j > x_j$, then there exists $k\in\{1,...,n\}$ such that $z_k < x_k$ and $x_k \leq x_j$. 
\end{definition}
On a compact and convex set $\mathcal{R}\subset \mathbb{R}^n$, the max-min fair vector is unique (see, e.g., \cite{Sarkar-Tassiulas, radunovic2007unified}). The following lemma shows that for $\alpha \geq \varepsilon^{-1}\ln(\wratio n A_{\max})$, the $\alpha-$fair vector and the max-min fair vector are $\varepsilon-$close to each other. Notice that because of a very large gradient of $p_\alpha(x)$ as $\alpha$ becomes large, the max-min fair solution provides only an $O(\varepsilon\alpha)-$approximation to $(P_\alpha)$. 
%We will work with the following definition of $\varepsilon-$approximate max-min fair vectors (also called coordinate-wise $(1+\varepsilon)-$approximate max-min fair vectors), which is due to Kleinberg, Rabani, and Tardos \cite{kleinberg1999fairness}:

% \begin{definition}\label{def:eps-approx-max-min-fairness}
% Let $\mathcal{R}\subset \mathbb{R}^n$ be a compact and convex set, and let $x$ be the max-min fair vector on $\mathcal{R}$. A vector $z$ is $\varepsilon-$approximate max-min fair on $\mathcal{R}$ if the $j^{\text{th}}$ largest coordinate of $x$ is no more than a factor $1+\varepsilon$ larger than the $j^{\text{th}}$ largest coordinate of $z$, for all $j\in\{1,..., n\}$.
% \end{definition}

\begin{lemma}\label{lemma:mmf-alpha-fair}
Let $x^*$ be the optimal solution to $(P_\alpha) = \max\{p_\alpha(x): Ax\leq 1, x\geq 0\}$, $z^*$ be the max-min fair solution for the convex and compact set determined by the constraints from $(P_\alpha)$. Then if $\alpha \geq {\varepsilon}^{-1} \ln\left(\wratio n A_{\max}\right)$, we have that:
\begin{enumerate}
\item $p_\alpha(x^*) \leq (1-{\varepsilon}(\alpha - 1))p_\alpha(z^*)$, i.e., $z^*$ is an ${\varepsilon}(\alpha-1)-$approximate solution to $(P_\alpha)$, and
\item $(1-\varepsilon)z_j^* \leq x_j^* \leq (1+\varepsilon)z_j^*$, for all $j \in \{1,...,n\}$.
\end{enumerate}
\end{lemma}
\iffullpaper
\begin{proof}
%Let $x^*$ be the optimal solution to $(P_\alpha)$, $z^*$ be the max-min fair solution for the convex and compact set determined by the constraints from $(P_\alpha)$. 
%
Suppose that, starting with $z^*$, we want to construct a solution $z$ that is feasible in $(P_\alpha)$ and is such that $p_\alpha(z) > p_\alpha(z^*)$. Then we need to increase at least one coordinate $j$ of $z^*$. Suppose that we increase a coordinate $j$ by a factor $1+\varepsilon$, so that $z_j = (1+\varepsilon)z_j^*$. Since $z^*$ is the max-min fair vector, to keep $z$ feasible, the increase over the $j^{\text{th}}$ coordinate must be at the expense of decreasing some other coordinates $k$ that satisfy $z_k^* \leq z_j^*$. We will assume that whenever we decrease the coordinates to keep the solution feasible, we keep the solution Pareto optimal (i.e., we decrease the selected coordinates by a minimum amount). %Let $z^1$ be an intermediary feasible solution between $z^*$ and $z$ such that $z_j^1 = (1+\varepsilon/2)z_j^*$. 
Using Fact \ref{fact:taylor}, we have:
\begin{align}
p_\alpha(z) - p_\alpha(z^*) &\leq \sum_{l=1}^n w_l\frac{z_l - z_l^*}{(z_l^*)^{\alpha}}< w_j\frac{z_j - z_j^*}{(z_j^*)^{\alpha}} = {\varepsilon}\cdot w_j(z_j^*)^{1-\alpha}. \label{eq:mmf-part-1}
\end{align}

Now, suppose that we want to further increase the $j^{\text{th}}$ coordinate by some small $\delta$. Call that new solution $z^1$. Then, the total amount by which other coordinates must decrease to keep the solution feasible is at least $\frac{\delta}{A_{\max}}$, since the feasible region is determined by packing constraints and it must be $Az\leq 1$, where $1\leq A_{ij}\leq A_{\max}$, $\forall i, j$. Moreover, since $z^*$ is max-min fair, each coordinate $k$ that gets decreased must satisfy $z_k^* \leq z_j^*$. It follows that:
\begin{align}
p(z^1) - p(z)&\leq \sum_{l=1}^n w_l\frac{z_l^1 - z_l}{(z_l)^{\alpha}}\notag\\
&= w_j \frac{\delta}{(1+\varepsilon)^{\alpha}(z_j^*)^{\alpha}} + \sum_{k: z_k^1 < z_k} w_k\frac{z_k^1 - z_k}{(z_k)^{\alpha}}\notag\\
&\leq w_{\max} \frac{\delta}{(1+\varepsilon)^{\alpha}(z_j^*)^{\alpha}} - w_{\min} \frac{\delta/A_{\max}}{(z_j^*)^{\alpha}}\notag\\
&= \frac{\delta(w_{\max} - (1+\varepsilon)^{\alpha}w_{\min}/A_{\max})}{(1+\varepsilon)^{\alpha}(z_j^*)^{\alpha}}\notag\\
&\leq 0. \label{eq:mmf-part-2}
\end{align}
The last inequality can be verified by solving the inequality $w_{\max} - (1+\varepsilon)^{\alpha}w_{\min}/A_{\max}\leq 0$ for $\alpha$, and verifying that it is implied by the initial assumption that $\alpha \geq \varepsilon^{-1}\ln(\wratio n A_{\max})$.

Therefore, the maximum amount by which any coordinate of $z^*$ can be increased to improve the value of the objective $p_\alpha(.)$ is by a multiplicative factor of at most $(1+\varepsilon)$. Since we can construct $x^*$, the optimal solution to $(P_\alpha)$, starting with $z^*$ and by choosing a set of coordinates $j$ that we want to increase and by only decreasing coordinates $k$ such that $z_k^* \leq z_j^*$ whenever coordinate $j$ is increased, it follows that $x_j^* \leq (1+\varepsilon)z_j^*$, $\forall j$.

Moreover, from (\ref{eq:mmf-part-1}) and (\ref{eq:mmf-part-2}):
\begin{align*}
p_\alpha(z^1) - p_\alpha(z^*) = p(z^1) - p(z) + p(z) - p(z^*) < {\varepsilon}\cdot w_j(z_j^*)^{1-\alpha}, %\label{eq:mmf-part-3}
\end{align*}
and we can conclude that:
\begin{align*}
p_\alpha(x^*) - p_\alpha(z^*) < \sum_{j=1}^n {\varepsilon}\cdot w_j(z_j^*)^{1-\alpha} = {\varepsilon}(1-\alpha)\cdot p_\alpha(z^*),
\end{align*}
which means that $z^*$ is an ${\varepsilon}(\alpha-1)-$approximate solution to $(P_\alpha)$.

Now consider the coordinates we need to decrease when we construct a solution $z$ from $z^*$, such that $p_\alpha(z) > p_\alpha(z^*)$. Suppose that to increase some other coordinates, a coordinate $k$ is decreased by a factor $(1-\varepsilon)$: $z_k = (1-\varepsilon)z_k^*$. As $z^*$ is max-min fair, only coordinates larger than $z_k^*$ can increase at the expense of decreasing $z_k^*$. Suppose now that we decrease the $k^{\text{th}}$ coordinate further by some small $\delta$. Call that solution $z^1$. Then the maximum number of other coordinates $j$ that can further increase is $\min\{n-1, m\}< n$. Moreover, each coordinate $j$ that gets increased satisfies $z_j^*\geq z_k^*$, and can be increased by at most $A_{\max} \delta$. Using Fact \ref{fact:taylor}, it follows that: 
\begin{align}
p_\alpha(z^1) - p_\alpha(z) &\leq \sum_{l=1}^n w_l\frac{z_l^1 - z_l}{(z_l)^{\alpha}}\notag\\
&= -w_k \frac{\delta}{(1-\varepsilon)(z_k^*)^{\alpha}} + \sum_{j: z_j^1 > z_j} w_j \frac{z_j^1-z_j}{(z_j^*)^{\alpha}}\notag\\
&< -w_{\min} \frac{\delta}{(1-\varepsilon)(z_k^*)^{\alpha}} + n w_{\max} \frac{A_{\max}\delta}{(z_k^*)^{\alpha}}\notag\\
&= \frac{\delta(n w_{\max} A_{\max} (1-\varepsilon)^{\alpha} - w_{\min})}{(z_k^*)^{\alpha}}\notag\\
&\leq 0, \label{eq:mmf-part-5}
\end{align}
where the last inequality follows from $(1-\varepsilon)^{\alpha} \leq (\wratio n A_{\max})^{-1}$, which is implied by the initial assumption that $\alpha \geq {\varepsilon}^{-1}{\ln(\wratio n A_{\max})}$.

Therefore, using (\ref{eq:mmf-part-5}), the $k^{\text{th}}$ coordinate can decrease by at most a multiplicative factor $(1-\varepsilon)$. Using similar arguments as for increasing the coordinates, it follows that $x_j^* \geq (1-\varepsilon)z_j^*$, $\forall j$.
\end{proof}
\fi

%%%%%%%%%%%%%%%%%%%%%%%%%%%%%%%%%%%%%%%%%%%%%%%%%%%%%%%%%%%%%%%%%%%%%%%%%%%%%%%%%%%%%%%%%%%%
%%%%%%%%%%%%%%%%%%%%%%%%%%%%%%%%%%%%%%%%%%%%%%%%%%%%%%%%%%%%%%%%%%%%%%%%%%%%%%%%%%%%%%%%%%%%
\section{Conclusion}\label{section:conclusion}

We presented an efficient stateless distributed algorithm for the class of $\alpha$-fair packing problems. To the best of our knowledge, this is the %most efficient distributed algorithm for this problem. We obtained 
first algorithm with poly-logarithmic convergence time in the input size.~{Additionally, we obtained results that characterize the fairness and asymptotic behavior of allocations in weighted $\alpha-$fair packing problems that may be of independent interest.}~{An interesting} open problem is to determine the class of objective functions for which the presented techniques yield fast and stateless distributed algorithms, together with a unified convergence analysis. This problem is especially important in light of the fact that $\alpha$-fair objectives are not Lipschitz continuous, do not have a Lipschitz gradient, and their dual gradient's Lipschitz constant scales at least linearly with $n$ and $A_{\max}$. Therefore, the properties typically used in fast first-order methods are lacking \cite{nesterov2004introductory,  zhuOrecchia2014novel}. {Finally, for applications of $\alpha$-fair packing that do not require uncoordinated updates, it seems plausible that the dependence on $\varepsilon^{-1}$ in the convergence bound can be improved from $\varepsilon^{-5}$ to $\varepsilon^{-3}$ by relaxing the requirement for asynchronous updates, similarly as was done in \cite{d-allen2014using} over \cite{AwerbuchKhandekar2009}.}
 
\iffullpaper
\section*{Acknowledgements}

We thank Nikhil Devanur for pointing out the equivalence of the $\alpha$-fair packing for $\alpha = 1$ and the problem of finding an equilibrium allocation in Eisenberg-Gale markets with Leontief utilities.
\fi
 
\newpage

\bibliographystyle{abbrv}
{\small
\bibliography{references}
}
\newpage
\iffullpaper
\appendix

\section{Scaling Preserves Approximation}\label{appendix:scaling}

Let the $\alpha$-fair allocation problem be given in the form:
\begin{align*}
{(Q_\alpha)}\quad\textbf{max} \Big \{\sum_{j=1}^n w_jf_{\alpha}(x_j): Ax \leq b, 
x\geq 0\Big\},
\text{ where }
f_{\alpha}(x_j) = 
\begin{cases} \ln(x_j), & \mbox{if } \alpha=1 \\ \dfrac{x_j^{1-\alpha}}{1-\alpha}, & \mbox{if } \alpha\neq 1 \end{cases},
\end{align*}

\noindent $w$ is an $n-$length vector of positive weights, $x$ is the vector of variables, $A$ is an $n\times m$ constraint matrix, and $b$ is an $m-$length vector with positive entries. Denote $p_{\alpha}(x) = \sum_{j=1}^n w_jf_{\alpha}(x_j)$.

It is not hard to see that the assumption $b_i=1$ $\forall i$ is without loss of generality, since for $b_i \neq 1$ we can always divide both sides of the inequality by $b_i$ and obtain 1 on the right-hand side, since for (non-trivial) packing problems $b_i >0$. Therefore, we can assume that the input problem has constraints of the form $A\cdot x\leq \mathds{1}$, although it may not necessarily be the case that $A_{ij}\geq 1$ $\forall A_{ij}\neq 0$.

The remaining transformation that is performed on the input problem is:
\begin{equation*}
{\widehat{x}_j} = c\cdot x_j, \quad \widehat{A}_{ij}= A_{ij}/c.
\end{equation*}
where 
\begin{equation*}
c = \begin{cases}\min_{i, j: A_{ij\neq 0}}{A_{ij}}, & \mbox{if }  \min_{i, j: A_{ij\neq 0}}{A_{ij}}<1\\
1, & \mbox{otherwise} \end{cases}. %\label{eq:scaling-constant-c}
\end{equation*}

The problem $(Q_\alpha)$ after the scaling becomes:

\begin{minipage}{.4\linewidth}
\begin{align*}
\quad\textbf{max} \quad&\sum_{j=1}^n w_jf_{\alpha}({\widehat{x}_j})\cdot c^{1-\alpha}\\
\textbf{s.t.}\quad& \widehat{A}\widehat{x} \leq \mathds{1}\\
&\widehat{x}\geq 0
\end{align*}
\end{minipage}
$\Leftrightarrow$
\begin{minipage}{.4\linewidth}
\begin{align*}
{({P_\alpha})}\quad \textbf{max} \quad&\sum_{j=1}^n w_jf_{\alpha}({\widehat{x}_j})\\
\textbf{s.t.}\quad& \widehat{A}\widehat{x} \leq \mathds{1}\\
&\widehat{x}\geq 0,
\end{align*}
\end{minipage}\\
as $c^{1-\alpha}$ is a positive constant. Recall that \textsc{$\alpha$-FairPSolver} returns an approximate solution to $({P_\alpha})$, and observe that $x$ is feasible for $(Q_\alpha)$ if and only if $\widehat{x}$ is feasible for $(P_\alpha)$.

Choose the dual variables (Lagrange multipliers) for the original problem $(Q_\alpha)$ as:
\begin{equation}
y_i = c^{\alpha-1} C \cdot e^{\kappa(\sum_{i=1}^n A_{ij}x_j - 1)} = c^{\alpha-1} C \cdot e^{\kappa(\sum_{i=1}^n \widehat{A}_{ij}{\widehat{x}_j} - 1)}  = c^{\alpha-1} \widehat{y}_i,\label{eq:transformed-y}
\end{equation}
and notice that 
\begin{align}
{x_j}^{\alpha}\sum_{i=1}^{m}y_i A_{ij} = {\widehat{x}_j}^{\alpha} \cdot c^{-\alpha}\cdot\sum_{i=1}^{m} (c^{\alpha-1}\cdot\widehat{y}_i \cdot c\cdot \widehat{A}_{ij} )
= {{\widehat{x}_j}}^{\alpha}\sum_{i=1}^m \widehat{y}_i \widehat{A}_{ij}.\label{eq:condition-equality}
\end{align}
It is clear that $y_i$'s are feasible dual solutions, since the only requirement for the duals is non-negativity.

\subsection{Approximation for Proportional Fairness} 

Recall (from (\ref{eq:duality-gap-alpha})) that the duality gap for a given primal- and dual-feasible $x$ and $y$ is given as:
\begin{equation*}
G(x, y) = \sum_{j=1}^nw_j\ln(w_j) - \sum_{j=1}^n w_j \ln\left(x_j\sum_{i=1}^m y_i A_{ij}\right)+\sum_{i=1}^m y_i -1.
\end{equation*}
Since $\alpha=1$, we have that $\widehat{y}_i = y_i$ for all $i$, and using (\ref{eq:condition-equality}), it follows that
\begin{equation*}
G(\widehat{x}, \widehat{y}) = G(x, y).
\end{equation*}
Since we demonstrate an additive approximation for the proportional fairness via the duality gap: $p(\widehat{x}^*) - p(\widehat{x})\leq G(\widehat{x}, \widehat{y})$, the same additive approximation follows for the original (non-scaled) problem.

\subsection{Approximation for $\alpha$-Fairness and $\alpha\neq 1$}
For $\alpha\neq 1$, we show that the algorithm achieves a multiplicative approximation for the scaled problem. In particular, we show that after the algorithm converges we have that: $p_\alpha({\widehat{x}}^*) - p_{\alpha}(\widehat{x})\leq r_\alpha p_\alpha(\widehat{x})$, where ${\widehat{x}}^*$ is the optimal solution, $\widehat{x}$ is the solution returned by the algorithm, and $r_\alpha$ is a constant.

Observe that since $\widehat{x} = c \cdot x$, we have that $p_\alpha (\widehat{x}^*) = c^{1 - \alpha}p(x^*)$ and $p_\alpha(\widehat{x}) = c^{1 - \alpha} p_\alpha(x)$. Therefore:
\begin{align*}
p_\alpha({{x}}^*) - p_{\alpha}({x}) & = c^{\alpha - 1}(p_\alpha({\widehat{x}}^*) - p_{\alpha}(\widehat{x}))\\
&\leq c^{\alpha - 1} \cdot r_\alpha p_\alpha(\widehat{x})\\
& = r_\alpha p_\alpha({x}).
\end{align*}

%%%%%%%%%%%%%%%%%%%%%%%%%%%%%%%%%%%%%%%%%%%%%%%%%%%%%%%%%%%%%%%%%%
%\section{Solution Lower Bound}\label{appendix:lower-bound}

%Note that in Lemma \ref{lemma:lower-bound}, the lower bound is stated for the scaled version of the problem (problem $(P_\alpha)$). To state the lower bound for a non-scaled problem $(Q_\alpha)$, $\widehat{A}_{ij}$ from $(P_\alpha)$ would need to be replaced by $c \cdot A_{ij}/b_i$ in $(Q_\alpha)$, where $c = \min\{\min_{\{i,j:A_{ij}\neq 0\}}\frac{b_i}{A_{ij}}, 1\}$, while $\widehat{A}_{\max}$ from $(P_\alpha)$ would need to be replaced by $\max_{\{i,j,k, l: A_{ij}\neq 0 \wedge A_{kl}\neq0\}} \frac{b_i/A_{ij}}{b_k/A_{kl}}$. To simplify the notation, in Lemma \ref{lemma:lower-bound}, we omit the $\widehat{(.)}$ in the statement of the scaled problem $(P_\alpha)$.

%%%%%%%%%%%%%%%%%%%%%%%%%%%%%%%%%%%%%%%%%%%%%%%%%%%%%%%%%%%%%%%%%%
\section{Primal, Dual, and the Duality Gap}\label{appendix:primal-dual-duality-gap}

%%%%%%%%%%%%%%%%%%%%%%%%%%%%%%%%%%%%%%%%%%%%%%%%%%%%%%%%%%%%%%%%%%
\subsection{Proportionally Fair Resource Allocation}

In this section we consider $(w, 1)$-proportional resource allocation, often referred to as the weighted proportionally fair resource allocation. Recall that the primal is of the form:
\begin{align*}
{(P_1)}\quad\textbf{max} \quad&\sum_{j=1}^n w_j\ln(x_j)\\
\textbf{s.t.}\quad& Ax \leq \mathds{1},\\
&x\geq 0.
\end{align*}

The Lagrangian for this problem can be written as:
\begin{equation*}
L_1(x; y, z) = \sum_{j=1}^n w_j \ln(x_j) + \sum_{i=1}^m y_i\cdot\left(1 - \sum_{j=1}^n A_{ij} x_j -z_i\right),
\end{equation*}
where $y_1,..., y_m$ are Lagrange multipliers, and $z_1, ..., z_m$ are slack variables. The dual to this problem is:
\begin{align*}
{(D_1)}\quad\textbf{min} \quad&g(y)\\
\textbf{s.t.}\quad& y\geq 0,
\end{align*}
where $g(y) = \max_{x, z \geq 0} L(x; y, z)$. To maximize $L_1(x; y, z)$, we first differentiate with respect to $x_j$, $j\in\{1,...,n\}$:
\begin{equation*}
\frac{\partial L_1(x; y, z)}{\partial x_j}=\frac{w_j}{x_j}-\sum_{i=1}^{m}y_i A_{ij}=0,
\end{equation*}
which gives:
\begin{equation}\label{eq:primal-dual-relation}
x_j\cdot\sum_{i=1}^m y_i A_{ij} = w_j, \quad\forall j\in \{1,...,n\}.
\end{equation}
Plugging this back into the expression for $L_1(x; y, z)$, and noticing that, since $y_i, z_i\geq 0$ $ \forall i\in\{1,...,m\}$, $L_1(x; y, z)$ is maximized for $z_i=0$, we get that:
\begin{align*}
g_1(y) &= \sum_{j=1}^n w_j\ln\left(\frac{w_j}{\sum_{i=1}^my_i A_{ij}}\right)+\sum_{i=1}^m y_i-\sum_{i=1}^m y_i \sum_{j=1}^n \frac{A_{ij}w_j}{\sum_{k=1}^my_kA_{kj}}\\
&=\sum_{j=1}^nw_j\ln(w_j) - \sum_{j=1}^n w_j \ln\left(\sum_{i=1}^m y_i A_{ij}\right)+\sum_{i=1}^m y_i -\sum_{j=1}^n w_j \sum_{i=1}^m \frac{y_i A_{ij}}{\sum_{k=1}^my_kA_{kj}} \\
&=\sum_{j=1}^nw_j\ln(w_j) - \sum_{j=1}^n w_j \ln\left(\sum_{i=1}^m y_i A_{ij}\right)+\sum_{i=1}^m y_i -W,
\end{align*}
since $\sum_{i=1}^m \dfrac{y_i A_{ij}}{\sum_{k=1}^my_kA_{kj}} = 1$ $\forall j\in\{1,...,n\}$, and $\sum_{j=1}^n w_j = W$.

Let $p_1(x) = \sum_{j=1}^n w_j \ln(x_j)$ denote the primal objective. The duality gap for any pair of primal-feasible $x$ and dual-feasible (nonnegative) $y$ is given by:
\begin{align*}
G_1(x, y) &= g_1(y) - p_1(x)\notag\\
&= - \sum_{j=1}^n w_j \ln\left(\frac{x_j\sum_{i=1}^m y_i A_{ij}}{w_j}\right)+\sum_{i=1}^m y_i -W.
\end{align*}
Since the primal problem maximizes a concave function over a polytope, the strong duality holds \cite{boyd2009convex}, and therefore $G_1(x, y)\geq 0$ for any pair of primal-feasible $x$ and dual-feasible $y$, with equality if and only if $x$ and $y$ are primal- and dual- optimal, respectively.

%%%%%%%%%%%%%%%%%%%%%%%%%%%%%%%%%%%%%%%%%%%%%%%%%%%%%%%%%%%%%%%%%%
\subsection{$\alpha$-Fair Resource Allocation for $\alpha\neq 1$}

Recall that for $\alpha\neq 1$ the primal problem is:
\begin{align*}
{(P_\alpha)}\quad\textbf{max} \quad&\sum_{j=1}^n w_j\frac{x_j^{1-\alpha}}{1-\alpha}\equiv p_{\alpha}(x)\\
\textbf{s.t.}\quad& Ax \leq 1,\\
&x\geq 0.
\end{align*}

The Lagrangian for this problem can be written as:
\begin{equation*}
L_{\alpha}(x; y, z) = \sum_{j=1}^n w_j\frac{x_j^{1-\alpha}}{1-\alpha}+\sum_{i=1}^m y_i\left(1 - \sum_{j=1}^n A_{ij}x_j - z_i\right),
\end{equation*}
where $y_i$ and $z_i$, for $i\in \{1,...,m\}$, are Lagrangian multipliers and slack variables, respectively.

The dual to $(P_\alpha)$ can be written as:
\begin{align*}
{(D_\alpha)}\quad\textbf{min} \quad&g(y)\\
\textbf{s.t.}\quad& y\geq 0,
\end{align*}
where $g_{\alpha}(y) = \max_{x, z \geq 0} L_{\alpha}(x; y, z)$.

Since $L_{\alpha}(x; y, z)$ is differentiable with respect to $x_j$ for $j\in \{1,...,n\}$, it is maximized for:
\begin{align}
\frac{\partial L_{\alpha}(x; y, z)}{\partial x_j} &= \frac{w_j}{{x_j}^{\alpha}} - \sum_{i=1}^m y_i A_{ij} = 0\notag\\
\Rightarrow w_j &= {x_j}^{\alpha}\sum_{i=1}^m y_i A_{ij}. \label{eq:kkt-condition-alpha}
\end{align}
As $z_i\cdot y_i \geq 0$ $\forall i\in \{1,...,m\}$, we get that:
\begin{align*}
g_{\alpha}(y) &= \sum_{j=1}^n \frac{w_j}{1-\alpha}\left(\frac{w_j}{\sum_{i=1}^m y_i A_{ij}}\right)^{\frac{1-\alpha}{\alpha}}+\sum_{i=1}^m y_i - \sum_{i=1}^m y_i \sum_{j=1}^n A_{ij} \left(\frac{w_j}{\sum_{k=1}^{m}y_k A_{kj}}\right)^{1/\alpha}\\
&= \sum_{j=1}^n \frac{w_j}{1-\alpha}\left(\frac{w_j}{\sum_{i=1}^m y_i A_{ij}}\right)^{\frac{1-\alpha}{\alpha}}+\sum_{i=1}^m y_i - \sum_{j=1}^n w_j^{1/\alpha} \left(\sum_{k=1}^{m}y_k A_{kj}\right)^{-1/\alpha}\sum_{i=1}^m A_{ij}y_i\\
&= \sum_{j=1}^n \frac{w_j}{1-\alpha}\left(\frac{w_j}{\sum_{i=1}^m y_i A_{ij}}\right)^{\frac{1-\alpha}{\alpha}}+\sum_{i=1}^m y_i - \sum_{j=1}^n w_j^{1/\alpha}\left(\sum_{i=1}^m A_{ij}y_i\right)^{\frac{\alpha-1}{\alpha}}.
\end{align*}

Similarly as before, for primal-feasible $x$ and dual-feasible $y$, the duality gap is given as:
\begin{align}
G_{\alpha}(x, y) &= g_{\alpha}(y) - p_{\alpha}(x)\notag\\
&= \sum_{j=1}^n \frac{w_j}{1-\alpha}\left(\frac{w_j}{\sum_{i=1}^m y_i A_{ij}}\right)^{\frac{1-\alpha}{\alpha}}+\sum_{i=1}^m y_i - \sum_{j=1}^n w_j^{1/\alpha}\left(\sum_{i=1}^m A_{ij}y_i\right)^{\frac{\alpha-1}{\alpha}} - \sum_{j=1}^n w_j\frac{x_j^{1-\alpha}}{1-\alpha}\notag\\
&= \sum_{j=1}^n w_j\frac{x_j^{1-\alpha}}{1-\alpha}\left(\left(\frac{w_j}{{x_j}^{\alpha}\sum_{i=1}^m y_i A_{ij}}\right)^{\frac{1-\alpha}{\alpha}}-1\right) +\sum_{i=1}^m y_i - \sum_{j=1}^n w_j^{1/\alpha}\left(\sum_{i=1}^m A_{ij}y_i\right)^{\frac{\alpha-1}{\alpha}}.\notag
\end{align}
Observing that:
\begin{align*}
w_j^{1/\alpha}\left(\sum_{i=1}^m A_{ij}y_i\right)^{\frac{\alpha-1}{\alpha}} &= w_j\cdot {w_j}^{-\frac{\alpha-1}{\alpha}}\cdot x_j^{1-\alpha}\cdot x_j^{\alpha\frac{\alpha-1}{\alpha}}\cdot \left(\sum_{i=1}^m A_{ij}y_i\right)^{\frac{\alpha-1}{\alpha}}\\
&= w_j x_j^{1-\alpha}\cdot \left(\frac{{x_j}^{\alpha}\sum_{i=1}^m A_{ij}y_i}{w_j}\right)^{\frac{\alpha-1}{\alpha}},
\end{align*}
we finally get:
\begin{equation*}
G_{\alpha}(x, y)=\sum_{j=1}^n w_j\frac{x_j^{1-\alpha}}{1-\alpha}\left(\left(\frac{{x_j}^{\alpha}\sum_{i=1}^m y_i A_{ij}}{w_j}\right)^{\frac{\alpha-1}{\alpha}}-1\right) +\sum_{i=1}^m y_i - \sum_{j=1}^n  w_j x_j^{1-\alpha}\cdot \left(\frac{{x_j}^{\alpha}\sum_{i=1}^m A_{ij}y_i}{w_j}\right)^{\frac{\alpha-1}{\alpha}}.
\end{equation*}

\fi

\end{document}